\documentclass[11pt]{article}
\usepackage[margin=1in]{geometry}
\usepackage[utf8]{inputenc}
\usepackage[T1]{fontenc}
\usepackage[authoryear,round]{natbib}
\usepackage{amsmath,amssymb,amsfonts,amsbsy,amscd,amsthm}
\usepackage{graphicx}
\usepackage{xcolor}
\usepackage{float}
\usepackage{url}
\usepackage{booktabs}
\usepackage{microtype}
\usepackage{algorithm}
\usepackage{algpseudocode}
\usepackage{comment}
\usepackage{nicefrac}
\usepackage{subcaption}
\usepackage{tikz}
\usetikzlibrary{arrows.meta, positioning, calc}
\usepackage{bm}
\usepackage{authblk}
\usepackage{shortcuts}
\usepackage[colorlinks,citecolor=blue,linkcolor=blue,urlcolor=blue]{hyperref}
\usepackage{cleveref}
\newtheorem{theorem}{Theorem}
\newtheorem{lemma}{Lemma}
\newtheorem{setting}{Setting}

\theoremstyle{definition}
\newtheorem{assumption}{Assumption}
\newtheorem{definition}{Definition}
\newtheorem{remark}{Remark}

\newtheorem{proposition}{Proposition}

\crefname{assumption}{assumption}{assumptions}
\crefname{setting}{setting}{settings}
\newcommand{\mcF}{\mathcal{F}}
\newcommand{\Var}{\mathrm{Var}}
\AddToHook{env/remark/begin}{\pushQED{\qed}}
\AddToHook{env/remark/end}{\popQED}
\AddToHook{env/setting/begin}{\pushQED{\qed}}
\AddToHook{env/setting/end}{\popQED}
\definecolor{revisionblue}{RGB}{0,0,0}
\definecolor{maroontodo}{rgb}{0.5,0,0}
\newcommand{\az}[1]{}

\newenvironment{revision}{}{}
\providecommand{\mixprob}[3]{\left[\frac{#1-\alpha #2}{1-\alpha}\right]_{\varepsilon}^{#3}}

\title{Optimal Expert Annotation for Evaluating Sequential Decision Policies from Text Records}
\author[1]{Woojin Chae}
\author[2]{Ezinne Nwankwo}
\author[3]{Haitong Qin}
\author[1]{Angela Zhou}
\affil[1]{University of Southern California}
\affil[2]{Stanford University}
\affil[3]{University of Washington}
\date{}

\begin{document}
\maketitle

\begin{abstract}
Offline reinforcement learning and off-policy evaluation evaluate dynamic decision rules  based on historical data, without the ability to deploy online. 
In recent artificial intelligence (AI) applications, simple state and reward information is recorded as complex text or image, expert annotation may be available but at a higher cost. Recent AI and machine learning (ML) predictors can label data cheaply but with unknown bias. We show how a limited budget for ground-truth data-annotation can safely make use of plentiful ML predictions in a new estimator for doubly-robust off-policy evaluation with missing rewards. Our method directs where data should be collected by optimizing annotation probabilities for the best downstream estimation variance, which we implement in a feasible batch-adaptive approach. 
Our work is motivated by a collaboration with a homelessness services nonprofit that writes casenotes for individuals over time. Our method can be used to unlock trustworthy inference from casenote data and answer new inferential questions such as: how does expanding outreach effort over time affect progress towards a housing application and improvement in housing placement? In simulations and on two real datasets --- casenotes from the nonprofit and human-preference votes from LMArena --- we see reductions in RMSE of $37$--$65\%$ for housing placement and $21$--$67\%$ for progress towards a housing application at budgets of $50\%$ of full annotation and above, and by $55$--$62\%$ at every budget on LMArena under an equal-exposure policy.
\end{abstract}

\section{Introduction}

Offline reinforcement learning seeks to improve dynamic decisions from already collected data, such as improving dynamic clinical rules from electronic health data of historical clinical decisions and patient outcomes \citep{murphy2003optimal,levine2020offline}, or improving adaptive instruction in education \citep{mandel2014offline,doroudi2019s}. \textit{Online} reinforcement learning can evaluate the value of candidate new policies by deploying them in the world \citep{nahum2016just}. But often online experimentation is infeasible or costly due to safety or data-efficiency concerns, necessitating offline RL. Instead, offline RL relies on \textit{off-policy evaluation} (OPE) to assess the value of new \textit{target policies} from already collected data. Statistical estimation must not only predict immediate outcomes, such as rewards, but also longer-range cumulative rewards, in settings with unknown dynamics between states.

New advances in artificial intelligence (AI) enable working with richer reward or state information, such as text or image that are noisy measures of simpler rewards and states. For example, in healthcare and social services, clinical or case notes may document the underlying state of the individual. Our method, motivated by modern data annotation paradigms, assumes that data can be labeled for simpler ground-truth at a cost. (This contrasts to other text-as-data settings where measurements of the text are of primary interest). 
On the one hand, for text data, modern pre-trained artificial intelligence (AI) or machine learning (ML) predictors, such as LLM-as-a-judge, offer annotation at scale \citep{halterman2024codebook}, but may introduce unknown biases, such as length bias \citep{zheng2023judging}. On the other hand, ground-truth data annotation is the gold-standard for coding or labelled unstructured data, but is costly and time-consuming.

In our work, we introduce methods that can optimally combine low-bias but scarce ground-truth annotation and potentially biased but plentiful AI/ML imputation for \textit{off-policy evaluation} in offline reinforcement learning. We assume that only a small fraction of the data can be ground-truth labeled. Our method takes a missing-data perspective on datapoints that are ground-truthed (hence observed) vs. those that are not, hence imputed with predictors learned on a pilot, or as in modern practice, imputed with pre-trained AI/ML predictors. Consequently, we introduce a new estimator for \textit{off-policy evaluation with missing rewards} that combines these two sources of information with doubly-robust longitudinal causal inference. We extend to settings where ground-truth rewards are recorded in cumulative state history. Our key contribution is to \textit{optimize} which datapoints should be labeled in order to achieve the lowest asymptotic variance of downstream OPE, rather than the prediction error of reward models alone. We consider richer sequential annotation strategies that go beyond deciding whether to label entire trajectories or not. 
Our approach first labels a small batch of data in order to learn auxiliary functions appearing in the asymptotic variance that are needed to optimize it. We show theoretical guarantees and discuss approximations with simpler estimators that trade-off simplicity for the optimality of variance reduction. 

Our general methodological work is motivated by our collaboration with a homelessness services nonprofit that conducts street outreach for homelessness services, whose key data about the open-ended outreach process is recorded in written casenotes after every outreach encounter. The nonprofit's street outreach program builds longitudinal relationships with unhoused individuals, supporting them in pursuing their goals and completing a housing application. There are not many quantitative studies of the causal impacts of street outreach. Further, coding outcomes from text enables assessing the impact of nonprofit services on intermediate progress towards a housing application, such as connecting clients with other resources for material, health, mental and income support. Such intermediate outcomes are important benefits of street outreach but difficult to capture in administrative data. We develop a text annotation schema to do so.

Off-policy evaluation can provide quantitative measurements on how more or less intensive street outreach improves housing outcomes and intermediate progress outcomes, and can provide insights on where program expansion capacity may be the most effective. For example, we can ask: for a client who is currently contacted about two to three times a month (the middle outreach quintile in our cohort), what is the impact of increasing outreach to five or six times a month (the next quintile up)? Sequential decision frameworks also enable us to ask dynamic questions: are there any differences, in terms of long-term impacts, between increasing outreach intensity during early rapport-building stages vs. increasing outreach intensity during later stages of the housing application that can include appointment-making and attendance? 


In our application, we see that off-the-shelf language models incur bias relative to human expertise and systematically under-code progress, likely since the casenotes reflect organization-internal assessments that do not appear in training data (\Cref{tab:llm-judge}). Our method enables us to leverage scarce ground-truth to correct for such systematic biases, while otherwise safely leveraging predictions for variance reduction. Although we cannot ground-truth our entire dataset due to its sheer volume ($\sim$40{,}000 casenotes), we ground-truth the casenote with the maximum fine-tuned label in each client-stage, which determines the progress reward (\Cref{sec:experiments}). 
On the full two-year casenote cohort, full-data DRL estimates indicate that shifting each client's outreach intensity up one quintile raises the rate of any housing-placement improvement from $22.5\%$ to $27.9\%$ ($95\%$ CI $23.0$--$32.8\%$). It also raises progress towards a housing application by $0.30$ coding steps over the two stages (about $0.15$ per six-month period), from just below to just above discussing plans.
In terms of label efficiency, the optimized design matches the interval width of random annotation with $26$--$31\%$ fewer expert labels for housing placement (budgets $30$--$50\%$), $25$--$28\%$ fewer for progress ($30$--$60\%$), and $37$--$48\%$ fewer on LMArena under equal exposure ($30$--$50\%$). 
\begin{table}[t]
\centering\footnotesize
\caption{Share of 2019 casenotes by coded progress level. \emph{Human}: design-weighted estimate (\%, 95\% CI) from $677$ casenotes drawn by stratified random sampling. \emph{Base} and \emph{Fine-tuned}: Gemma 3 12B before and after fine-tuning on expert labels, evaluated on held-out data, as percentage-point differences from the human estimate (stratified bootstrap 95\% CIs).}\label{tab:llm-judge}
\setlength{\tabcolsep}{2.3pt}
\begin{tabular}{@{}lccc@{}}
\toprule
 & Human & Base & Fine-tuned \\
 & (\%) & \multicolumn{2}{c}{(percentage points, minus\ human)} \\
\midrule
No progress         & 44.0 (42.0, 45.7) & $+6.7$ ($+4.2$, $+9.3$) & $-0.9$ ($-2.8$, $+0.7$) \\
Conversation        & 43.1 (40.6, 45.7) & $-1.8$ ($-4.9$, $+1.4$) & $-1.5$ ($-4.3$, $+1.6$) \\
Plans \& paperwork  & 8.2 (6.6, 10.0)   & $-1.3$ ($-3.3$, $+0.6$) & $+2.4$ ($-0.1$, $+5.0$) \\
Milestones          & 4.6 (3.5, 6.0)    & $-3.6$ ($-5.0$, $-2.4$) & $0.0$ ($-2.2$, $+1.7$) \\
\bottomrule
\end{tabular}
\end{table}

\section{Problem Setup}\label{sec:problem-setup}

\paragraph{Full-data sequential target.}
We assume the data are generated by a full-information Markov decision process $\mathcal{M} = (\mathcal{S}, \mathcal{A}, \mathcal{R}, P, P_0, \gamma)$ with state space $\mathcal{S}$, action space $\mathcal{A}$, reward function $\mathcal{R}$, transition kernel $P$, initial state distribution $P_0$, and discount factor $\gamma$. For each unit, the (unobserved) full-data trajectory is $
(S_1,A_1,R_1,S_2,A_2,R_2,\ldots,S_T,A_T,R_T)$. Here \(S_t\) denotes the true state, \(A_t\) the action, and \(R_t\) the reward.

As in our motivating application, the true state is unobserved due to unobserved reward (and hence history), and we allow $S_{t+1} = f(S_{t+1}', R_t)$, for some Markov structural transition function $f(\cdot)$. There is an always-observed tabular state-space $S_{t}'$. 

Let \(\pi_b(a\mid s)\) denote the behavior policy, whose distribution governs the actions in the historical dataset,
and let \(\pi_e(a\mid s)\) denote a fixed evaluation policy. We evaluate policies that act through the tabular state,
\begin{equation}\label{eqn-tabular-pie}
\pi_e(A_t\mid S_t)=\pi_e(A_t\mid S_t'),\qquad t=1,\ldots,T,
\end{equation}
as does the policy in our application, which shifts each client's outreach intensity up one quintile. The target value is
\[ \textstyle
\Phi^{\pi_e}
=
\mathbb{E}_{\pi_e}\left[
\sum_{t=1}^T \gamma^{t-1}R_t
\right].
\]
Though we want to estimate the target policy value $\Phi^{\pi_e}$, the central challenge of off-policy evaluation is that we only have access to data collected under $\pi_b$.
We adopt the convention that $V_{T+1}^{\pi_e}\equiv 0$. 

\paragraph{Observed proxies for reward and state information}
The key challenge is that full-data trajectories $\{(S_t,A_t,R_t)_{t=1}^T\}$ are not observed. The always-observed data, without any annotation, are: 
$$ \mathcal{O}_0 = \{ (S_t',\tilde{S}_t, A_t, \tilde{R}_t)_{t=1}^T\} $$

An always-observed tabular state is $S_t'$; this is tabular data that does not require annotation. We have an \textit{always-observed} text or image measurement of the reward $\tilde{R}_t$.

The annotation setting presumes that a human annotator can easily extract or measure $R_t$ given $\tilde{R}_t$, though not all datapoints can have ground-truth revealed. For example, through qualitative coding of text \citep{egami2022make,halterman2024codebook}.

For example, in our motivating street outreach application, $S_t'$ reflects administrative data such as history of housing placement durations and historical counts of attempted/successful outreach. A casenote written after every outreach encounter is $\tilde{R}_t$, which can provide noisy measurements of housing placement information as $R_t$. We also developed a qualitative coding and annotation schema to code intermediate progress towards a housing application, such as concrete actions taken to complete appointments and other prerequisite steps. Ground-truth observations $R_t$ correspond to expert human annotation of casenotes, which are then used to update the state $S_{t+1} = (S_{t+1}', R_{t})$ with the placement or progress history.

\paragraph{Handling Markovian and non-Markovian problem setups.}
Our analysis extends both to non-Markovian and Markovian decision-process settings, although we use different estimators in either case. While details of estimation can differ, our general framework covers both settings, and otherwise we discuss the resulting differences explicitly when it is relevant. For notational ease, we primarily notate the MDP case. Throughout, we generally assume either a Markovian or non-Markovian decision process holds, with actions taken based on the simpler tabular state $S'$ alone. In particular, we suppose sequential ignorability holds with respect to the tabular state $S_t'$ and the reward $R_t$, rather than the high-dimensional measurements $\tilde R_t$ and $\tilde S_t$. Sequential ignorability effectively assumes that the underlying dynamics follow an MDP rather than a partially observed MDP.

\paragraph{Improving full-state MDP estimation by annotating $\tilde{S}_t, \tilde{R}_t$.} 
Crucially, in this setting, we view $(\tilde{S}_t, A_t,\tilde{R}_t)$ as \textit{measurements} of the underlying tabular MDP on $(S_t,A_t,R_t)$, but we are not interested in rewards or dynamics in the observed text spaces $\tilde{S}_t,\tilde{R}_t$ themselves. Therefore, our target estimation quantities are those of the unobserved tabular MDP.

For any given policy $\pi$, the state-action value ($Q_t^{\pi}$) and value functions ($V_t^{\pi}$) are
\begin{align*}
   \textstyle  Q_t^{\pi}(s,a) &= \mathbb{E}\left[ R_t+\gamma V_{t+1}^{\pi}(S_{t+1}) \mid S_t=s,A_t=a \right], 
    \\ 
    V_t^{\pi}(s) &
    \textstyle = \sum_a \pi_e(a\mid s)Q_t^{\pi}(s,a)
\end{align*}

Define the product of behavior policy weights and the marginal density ratio, where $p_{\pi_t}\left(s_t, a_t\right)$ denotes the marginal distribution of $s_t, a_t$ under $P_{\pi_t}$.
\[
\textstyle
\rho_{1:t}
=
\prod_{j=1}^t
\frac{\pi_e(A_j\mid S_j)}
{\pi_b(A_j\mid S_j)}, \qquad \mu_t = \frac{p_{\pi_e}(S_t,A_t)}{p_{\pi_b}(S_t,A_t)}.
\]


If we had completely observed rewards, the efficient estimator for off-policy evaluation is as follows \citep{kallus2019double}:
\[ \textstyle
\Gamma_T^{\pi_e}
=
V_1^{\pi_e}(S_1)
+
\sum_{t=1}^T
\gamma^{t-1} \mu_t
\{R_t+\gamma V_{t+1}^{\pi_e}(S_{t+1})-Q_t^{\pi_e}(S_t,A_t)\}
\]
However, only a fraction of the data can be annotated, so relying on full-data annotation alone is high-variance, although low-bias with valid inference. 
We take the annotation budget $B\in(0,1)$ to be a fraction of full annotation, so that at most $BT$ gold labels are expected per trajectory.



\begin{figure*}[t]
  \begin{minipage}[c]{0.25\textwidth}
    \centering
    \resizebox{\linewidth}{!}{
%
%
\begin{tikzpicture}[
  >=Stealth,
  every node/.style={font=\small},
  obs/.style    ={draw, circle, inner sep=1pt, minimum size=8mm, fill=white},
  hidden/.style ={draw, circle, inner sep=1pt, minimum size=8mm, fill=gray!30},
  edge/.style   ={->, semithick}
]
  \node[obs]    (Spt)  at (0,1.6)  {$S'_t$};
  \node[obs]    (At)   at (2,1.6)  {$A_t$};
  \node[obs]    (Sptp) at (4,1.6)  {$S'_{t+1}$};
  \node[obs]    (Atp)  at (6,1.6)  {$A_{t+1}$};

  \node[hidden] (Rt)   at (1,0)    {$R_t$};
  \node[hidden] (Rtp)  at (5,0)    {$R_{t+1}$};

  \node[obs]    (Rtt)  at (1,-1.5) {$\tilde R_t$};
  \node[obs]    (Rttp) at (5,-1.5) {$\tilde R_{t+1}$};

  \draw[edge] (Spt) -- (At);
  \draw[edge] (Spt) -- (Rt);
  \draw[edge] (At)  -- (Rt);
  \draw[edge] (Rt)  -- (Rtt);

  \draw[edge] (Sptp) -- (Atp);
  \draw[edge] (Sptp) -- (Rtp);
  \draw[edge] (Atp)  -- (Rtp);
  \draw[edge] (Rtp)  -- (Rttp);

  \draw[edge] (At)  -- (Sptp);
  \draw[edge] (Spt) to[bend left=30] (Sptp);
\end{tikzpicture}}
    \caption{Sequential annotation model for street outreach: reward annotation only.}
    \label{fig:bg-mdp}\label{fig:bg-mdp-reduced}
  \end{minipage}\hfill
  \begin{minipage}[c]{0.70\textwidth}
    \centering
\begingroup
\definecolor{annotationnew}{RGB}{218,119,24}
\newcommand{\annotationpanel}[1]{%
\begin{tikzpicture}[
  x=0.49cm,y=0.68cm,>=Stealth,
  every node/.style={font=\fontsize{7}{8}\selectfont},
  known/.style={draw,circle,minimum size=5mm,inner sep=0pt,fill=white},
  unavailable/.style={draw,circle,minimum size=5mm,inner sep=0pt,fill=gray!30},
  fresh/.style={draw=annotationnew!85!black,line width=1.3pt,circle,minimum size=5mm,inner sep=0pt,fill=annotationnew!27},
  edge/.style={->,draw=black,line width=0.4pt},
  sametime/.style={dotted,draw=black,line width=0.4pt}
]
  \ifcase#1
    \tikzset{firstgold/.style={unavailable},lastgold/.style={unavailable}}
  \or
    \tikzset{firstgold/.style={fresh},lastgold/.style={unavailable}}
  \or
    \tikzset{firstgold/.style={known},lastgold/.style={fresh}}
  \fi
  \node[known] (Sp1) at (0,2) {$S'_1$};
  \node[known] (A1) at (2.2,2) {$A_1$};
  \node[known] (S1) at (1,0.5) {$S_1$};
  \node[firstgold] (R1) at (3.2,0.5) {$R_1$};
  \node[known] (St1) at (0,-1) {$\tilde S_1$};
  \node[known] (Rt1) at (2,-1) {$\tilde R_1$};
  \node[known] (Sp2) at (5,2) {$S'_2$};
  \node[known] (A2) at (7.2,2) {$A_2$};
  \node[firstgold] (S2) at (6,0.5) {$S_2$};
  \node[lastgold] (R2) at (8.2,0.5) {$R_2$};
  \node[known] (St2) at (5,-1) {$\tilde S_2$};
  \node[known] (Rt2) at (7,-1) {$\tilde R_2$};
  \draw[edge] (Sp1) -- (A1);
  \draw[edge] (Sp1) -- (S1);
  \draw[edge] (Sp1) to[bend right=18] (St1);
  \draw[edge] (Sp1) -- (R1);
  \draw[edge] (A1) -- (R1);
  \draw[sametime] (S1) -- (R1);
  \draw[edge] (R1) -- (Rt1);
  \draw[edge] (R1) to[bend right=18] (S2);
  \draw[edge] (Rt1) -- (St2);
  \draw[edge] (Sp1) to[bend left=28] (Sp2);
  \draw[edge] (Sp2) -- (A2);
  \draw[edge] (Sp2) -- (S2);
  \draw[edge] (Sp2) to[bend right=18] (St2);
  \draw[edge] (Sp2) -- (R2);
  \draw[edge] (A2) -- (R2);
  \draw[sametime] (S2) -- (R2);
  \draw[edge] (R2) -- (Rt2);
\end{tikzpicture}%
}
\begin{subfigure}[t]{0.32\linewidth}
  \centering\resizebox{\linewidth}{!}{\annotationpanel{0}}
  \caption{$\mathcal F_0$: before annotation.}
\end{subfigure}\hfill
\begin{subfigure}[t]{0.32\linewidth}
  \centering\resizebox{\linewidth}{!}{\annotationpanel{1}}
  \caption{$\mathcal F_1$: reveal $R_1,S_2$.}
\end{subfigure}\hfill
\begin{subfigure}[t]{0.32\linewidth}
  \centering\resizebox{\linewidth}{!}{\annotationpanel{2}}
  \caption{$\mathcal F_2$: reveal $R_2$.}
\end{subfigure}
\par\smallskip
{\footnotesize
\tikz[baseline=-0.6ex]\node[draw,circle,minimum size=2.2mm,inner sep=0pt,fill=white]{}; Previously available\qquad
\tikz[baseline=-0.6ex]\node[draw=annotationnew!85!black,line width=1.1pt,circle,minimum size=2.2mm,inner sep=0pt,fill=annotationnew!27]{}; Newly annotated\qquad
\tikz[baseline=-0.6ex]\node[draw,circle,minimum size=2.2mm,inner sep=0pt,fill=gray!30]{}; Unrevealed}
\endgroup
    \caption{Forward-monotone annotation for $T=2$, with $S_1$ initially observed.
    The panels show successive information sets along a labeled trajectory:
    stage 1 reveals $(R_1,S_2)$, and stage 2 reveals $R_2$.
White nodes are observed, gray nodes are latent, and orange nodes are newly annotated.}
    \label{fig:forward-monotone-panels}
  \end{minipage}
\end{figure*}

\paragraph{Sequential annotation.}
As a result, in our work we 1) derive estimators that combine plentiful imputation with scarce ground-truth labels, and 2) optimize how data should receive ground-truth label annotations. We describe our estimators later in \Cref{sec-method}, and describe our annotation protocol here. 

Let $C_t\in\{0,1\}$ indicate whether stage $t$ is annotated. What an annotation reveals depends on the setting.
\begin{setting}[Reward-only annotation]\label{setting-reward-only}
The state is the tabular state, $S_t=S_t'$, and annotation reveals only $R_t$ from $\tilde R_t$ (\Cref{fig:bg-mdp-reduced}).
\end{setting}
\begin{setting}[Reward--state annotation]\label{setting-intermediate}
Annotation reveals $(R_t,S_{t+1})$: the annotated reward also augments the state history, entering $V$ and $Q$ but not $R$ (\Cref{fig:forward-monotone-panels}).
\end{setting}
\noindent Together with the always-observed data $\mathcal O_0$, the observed data are
\[
\mathcal O=\mathcal O_0\cup\{(C_t,\,C_tR_t,\,C_tS_{t+1})\}_{t=1}^T.
\]

\begin{definition}[Forward-monotone annotation]\label{def:forward-monotone-sampling}
Annotation proceeds along prefixes of the trajectory, $C_1\ge C_2\ge\cdots\ge C_T$:
stage $t$ reveals $(R_t,S_{t+1})$, and the final stage reveals $R_T$. The information
after $t$ stages is:
\[
\mcF_0=\sigma(\mathcal O_0),\qquad
\mcF_t=\sigma(\mcF_{t-1},R_t,S_{t+1}),\quad t=1,\ldots,T,
\]
\end{definition}
Prefix annotation is natural for multi-stage follow-up, such as tracking a client
through successive records (\Cref{fig:forward-monotone-panels}).

\section{Related Work}\label{sec:related-work}
There is a large literature on offline reinforcement learning (evaluation and optimization) \citep{jin2021pessimism,xie2021bellman}, including approaches that leverage importance sampling or introduce marginalized versions \citep{jl16,thomas2015high,pb2016,kallus2019double,liu2018breaking}. Off-policy evaluation is also closely related to longitudinal causal inference \citep{chernozhukov2022automatic,singh2025sequential,orellana2010dynamic}. 

Our general approach for obtaining trustworthy inference from LLM-imputed rewards centers on optimizing the selection of a validation set for off-policy evaluation. We do not impose distributional assumptions on exactly how the LLM predictions differ from the true underlying reward, especially since typical distributional conditions for non-standard measurement error \citep{schennach2016recent} are generally inapplicable to text or images, our motivating application. As such, our estimator takes the form of off-policy evaluation with missing rewards, wherein we choose annotation probabilities to optimize the asymptotic variance.

\paragraph{Missing ground-truth observations and semiparametric causal inference: prediction-powered inference and semi-supervised learning.}
In single time-step statistical inference, several works have fruitfully exploited the \textit{missing data} model of ground-truth data annotation \citep{egami2022make,wang2020methods,angelopoulos2023prediction}, sometimes called ``prediction-powered inference''. (See \citep{song2026demystifying,ji2026predictions} for a survey overview.) \citet{zrnic2024active} studies optimizing the ground-truth annotation probabilities for mean and $M$-estimation. \citet{kluger2026m} study $M$-estimation under multi-wave sampling, with annotated covariates, leveraging parametric efficiency adjustments. 
The key difference is that we apply analogous debiasing-data-annotation ideas to the sequential offline reinforcement learning setting, and we optimize the data-annotation probabilities therein. 
In the single time-step setting of causal inference, prior work \cite{nwankwobatch} has optimized the ground-truth annotation probabilities for average treatment effect (ATE) causal estimation, and characterized the closed-form solution. 

\paragraph{Missing data in longitudinal causal inference}
\citet{robins1995analysis} develop semiparametric missing data theory for multivariate regression.  Such monotone coarsening is the same missingness pattern as our forward monotone annotation. \citet{tsiatis2006semiparametric} develops monotone coarsening for general estimation. These prior papers do not consider the offline reinforcement learning setting, and also do not optimize the missingness/annotation probabilities. 

The closest work in the sequential setting in motivation, \citet{sonabend2023semi}, similarly combines small ``gold''-labeled longitudinal datasets with surrogates in semi-supervised learning, but in dynamic treatment regimes and hence with terminal rewards only, and with uniform allocation. 
\citet{mandyam2025perry} augments a small ``gold''-labeled sequential dataset with auxiliary samples, combining entire trajectories, but therefore also generating synthetic dynamic transitions and trajectories. 
In our work, we focus on optimal sequential annotation of rewards that enter history. 

\paragraph{Asymptotic-variance optimizing treatment allocation.}Adaptive annotation protocols build on ideas and methods from adaptive treatment allocation, with some crucial distinctions.
Our work finds reward annotation probabilities that optimize the asymptotic variance, whether via batch or full adaptivity, most closely tracking the general batch-adaptive protocol of \citet{hahn2011adaptive}. In this work, we study the more complex sequential annotation problem. \cite{li2024csbae} 
also considers a double machine learning version of \citet{hahn2011adaptive}.
Many other papers study adaptive treatment allocation, and the bandit, active-learning and best-arm identification literature is simply enormous \citep{gao2019batched,zhao2023adaptive,cook2024semiparametric,shiusing,zhao2024experimental,simchi2023multi,qin2024optimizing}. 

\paragraph{Active learning.} There is also an enormous and rich literature on active learning \citep{Settles2009ActiveLL,xia2025selection,jesson2021causal,sundin2019active}, but in general it \emph{optimizes for prediction error} which is not the off-policy evaluation error itself, and in sequential settings, is a proxy/surrogate loss.
But for our data-annotation in OPE motivation, in-sample regret is not meaningful for post-hoc data annotation of already-realized rewards, prediction error differs from the target policy value functional overall, and minimum-variance estimation differs from classification of the best arm in general.  Batch annotation is more relevant for querying human annotators, instead of full adaptivity. However, the asymptotic inferential guarantees of \citet{hahn2011adaptive} can indeed be strengthened by modern technical tools, which we leave for future work.

\paragraph{Experimental design for off-policy evaluation.} The offline reinforcement learning literature has introduced \textit{optimal experimental design} for off-policy evaluation by optimizing the behavior policy's actions for low-variance OPE afterwards\citep{li2023optimal,mukherjee2024saver,hanna2024data,mukherjee2024speed,liu2025doubly,liu2023efficient}. One particularly closely related work \citep{li2023optimal} optimizes the behavior policy for contrasting all-treat vs. all-control policies, resulting in a Neyman-like allocation for entire trajectories based on the initial-state-conditional variance.

\section{Method}\label{sec-method}

We introduce our estimator for off-policy evaluation with missing rewards and/or states, its asymptotic variance, and characterize the corresponding optimal annotation probabilities. Proofs of all results are in the Appendix.

\subsection{Estimation under sequential annotation protocols}
With $C_{1:t}=\prod_{i\le t} C_i$, define the annotation probabilities and their products under the forward monotone annotation as follows:
\begin{align*}
\lambda_t(\mcF_{t-1})
&=
P(C_t=1\mid C_{1:t-1}=1,\mcF_{t-1}),\\
\Lambda_{1:t}
&= \textstyle 
\prod_{j=1}^t \lambda_j(\mcF_{j-1}),
\qquad t=1,\ldots,T.
\end{align*}

\paragraph{Identification assumptions.}

We make the following assumptions for causal identification. 
Define the history prior to $A_t$ as $\mathcal{H}_{S_t}=(S_1,A_1,R_1,\ldots,S_{t-1},A_{t-1},R_{t-
  1},S_t)$. Throughout, we always assume \cref{asn-seq-igno-stilde}: sequential ignorability holds with respect to the tabular state $S'$ alone. 

\begin{assumption}[Sequential ignorability given $S_t'$]\label{asn-seq-igno-stilde}
We assume the underlying dynamics factorize as
\begin{align*}
&P_{\pi_b}\big(S_1,A_1,R_1,\ldots,S_T\big)
=\\
& \textstyle \qquad \qquad P_0(S_1)\prod_{t=1}^{T}\pi_b\big(A_t\mid S_t'\big)\, P\big(S_{t+1},R_t\mid \mathcal{H}_{S_t},A_t\big).
\end{align*}
In a Markov decision process (\Cref{setting-mdp}: rewards and transitions depend on the history only through the current state and action), $ P\big(S_{t+1},R_t\mid \mathcal{H}_{S_t},A_t\big)= P\big(S_{t+1},R_t\mid S_t,A_t\big)$.
\end{assumption}

\Cref{asn-seq-igno-stilde} constrains the behavior policy to act through the tabular state alone, $\pi_b(A_t\mid S_t)=\pi_b(A_t\mid S_t')$. Such a restriction enables identifiability throughout.\footnote{If instead the behavior policy also depended on unobserved rewards $R_t$, then our estimator would suffer an additional challenge of unobserved confounding, which introduces a different desiderata for data annotation to reduce bias and restore causal identifiability, not only reduce variance. We leave such extensions for future work.} 

\Cref{asn-seq-igno-stilde} is plausible when observed  actions do not differ based on the contextual information. For example, in our application, outreach frequency is typically on a schedule aimed at balancing geographic constraints with reaching caseload goals of meeting clients at least three times a month.

As a result, $\mu_t$ is
$\mcF_0$-measurable for all $t$ for a NMDP and for an MDP when only annotating rewards. (We discuss approximations for the MDP setting with joint reward and state annotation).

\begin{assumption}[Sequential overlap]\label{asn-seq_overlap}
The annotation probabilities and density ratios satisfy $c_\lambda\leq \lambda_t\leq 1$ and $0\leq\mu_t\leq C_\mu$ for some positive constant $c_\lambda >0$ and $C_\mu >0$ for all $t=1,\dots, T$, and the behavior policy satisfies $\pi_b(a\mid s)\geq \varepsilon_b>0$ for all $(s,a)$.
\end{assumption}


Action overlap can be replaced by weaker concentrability assumptions that are standard in offline reinforcement learning.
\begin{assumption}[Annotation ignorability] \label{asn-anno-ignorability}
The annotation design satisfies
\[
P(C_t=1\mid \mcF_T,C_{1:t-1}=1)
=
P(C_t=1\mid \mcF_{t-1},C_{1:t-1}=1).
\]
\end{assumption}
\begin{revision}
The stage-\(t\) annotation decision is made before the newly revealed stage-\(t\) information is observed. This assumption holds by design in our batch-adaptive data-annotation setting, where we design the annotation probabilities based on $S_t, A_t, C_{t'}, t'<t$ alone.\end{revision}

\paragraph{Sequential estimation combining annotation with scarce ground truth.}

If we had completely observed rewards, the efficient estimator for off-policy evaluation is as follows \citep{kallus2019double}:
\[ \textstyle
\Gamma_T^{\pi_e}
=
V_1^{\pi_e}(S_1)
+
\sum_{t=1}^T
\gamma^{t-1} \mu_t
\{R_t+\gamma V_{t+1}^{\pi_e}(S_{t+1})-Q_t^{\pi_e}(S_t,A_t)\}.
\]
The challenge is that the small budget results in a high-variance estimator. 

We build our estimator on projections of this full-data score onto progressive annotation filtrations, and therefore denote such $t-$annotated projections:
$$ m_t^{\pi_e}(\mcF_t) = \E[\Gamma_{T}^{\pi_e}\mid \mcF_t]. $$

With the $(\tilde{S}_t, \tilde{R}_t)$ proxies, the other extreme from scarce full-data estimation is  imputation: estimate $R_t$ from $\tilde{R}_t$ and estimate $S_t$ from $\tilde{S}_t$. The initial imputation $m_0^{\pi_e}(\mcF_0)$ conditions only on the initially observed data $\mathcal O$:
\begin{align*} \textstyle
&m_t^{\pi_e}(\mcF_0)\coloneqq   \E[ \Gamma_T^{\pi_e} \mid \mcF_0] 
=
\E[ V_1^{\pi_e}(S_1)\mid \mcF_0]\\
&\;
+
\sum_{t=1}^T
\gamma^{t-1} 
\E[ \mu_t
\{R_t+\gamma V_{t+1}^{\pi_e}(S_{t+1})-Q_t^{\pi_e}(S_t,A_t)\} \mid \mcF_0] .
\end{align*}
The challenge is that imputation is generally of unknown bias.

\begin{revision}Instead of being limited to two extremes of labeling a fraction $B$ of the data or complete imputation, we introduce a \textit{doubly-robust} estimator which blends these two extremes while maintaining validity and reducing variance. Our sequentially-annotated estimator begins with the full-imputation estimator $m_0^{\pi_e}(\mcF_0)$ and proceeds with sequential inverse-annotation-weighted increments that correct the incremental information we gain at every annotation timestep,  $m_{t}^{\pi_e}(\mcF_{t})-m_{t-1}^{\pi_e}(\mcF_{t-1})$. Our estimator is as follows: 
\begin{equation}\label{eqn-estimator}
  \tilde{\Gamma}_T^{\pi_e}(O;\bm\Lambda,\hat{\bm\eta}_m) =   \hat m_0^{\pi_e}
      + \sum_{t=1}^T
      \frac{C_{1:t}}{\Lambda_{1:t}} (\hat m_t^{\pi_e} - \hat m_{t-1}^{\pi_e}),
\end{equation}
where $\bm\Lambda=(\Lambda_{1:t})_{t=1}^{T}$ is the vector of annotation probabilities,
$\bm\eta_m=(\mu_t,Q_t,b_t,\tilde Q_t)_{t=1}^{T}$ is the vector of nuisance functions used to
construct the projections $m_t$, and $\hat{\bm\eta}_m=(\hat\mu_t,\hat Q_t,\hat b_t,\hat{\tilde Q}_t)_{t=1}^{T}$ is its estimated
counterparts.
Under annotation ignorability and using the true annotation probabilities, the projection increments are mean-zero and telescope, cancelling in expectation but providing variance reduction. \begin{equation}\label{eq:approximate-projection-telescope}
\E[\tilde\Gamma_T^{\pi_e}(O;\bm\Lambda,\hat{\bm\eta}_m)\mid\mcF_T]
=\hat m_0^{\pi_e}+\sum_{t=1}^T(\hat m_t^{\pi_e}-\hat m_{t-1}^{\pi_e})=\hat m_T^{\pi_e}.
\end{equation}

The estimator is simplest when only annotating rewards. In that case,
\begin{equation}
\tilde\Gamma_T^{\pi_e}=m_0^{\pi_e}(\mcF_0)+\sum_{t=1}^{T}\frac{C_{1:t}}{\Lambda_{1:t}}\gamma^{t-1}\mu_t\bigl(R_t-b_t(\mcF_{t-1})\bigr).
\end{equation}
We discuss the more complex joint state-reward annotation case shortly, including additional approximations. Consequently, this estimator, with oracle nuisance functions, is unbiased for our target policy value estimand. 
\begin{proposition}[Unbiased estimation]\label{prop:seq-annotation-unbiased} $\E[   \tilde{\Gamma}_{T}^{\pi_e}] = \Phi^{\pi_e}$.
\end{proposition}
Suppose the forward monotone annotation strategy of \Cref{def:forward-monotone-sampling}.
In the sequential setting, the martingale nature of the annotation filtrations admits the following variance decomposition.
\begin{proposition}[Variance decomposition]\label{prop-variance-decomposition}
\begin{equation}\op{Var}(\tilde\Gamma_T^{\pi_e})
= \op{Var}( m_0^{\pi_e}(\mcF_0)) +
\sum_{t=1}^T \E\left[ \frac{
\Var(m_t^{\pi_e}(\mcF_t)\mid \mcF_{t-1})}{\Lambda_{1:t}(\mcF_{t-1}) }
\right] 
\label{eqn-variance-decomposition}
\end{equation}
\end{proposition}
\paragraph{Closed-form simplifications of the estimator.}
We go into more detail about the details of the estimator, the component nuisance functions, how it may differ between NMDP and MDP settings, and our final simplified approximation form.

Our approach estimates components of the estimator $R_t, V_t^{\pi_e}, Q_t^{\pi_e}$ from (partially) annotated data. To distinguish trajectory time $t$ from annotation stage $j$, define
\begin{align*}
b_t(\mcF_j)&=\E[R_t\mid\mcF_j],
\\
\tilde V_t^{\pi_e}(\mcF_j)&=\E[V_t^{\pi_e}(S_t)\mid\mcF_j],\;\;
\tilde Q_t^{\pi_e}(\mcF_j)=\E[Q_t^{\pi_e}(S_t,A_t)\mid\mcF_j].
\end{align*}
Initial projections use $\mcF_0$; immediately before revelation, the corresponding projections are $b_t(\mcF_{t-1})$ for $R_t$ and $\tilde V_t^{\pi_e}(\mcF_{t-2}),\tilde Q_t^{\pi_e}(\mcF_{t-2})$ for $t\ge2$.

Our estimator is defined with slightly different $\mu_t$ functions depending on whether the decision process is Markovian (\Cref{setting-mdp}: $(R_t,S_{t+1})$ depends on $\mathcal H_{S_t}$ only through $(S_t,A_t)$) or non-Markovian (\Cref{setting-nmdp}: $(R_t,S_{t+1})$ may depend on the full history $\mathcal H_{S_t}$). Under \eqref{eqn-tabular-pie}, if we have a Markov decision process under reward-only annotation, then the cumulative importance weight $\rho_{1:t}(S_t,A_t)$ is $\mcF_0$-measurable, as is $\mu_t(S_t, A_t)$, and can therefore be learned via the standard recursion. 

If we have a Markov decision process (\Cref{setting-mdp}) with joint reward-state annotation (\Cref{setting-intermediate}), $\mu_t(S_t,A_t)$ satisfies an analogous recursion by \Cref{asn-seq-igno-stilde}
\begin{equation}
    \mu_t(s_t,a_t) 
    =\E\left[\mu_{t-1}^{\pi_e}\cdot \frac{\pi_e(a_t\mid s_t')}{\pi_b(a_t\mid s_t')} \bigg\vert s_t, a_t\right]. \label{eqn-densratio-recursion}
\end{equation}

Note that we overload notation so that $\mu_t$ is the corresponding inverse propensity term in either the non-Markovian or Markovian setting. In the non-Markovian setting, $\mu_t = \prod_{i=1}^{t}\frac{\pi_e(a_i\mid s_i')}{\pi_b(a_i\mid s_i')}$.

\end{revision}

\begin{revision}

\end{revision}

\paragraph{Further simplified estimator via local approximation.}
The final estimator that we use is a local approximation that
  retains corrections
  to the current reward and next-stage value functions,
  while omitting
  updates to predictions of more distant score
  components (the exact increment is \eqref{eqn-mdp-increment} in the Appendix).
  \begin{align*}
   & \tilde \Gamma_T^{\pi_e}
    =m_0^{\pi_e}(\mcF_0)+\\
    &\sum_{t=1}^{T}\gamma^{t-1}\frac{C_{1:t}}{\Lambda_{1:t}}
    \Bigl[\mu_t\bigl\{R_t-b_t(\mcF_{t-1})
      +\gamma\bigl(V_{t+1}^{\pi_e}-\tilde V_{t+1}^{\pi_e}(\mcF_{t-1})\bigr)\bigr\}
      \\
    &  \qquad \qquad \qquad\qquad
      -\gamma\mu_{t+1}
      \bigl(Q_{t+1}^{\pi_e}-\tilde Q_{t+1}^{\pi_e}(\mcF_{t-1})\bigr)\Bigr].
\end{align*}

\begin{remark}[Local approximation]
    These predictions are frozen from annotation updates earlier than $\mcF_{t-2}$ and have a different asymptotic variance, larger than that in \Cref{prop-variance-decomposition}. But they result in the same unbiased estimation (Appendix).

In the MDP setting with joint states and rewards, when $T>2$, a small correction to the estimator restores telescoping (Appendix, \cref{eqn-mdp-estimator}).
\end{remark}

\begin{remark}[Unbiased estimation and (suboptimal) variance reduction under misspecified nuisance estimation]
    Crucially, no consistency of $\hat m_t^{\pi_e}$, $t<T$, is required for unbiased estimation. When $\hat m_t^{\pi_e}$ are misspecified, such that they converge to a probability limit $m^\dagger_t \neq m_t^{\pi_e}$, we retain orthogonal estimation and unbiasedness, but the variance minimization can be suboptimal. See \Cref{si-local-approx} for more details and additional results. We therefore recommend the first batch uniformly sample independent trajectories to enable estimation of the full-state nuisances $\mu_t, \hat V^{\pi_e}, \hat Q^{\pi_e}$, to ensure the above error is $o_p(n^{-\frac 12})$ under typical assumptions. 

\end{remark}

\subsection{Optimized allocation probabilities} \label{sec:annotation_prob_derivation}
We optimize the asymptotic variance with respect to annotation probabilities. We state the optimization problem in terms of a generic \begin{align*}\sigma_t^2(\mcF_{t-1}) 
&\coloneqq 
  \Var\big(m_t^{\pi_e}(\mcF_t)\mid \mcF_{t-1}\big)\\
  &=\E\big[ (m_t^{\pi_e}(\mcF_t)- m_{t-1}^{\pi_e}(\mcF_{t-1}))^2 \mid \mcF_{t-1}\big]
\end{align*}
since $\{(m_t^{\pi_e}, \mcF_t), t=1,\dots, T\}$ is a martingale.

  Next we discuss the annotation budget. Since the probability of annotating stage $t$ is $P(C_{1:t}=1\mid \mcF_{t-1})=\Lambda_{1:t}(\mcF_{t-1})$, the expected number of gold labels is $\sum_{t=1}^{T}\E[\Lambda_{1:t}(\mcF_{t-1})]$. Therefore, the optimal annotation problem is:
\begin{align*} \label{eqn-opt_prob}
&\min_{\{\Lambda_{1:t}\}_{t=1}^{T}\in[0,1]^{T}}
\sum_{t=1}^{T}\,\E\left[\frac{\sigma_t^2(\mcF_{t-1})}{\Lambda_{1:t}(\mcF_{t-1})}\right] \tag{OPT}\\
&\text{ s.t. }\sum_{t=1}^{T}\E\big[\Lambda_{1:t}(\mcF_{t-1})\big]\le BT;\ \Lambda_{1:t+1}(\mcF_{t}) \le \Lambda_{1:t}(\mcF_{t-1}),\ \forall t\nonumber 
\end{align*}
where $B$ is the annotation budget. The first constraint enforces the budget, and the second enforces forward monotonicity of the resulting annotation policy.

For different settings such as non-Markovian or Markovian decision processes, reward or joint reward-and-state annotation, or different approximate projections, different increments of conditional variance enter the above additively. 
\begin{assumption}[Reward-only prediction stability]\label{asn-reward-only-stability}
Under reward-only annotation, revealing $R_t$ does not change the prediction of later rewards given the always-observed data: $\E[R_j\mid\mcF_t]=\E[R_j\mid\mcF_{t-1}]$ for $j>t$.
\end{assumption}
\noindent It holds exactly in a Markov decision process (\Cref{setting-mdp}) where rewards do not enter the state and each silver label depends only on its own stage. It affects only the optimality of the allocation, not the validity of the estimator.
In the reward-only annotation setting (\Cref{setting-reward-only}), under \Cref{asn-reward-only-stability}, the conditional variance is $$\sigma_t^2(\mcF_{t-1}) 
=\gamma^{2(t-1)}\mu_t^2\Var(R_t\mid\mcF_{t-1}).$$
Note that the conditional variance $\Var(R_t\mid\mcF_{t-1})$ is the mean-squared error of the imputation $b_t$. In the joint reward-state annotation setting, the conditional variance of the local-approximation increment is $\sigma_t^2(\mcF_{t-1})=$
\[
\gamma^{2(t-1)}\Var\Big[\mu_t(R_t+\gamma V_{t+1}^{\pi_e})-\gamma\mathbf{1}\{t<T\}\,\mu_{t+1} Q_{t+1}^{\pi_e} \,\Big|\, \mcF_{t-1}\Big] .
\]

Since both $\sigma_t^2$ and $\Lambda_{1:t}$ are $\mcF_{t-1}$-measurable and $\sigma_t^2\ge0$, the optimization problem \eqref{eqn-opt_prob} is convex. The optimal solution introduced in the following theorem has an intuitive structure: the annotation probability at each stage is proportional to the conditional \emph{standard deviation} of that stage annotation's contribution to the score, normalized so that it binds the budget constraint. 


\paragraph{Optimal annotation probabilities: special case of $T=2$.}
For simplicity, first let us again specialize to the two-stage case, $T=2$.


\begin{theorem}[Optimal allocation probabilities, $T=2$.]\label{thm-twostage-optimal-allocation-T=2}
    Define the normalizing constant $$\beta =\frac{1}{4B^2} \left(\E\left[\sqrt{\sigma_1^2(\mcF_0)} + \sqrt{\sigma_2^2(\mcF_1)} \right]\right)^2.$$
    Provided $\sigma_2(\mcF_1)\le\sigma_1(\mcF_0)\le\sqrt\beta$ almost surely (so that neither forward monotonicity nor $\lambda_1^*\le1$ binds), the optimal annotation probabilities are:
    \begin{align*}
        \lambda_1^*(\mcF_0) &= \sqrt{\frac{  \sigma_1^2(\mcF_0)} {{\beta}}},\\
        \Lambda_{1:2}^*(\mcF_1) &=
       \sqrt{ \frac{\sigma_2^2(\mcF_1)}{{\beta}}}.
    \end{align*}
    This closed form is the interior solution. When the
    forward monotonicity or upper-bound constraint binds, $\lambda_1^*(\mcF_0)=\min\{1,\sqrt{Z(\mcF_0)/\beta}\}$ and $\Lambda_{1:2}^*(\mcF_1)=\min\{\lambda_1^*(\mcF_0),\sqrt{\sigma_2^2(\mcF_1)/\beta}\}$, where $Z$ solves a
    scalar fixed-point equation and $\beta$ is chosen to meet the budget, see Appendix~\ref{pf:thm-twostage-optimal-allocation-T=2}.
\end{theorem}

The simpler interior solution has the familiar square-root form of a stagewise Neyman allocation \citep{neyman1992two,zrnic2024active,hahn2011adaptive}, applied to the conditional variance of each sequential increment, due to the martingale structure and the corresponding additively separable variance. 

 However, in our problem, the $\sigma_t^2(\mcF_{t-1})$ includes terms beyond prediction variance of rewards alone. Annotations are optimally spent where the proxy and earlier labels leave the reward most uncertain, \textit{additionally weighted} by the density ratio $\mu_t$ toward states the evaluation policy visits more often than the behavior policy. In the joint setting, annotation also targets transitions whose next state is hard to predict.

When the forward monotonicity constraint binds, stage 1 must also account for the stage-2 constraint on annotation probabilities and the solution departs from this form.

\subsection{Batch-adaptive experiment scheme}

\Cref{thm-twostage-optimal-allocation-T=2} gives the annotation probability that makes the
resulting estimator achieve the lowest asymptotic variance. However, we cannot compute it in advance,
since it depends on data that are unobserved prior to the experiment. It must instead be \emph{learned}
as the experiment proceeds. This is why we adopt a batch-adaptive scheme, following the convergent
split batch-adaptive experiment (CSBAE) framework of \citet{li2024csbae}, that learns the optimal
probability over successive batches to approximate it at the population level. The data are
partitioned along two dimensions: $M$ \emph{batches}, which sequence the adaptive updating of the
annotation probability, and $K$ \emph{folds}, which preserve independence between a nuisance
estimate and the data it is applied to at final inference. (See \Cref{fig:csbae} for a schematic).

\paragraph{Adaptive protocol.}
The first batch (pilot batch) annotates entire trajectories with a prespecified probability $p$, that is,
\[
    C_{i;1:T}\sim\mathrm{Ber}(p),
\]
where $C_{i;1:T}=1$ indicates that the full trajectory of unit $i$ is annotated. This corresponds to the forward-monotone design $\hat\Lambda_{1,1:t}^{\ast(k)}=p$ for every fold $k$ and stage $t$.

Starting from the second batch, annotation at each stage $t$ instead follows the probability
obtained by solving the optimization problem of \cref{sec:annotation_prob_derivation}. This
probability is updated after each batch using only the data collected in preceding batches within
the same fold. The batch-$b$ probabilities
$\{\hat{\bm\Lambda}_b^{\ast(k)}\}_{b=1}^{M}$ are chosen so that the pooled design
$\sum_b(n_b/n)\hat\Lambda_{b,1:t}^{\ast(k)}$ attains the empirical budget-optimal target.

We assume that batch size $n_b$ scales proportional to data size $n$ as follows:
\begin{assumption}[Asymptotic batch proportion] 
    \[
        \lim_{n\to\infty}\frac{n_b}{n}= \kappa_b \in (0, 1), \quad \forall b=1,\dots, M.
    \]
\end{assumption}
\paragraph{Cross-fitting for final inference \citep{li2024csbae}.}
After the experiment concludes, we pool the data across batches and re-estimate the nuisance
functions required for the OPE estimator by $K$-fold cross-fitting: for each fold, nuisances are
estimated on the remaining folds and evaluated on the held-out fold, preserving the independence
between nuisance estimation and score evaluation required for the asymptotic analysis below.
Unlike the in-experiment design probabilities $\{\hat{\bm\Lambda}_b^{\ast(k)}\}_{b=1}^{M}$, the annotation
probability entering the final estimator, $\hat{\bm\Lambda}^{(-k)}$, is estimated directly from the
realized annotation outcomes pooled across all batches, rather than from solving the optimization
problem. Crucially, $\hat{\bm\Lambda}^{(-k)}$ and $\{\hat{\bm\Lambda}_b^{\ast(k)}\}_{b=1}^{M}$ are
independent: the former is estimated using data outside fold $k$, whereas the latter is estimated using data within fold $k$. Thus, their underlying data dependencies do not overlap. This independence is
the key property that our asymptotic analysis relies on. The resulting feasible estimator is
\begin{equation}\label{eqn:pooled_adaptive}
   \hat\psi_{\mathrm{ad}} =\frac{1}{n}\sum_{k=1}^{K}\sum_{(b,i)\in\mathcal I_k}
    \tilde\Gamma_T^{\pi_e}(O_{b,i};\hat{\bm\Lambda}^{(-k)},\hat{\bm\eta}^{(-k)}_{m}),
\end{equation}
where $\tilde\Gamma_T^{\pi_e}$ is the estimator in \eqref{eqn-estimator}; $\mathcal I_k$ is the set of (batch, unit) pairs in fold $k$, with $O_{b,i}$ the realized outcome of unit $i$ in batch $b$;
$\hat{\bm\Lambda}^{(-k)}=(\hat\Lambda_{1:t}^{(-k)})_{t=1}^{T}$ is the vector of annotation probabilities estimated using data outside fold $k$; and $\hat{\bm\eta}_m^{(-k)}=(\hat\mu_t^{(-k)},\hat Q_t^{(-k)},\hat b_t^{(-k)},\hat{\tilde Q}_t^{(-k)})_{t=1}^{T}$ is the vector of nuisance estimates, also using data outside fold $k$, from which the projection estimates $\hat m_t^{(-k)}$ are constructed.

\paragraph{Nuisance estimation.}
Although the nuisance estimation for the in-experiment design and for the post-experiment inference differ in the data used for fitting (the earlier batches of the same fold during the experiment, and the folds other than $k$ in the final inference), both rely on standard supervised learning. The most involved nuisances are the value functions $Q_t^{\pi_e}$ and $\tilde Q_t^{\pi_e}$, which we estimate jointly by a backward recursion from $t=T$, with the convention $V_{T+1}^{\pi_e}=\tilde V_{T+1}^{\pi_e}=Q_{T+1}^{\pi_e}=\tilde Q_{T+1}^{\pi_e}\equiv 0$. At each $t$, we first learn $Q_t^{\pi_e}$ by debiased fitted Q-evaluation:
\begin{align*}
Q_t^{\pi_e}&\in\argmin_{q_t}\E\Bigl[\tfrac{C_{1:t-1}}{\Lambda_{1:t-1}}\bigl(\tilde Y_t-q_t(S_t,A_t)\bigr)^2\Bigr],\quad\text{where}\\
\tilde Y_t&=\Bigl(1-\tfrac{C_t}{\lambda_t}\Bigr)\bigl(b_t(\mcF_{t-1})+\gamma\tilde V_{t+1}^{\pi_e}(\mcF_{t-1})\bigr)
+\tfrac{C_t}{\lambda_t}\bigl(R_t+\gamma V_{t+1}^{\pi_e}\bigr).
\end{align*}
The pseudo-outcome $\tilde Y_t$ equals the annotated target when $C_t=1$ and falls back on the projection when $C_t=0$. We then obtain $\tilde Q_t^{\pi_e}$ by regressing the fitted $Q_t^{\pi_e}$ on $\mcF_{t-2}$:
\[
\tilde Q_t^{\pi_e}(\mcF_{t-2})=\E[Q_t^{\pi_e}(S_t,A_t)\mid\mcF_{t-2}],
\]
which in turn yields $\tilde V_t^{\pi_e}$ for the next backward step.
Estimating $\mu_t$ follows from the recursion in \Cref{eqn-densratio-recursion}; see \citep{kallus2020double} for more details.

\subsection{Feasible batch-adaptive annotation}
\label{subsec:batch-adaptive-annotation}

\providecommand{\mixprob}[3]{\left[\frac{#1-\kappa_1 #2}{1-\kappa_1}\right]_{\varepsilon}^{#3}}

\Cref{alg:batch-adaptive-seq-annotation} gives a feasible two-batch implementation of
\Cref{thm-twostage-optimal-allocation-T=2}. The pilot batch annotates each full trajectory with probability $p$,
and its data are used to estimate the optimal annotation score. The adaptive batch then applies a
mixture correction to this plug-in estimate, so that the pooled design attains the target
allocation. Specifically, for $u\in(0,1]$, let
\[
g(x;u)\;:=\;\min\Big\{u,\;\max\Big\{\varepsilon,\;\tfrac{x-\kappa_1 p}{1-\kappa_1}\Big\}\Big\}
\]
denote the adaptive-batch probability $x_{\rm ad}$ that solves
$\kappa_1 p+(1-\kappa_1)x_{\rm ad}=x$, clipped to $[\varepsilon,u]$. Within fold $k$ of
\Cref{alg:batch-adaptive-seq-annotation}, we suppress the superscript $(k)$ on design-time
quantities. Let $\Lambda_{b,1:t}^\ast$ denote the limit of the batch-$b$ design and $\Lambda_{1:t}^\dagger:=\sum_b\kappa_b\Lambda_{b,1:t}^\ast$ the pooled design actually implemented. Clipping makes $\Lambda^\dagger_{1:t}\neq\Lambda^\ast_{1:t}$ where $\Lambda^\ast_{1:t}$ lies outside $[\kappa_1p+(1-\kappa_1)\varepsilon,\,\kappa_1p+(1-\kappa_1)u]$; otherwise $\Lambda^\dagger_{1:t}=\Lambda^\ast_{1:t}$.

\begin{algorithm}[t]
\caption{Two-batch adaptive prefix annotation ($T=2$)}
\label{alg:batch-adaptive-seq-annotation}
\small
\begin{algorithmic}[1]
\Require $\mathcal D=\mathcal{O}_0,\ \pi_e,\ \hat\pi_b,\ B,\ \kappa_1,\ p,\ K,\ \varepsilon$
\Ensure $\hat\psi^{\pi_e}$

\Statex \textbf{Setup}
\State Split units into folds $\{D^{(k)}\}_{k=1}^{K}$
\State Within each fold, split into $\mathcal I_{\rm pil}$, $\mathcal I_{\rm ad}$ with $|\mathcal I_{\rm pil}|=\kappa_1 n$

\Statex \textbf{Pilot batch}
\For{$i\in\mathcal I_{\rm pil}$}
  \State Draw $C_{i;1:2}\sim\mathrm{Bern}(p)$; if $C_{i;1:2}=1$, annotate both stages
  \State $\bm\Lambda_i\gets(p,\,p)$
\EndFor

\Statex \textbf{Adaptive batch}
\For{$k=1,\dots,K$}
  \State Fit $(\hat\lambda^{\ast}_{1},\hat\Lambda^{\ast}_{1:2})$ via \Cref{thm-twostage-optimal-allocation-T=2} on $D^{(k)}\cap\mathcal I_{\rm pil}$
  \For{$i\in D^{(k)}\cap\mathcal I_{\rm ad}$}
    \State $\lambda_{i,1}\gets g\big(\hat\lambda^{\ast}_{1}(\mcF_{i,0});\,1\big)$
    \State Draw $C_{i,1}\sim\mathrm{Bern}(\lambda_{i,1})$
    \If{$C_{i,1}=1$}
      \State Reveal $(R_{i,1},S_{i,2})$
      \State $\Lambda_{i,1:2}\gets g\big(\hat\Lambda^{\ast}_{1:2}(\mcF_{i,1});\,\lambda_{i,1}\big)$
      \State Draw $C_{i,2}\sim\mathrm{Bern}(\Lambda_{i,1:2}/\lambda_{i,1})$
    \Else
      \State $C_{i,2}\gets 0$
    \EndIf
    \State $\bm\Lambda_i\gets(\lambda_{i,1},\,\Lambda_{i,1:2})$
  \EndFor
\EndFor

\Statex \textbf{Inference}
\For{$k=1,\dots,K$}
  \State Fit $\hat{\bm\Lambda}^{(-k)}$, $\hat{\bm\eta}^{(-k)}_m$ on $\{(O_i,\bm\Lambda_i)\}_{i\notin D^{(k)}}$
\EndFor
\State \Return $\hat\psi^{\pi_e}$ in \eqref{eqn:pooled_adaptive}
\end{algorithmic}
\end{algorithm}

\section{Analysis}\label{sec:analysis}

Because the annotation probabilities of later batches are learned from earlier ones, the realized outcomes are dependent across batches, so the standard i.i.d.\ central limit theorem does not apply to $\hat\psi_{\rm ad}$ directly. Nevertheless, we show that, under suitable convergence rates of the nuisance estimators, $\hat\psi_{\rm ad}$ is asymptotically normal, with the variance of \cref{prop-variance-decomposition} evaluated at the implemented design $\bm\Lambda^\dagger$. When clipping does not bind, $\bm\Lambda^\dagger=\bm\Lambda^\ast$, and the adaptive design incurs no asymptotic loss in efficiency relative to an oracle that knows the optimal design. Throughout, for a nuisance function $f$ with cross-fitted estimate $\hat f^{(-k)}$ (fit without fold $k$), we write $\Delta f^{(-k)}=\hat f^{(-k)}-f$ for its estimation error, and $\|\cdot\|_2$ for the $L_2(P)$ norm under the behavior distribution.



\subsection{Asymptotic normality under the batch-adaptive experiment}

In the batch-adaptive experiment, the annotation probability used in batch $b$, fold $k$, denoted $\hat\Lambda_{b,1:t}^{\ast(k)}$, is computed from the earlier batches $\mathcal D_{1:b-1}^{(k)}$ of
the same fold. The final estimator, in contrast, uses the cross-fitted estimate $\hat\Lambda_{1:t}^{(-k)}$ obtained from the pooled data. Hence $\hat\psi_{\rm ad}$ involves two sources of annotation-probability error: the in-experiment design $\hat\Lambda_{b,1:t}^{\ast(k)}$ and the post-experiment estimate $\hat\Lambda_{1:t}^{(-k)}$. \Cref{asn-rate_adaptive} below controls the two separately. We state it first and then present the main result.



\begin{assumption}[Product error rate]
\label{asn-rate_adaptive}
    Define the sum of nuisance errors regarding projection estimates $\hat m_t^{(-k)}$: 
    $$ 
  \mathcal{E}_{m_t}^{(-k)} \coloneqq  \|\Delta b_t^{(-k)}\|_2+\|\Delta Q_{t+1}^{(-k)}\|_2+\|\Delta\mu_t^{(-k)}\|_2+\|\Delta\mu_{t+1}^{(-k)}\|_2
    $$
    With the convention that $\Delta f_{T+1}:=0$, the following product error rates hold for all $b\in [M]$, $k\in[K]$, and $t\in[T]$:
\begin{align*}
    \textup{(R1)}\qquad &\|\Delta\mu_t^{(-k)}\|_2\Big(\|\Delta Q_t^{(-k)}\|_2+\|\Delta Q_{t+1}^{(-k)}\|_2\Big)=o_p(n^{-1/2}),\\
    \textup{(R2a)}\qquad &\sum_{j\leq t}\|\Delta\hat\lambda_{b,j}^{(k)}\|_2 \cdot \mathcal{E}_{m_t}^{(-k)}=o_p(n^{-1/2}),\\
    \textup{(R2b)}\qquad &\sum_{j\leq t}\|\Delta\hat\lambda_{j}^{(-k)}\|_2 \cdot \mathcal{E}_{m_t}^{(-k)}=o_p(n^{-1/2}).
\end{align*}
We refer to (R1) as the double reinforcement learning (DRL) product error and to (R2a)--(R2b) as the annotation product errors.
\end{assumption}
\begin{remark}
The annotation product errors (R2a) and (R2b) correspond to the two sources of annotation error:
(R2a) to the probability used during the experiment in batch $b$ and fold $k$, and (R2b) to the
probability re-estimated on the pooled data by cross-fitting after the experiment.
\end{remark}

Under product rate assumption along with some regular assumptions (see the Appendix for more detail), we show that the batch-adaptive estimator with exact projections $m_t$ is asymptotically normal, with the minimal variance at $\bm\Lambda^\ast$ when clipping does not bind. This holds for reward-only annotation, where the projections are exact; the frozen-projection estimator is covered by \Cref{prop:clt_adaptive_endpoint} in the Appendix.

\begin{theorem}\label{thm:clt_adaptive}
    Suppose \cref{asn-seq-igno-stilde,asn-anno-ignorability,asn-seq_overlap,asn-rate_adaptive}, bounded rewards, $0\le R_t\le R_{\max}$ (\Cref{asn-bdd-reward}), and uniformly bounded fitted nuisances with fitted annotation probabilities bounded away from zero (\Cref{asn-bdd_estimators_adaptive}) hold, together with $\|\hat{\bm\Lambda}_b^{(k)}-\bm\Lambda_b^\ast\|_{2}=o_p(1)$, $\|\hat{\bm\Lambda}^{(-k)}-\bm\Lambda^\dagger\|_{2}=o_p(1)$, and $\|\hat{\bm\eta}_m^{(-k)}-\bm\eta_m\|_{2}=o_p(1)$ for all $b\in[M]$ and $k\in[K]$.
    Then,
    \begin{align*}
       & \sqrt{n}\bigl(\hat\psi_{\mathrm{ad}}-\Phi^{\pi_e}\bigr)
        \Longrightarrow
        \mathcal N(0,\sigma_\ast^2),\\
        & \text{where} \qquad \sigma_{\ast}^{2} = \Var(m_0) + \sum_{t=1}^{T}
    \E\left[ \frac{\Var(m_t\mid\mcF_{t-1})}{\Lambda_{1:t}^\dagger} \right] 
    \end{align*}
\end{theorem}

Finally, the asymptotic inference results also extend to our family of simplified estimators, including for the more complex MDP setting with reward and state annotation and the corrected estimator of \Cref{eqn-mdp-estimator} (Appendix, \Cref{prop:clt_adaptive_endpoint}). This comes at small cost in terms of suboptimal variance reduction, although we retain validity.

\section{Experiments}\label{sec:experiments}

\subsection{Simulated data}
First, we illustrate our methods in a simulation with known ground truth.
For the data-generating process, we summarize here and refer to \Cref{apx-experiments} for the full details. Each unit has correlated (5-dimensional) Gaussian states across stages, binary
  treatments assigned by truncated logistic behavior policies with overlap bounded below by 0.05, and additive
  rewards depending on state, treatment, and a second-stage treatment interaction. The key design feature is
  heteroskedasticity targeted to the harder-to-learn arm: the less common control arm has much larger conditional
  reward variance at both stages, creating a setting where adaptive reward annotation can prioritize observations
  with high variance and high overlap importance.

\paragraph{Results.}

\begin{figure}
    \centering
    \includegraphics[width=0.49\columnwidth]{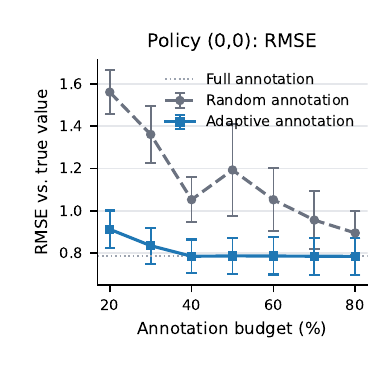}\includegraphics[width=0.49\columnwidth]{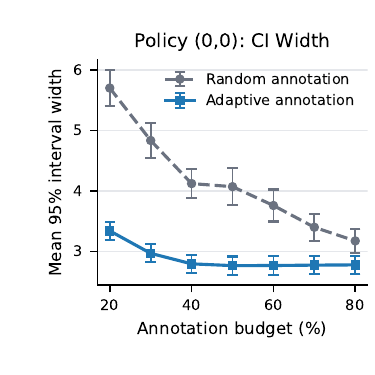}
    \caption{\textcolor{revisionblue}{Simulated data results for the harder target policy $(A_1=0,A_2=0)$, DRL estimator, $n=1000$, $100$ replications. Left: RMSE against the true policy value; the dotted line is the RMSE under full annotation. Right: mean $95\%$ interval width. Coverage and the other three policies are in \Cref{fig:synth-all-regimes}.}}\label{fig:simulated}
\end{figure}

\begin{figure*}
    \includegraphics[width=0.25\textwidth]{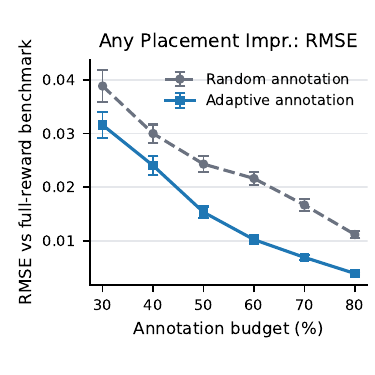}\includegraphics[width=0.25\textwidth]{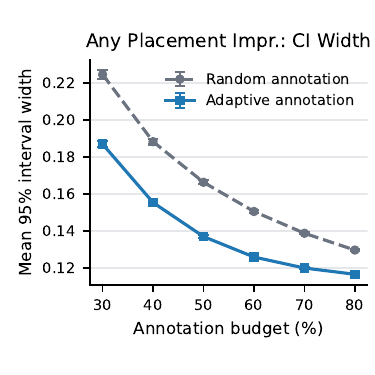}\includegraphics[width=0.25\textwidth]{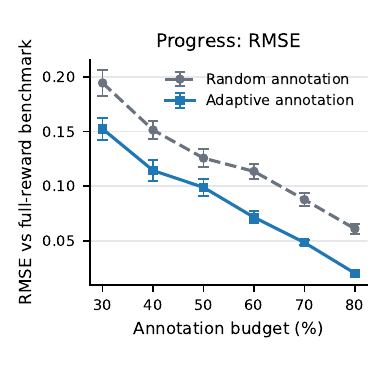}\includegraphics[width=0.25\textwidth]{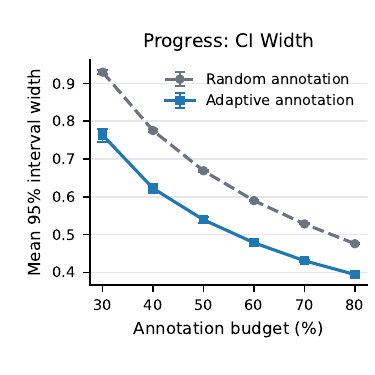}
    \caption{\textcolor{revisionblue}{Real-world data results, casenote data. Left two subplots showcase RMSE and interval width for the ``Any Placement'' outcome, with the binary-link variance estimate $\hat s_t^2=\hat\mu_t^2\hat b_t(1-\hat b_t)$, $\hat b_t$ fit on the first batch ($\kappa_1=0.3B$, $p=1$, i.e.\ $30\%$ of the budget). Right two subplots show RMSE and interval width for the ``Progress'' outcome, whose variance proxy is the sample variance $\hat{\Var}(\tilde R_t)$ of per-outreach labels. Since this is $\mcF_0$-measurable, we omit the initial pilot and the entire budget is allocated by the stagewise optimized probabilities of \Cref{thm-twostage-optimal-allocation-T=2}.}}\label{fig:bg}
\end{figure*}

\Cref{fig:simulated} illustrates the improvements from adaptive sampling for the harder target policy $(A_1=0,A_2=0)$, with the double reinforcement learning (DRL) estimator of \citet{kallus2019double}. We consider learning from a dataset of size $n=1000$, where we sample 100 Monte Carlo replications. The x-axis ranges over different budgets, ranging from low $20\%$ to high $80\%$ budget on the right-hand side. The left subplot shows the \textit{root-mean-squared error (RMSE) of the target policy value} against the truth. The adaptive allocation reduces DRL RMSE by $42\%$ at a $20\%$ budget, decaying to $\sim$$12\%$ at $80\%$; from a $40\%$ budget it attains the RMSE of full annotation, which random annotation does not reach even at $80\%$.
The right subplot shows the mean confidence interval width. \textcolor{revisionblue}{Adaptive intervals are substantially narrower (therefore more informative) throughout ($12$--$41\%$ narrower), and both designs attain near-nominal coverage ($0.93$--$0.96$ adaptive, $0.91$--$0.95$ random; \Cref{fig:synth-all-regimes}). The same pattern holds for each of the four unpersonalized policies (\Cref{fig:synth-all-regimes}).} However, our inferential results are asymptotic, and could potentially be further improved.

\subsection{Real-world data: casenotes from street outreach for homelessness services}

\paragraph{Data.}
We revisit the data of our motivating application. 
We restricted our analysis cohort to 775 clients who were observed regularly throughout the entire two-year period. 
The nonprofit typically seeks to engage clients for outreach at least three times a month, and record intended outreach attempts, including when they do not find the client. Casenotes are written after every outreach attempt and record unstructured interactions with outreach workers, ranging from initial engagement about services, to personal conversations that reveal important eligibility information, to actively working towards completing a housing application by collecting required documentation and income support and attending various required appointments. 

Our annotation task focuses on structured coding of \textit{progress towards a housing application}. We conducted extensive conversations with the nonprofit to refine the coding schema and recruited several Masters of Social Work students to code the casenotes accordingly. The schema is ordinal, with labels $(0)$ No progress made, $(0.975)$ conversation or meeting client need, $(2)$ discussing plans to complete intermediate tasks for housing application or other personal goals, $(2.5)$ signing documents and paperwork towards either of these, $(3)$ completing an appointment, $(3.5)$ improved condition, $(4)$ new placement. (Similar ordinal schemas appear in prior literature on impacts of street outreach \citep{ng2004outreach,levy2010homeless}.) We used the initial annotations to fine-tune a local language model (Gemma 3 12B) to approximate ground truth, and then manually ground-truthed the casenotes achieving max progress ($\sim$1{,}500 notes including ties, see \Cref{apx:reconciled} for more details). \textcolor{revisionblue}{
In the experiments, the zero-shot labels of the same base model are the always-observed silver labels. 
}

\paragraph{Experiment set-up.}
We embed the two-year casenote cohort as a two-stage dynamic treatment regime problem. Stage 1 corresponds
  to the first six months after cohort entry, with treatment discretized as the client's outreach intensity quintile and
  intermediate reward defined from observed progress during that period. \textcolor{revisionblue}{First-period progress summaries (constructed moments of the progress trajectory within that period) enter the stage-2 history. For simplicity, we consider reward annotation only.} Stage 2 corresponds to months 6--12, with treatment again defined by outreach-intensity quintile during this period. We consider two different estimands. \begin{itemize}
      \item The \textbf{progress} task treats the reward as cumulative
  progress, combining first-period and second-period progress rewards $R=(R_1, R_2)$, with the max progress over each action period. 
  \item The \textbf{placement improvement} task focuses only on a terminal
  housing outcome: ``any placement improvement'' is a binary indicator for whether the client's observed placement
  status improved at any point relative to baseline during the second-year outcome window (months 12--24). Thus $R=(0,Y_2)$ (where $Y\in\{0,1\}$ is housing placement improvement). 
  \end{itemize}Thus, the progress
  analysis asks whether adaptive annotation improves estimation of progress-based rewards, while the placement-
  improvement analysis asks whether adaptive annotation improves estimation and inference for a more directly policy-relevant
  terminal housing outcome. The progress is a dense, early-observed signal that is predictive of the sparse terminal placement improvement outcome: 22.5\% of the clients see a placement improvement by the end of the 2-year period, and more achieve intermediate progress, i.e. 58.1\% of clients achieve progress $\geq 2$ (start making plans, appointments) in stage 2 (months 6--12). 

The target policy increases outreach intensity by increasing each client's outreach intensity by one quintile. \textcolor{revisionblue}{Based on full-data DRL estimates, for the progress reward task, shifting outreach up one quintile has DRL-estimated value $4.22$ of total max-progress $(R_1+R_2)$, $95\%\;\mathrm{CI}: (4.05, 4.38)$, an absolute increase of $0.30$ progress units upon existing outreach. For any placement improvement, the corresponding DRL estimate is $0.279$ ($95\% \;\mathrm{CI}: (0.230, 0.328)$), an absolute increase of $5.4$ percentage points. Baselines for these increases are the observed means on the analysis population, i.e., the value of the existing outreach policy evaluated on its own data ($3.92$ progress units; $22.5\%$ placement improvement). Shifting outreach at a single stage, with the other stage left at the existing policy, increases any placement improvement by $4.2$ percentage points ($95\%$ CI $1.8$--$6.6$) at stage 2 and $2.1$ ($0.0$--$4.2$) at stage 1, and progress by $0.15$ ($0.09$--$0.21$) and $0.17$ ($0.09$--$0.25$) units, respectively, so the placement gain comes mainly from stage-2 outreach.}

\Cref{fig:bg} illustrates the results. Of course, the ground-truth is unknown, but we give a sense of uncertainty quantification by repeatedly subsampling 600 clients from the cohort for $100$ Monte Carlo replications --- however resulting standard errors are of course optimistically small. 

\textcolor{revisionblue}{The different rewards use different variance measures $\hat s_t^2$. \textit{Placement Improvement} is a binary outcome, so its conditional variance is exactly $b_t(1-b_t)$ for $b_t=\E[R_t\mid\mcF_{t-1}]$. We combine the fitted $b_t$ with this structural knowledge of the estimator and obtain $\hat s_t^2=\hat\mu_t^2\hat b_t(1-\hat b_t)$. The first two subplots show that adaptive annotation reduces DRL RMSE by about $20\%$ at budgets of $30$--$40\%$ and by $37$--$65\%$ at budgets of $50\%$ and above, and interval width by $16$--$18\%$ at budgets up to $60\%$. For the Progress task, every casenote has an LLM-coded silver label, and we can use the within-window sample variance of the silver (LLM) labels: $\hat s_t^2=\hat\mu_t^2\hat{\Var}(\tilde R_t)$, available for every casenote before annotation, so $\kappa_1=0$, the entire budget is allocated by the optimized probabilities, and the value functions are then fit on the annotated rows with inverse-inclusion weights. The design is the stagewise allocation of \Cref{thm-twostage-optimal-allocation-T=2}, so annotation of a client may stop after the first stage. This zero-gold-cost design reduces DRL RMSE by $21$--$37\%$ at budgets of $40$--$60\%$ and by $45$--$67\%$ at $70$--$80\%$, and interval width by $17$--$20\%$ across budgets. In terms of label efficiency, relative to the random-annotation budget that attains the same mean interval width, the optimized design uses $26$--$31\%$ fewer expert labels for placement at budgets $30$--$50\%$ and $25$--$28\%$ fewer for progress at budgets $30$--$60\%$.
Crucially, structure-specific prediction error estimation for $\sigma_t^2$ is important for improvements from adaptive allocations upon random. }
\paragraph{Comparison to active learning baseline.}
We compare to active learning baselines to ablate the importance of targeting our specific offline RL estimand. We consider a simple uncertainty-proportional design derived for a different mean estimation task \citep{zrnic2024active}, which allocates labels by the same reward variance estimates but without the action density ratios. (Other active learning methods for reward learning would similarly target prediction uncertainty). It recovers little of the gain. Therefore, gains from our method come from targeting the off-policy estimand, not reward uncertainty alone. 

\begin{table}[t]
\centering\footnotesize
\caption{\textcolor{revisionblue}{Comparison to active learning baselines. Change in DRL interval width and RMSE relative to random annotation (\%), casenote data, $100$ replications of $600$ clients. \emph{Optimized}: the designs of \Cref{fig:bg}. \emph{Uncertainty-prop.}: uncertainty-proportional sampling without the density ratios \citep{zrnic2024active}; same budget and estimator.}}\label{tab:active-inference}
\setlength{\tabcolsep}{2.3pt}
\begin{tabular}{@{}llcccccc@{}}
\toprule
 & & \multicolumn{3}{c}{Interval width} & \multicolumn{3}{c}{RMSE} \\
\cmidrule(lr){3-5}\cmidrule(l){6-8}
 & Budget & $30\%$ & $50\%$ & $70\%$ & $30\%$ & $50\%$ & $70\%$ \\
\midrule
Progress  & Optimized (Ours)        & $-18$ & $-19$ & $-18$ & $-22$ & $-21$ & $-45$ \\
          & Uncertainty-prop. & $-2$ & $-2$ & $-7$ & $+10$ & $+1$ & $-14$ \\
\addlinespace
Placement & Optimized (Ours)        & $-17$ & $-18$ & $-14$ & $-19$ & $-37$ & $-58$ \\
          & Uncertainty-prop. & $+2$ & $0$ & $0$ & $+17$ & $+9$ & $-2$ \\
\bottomrule
\end{tabular}
\end{table}

\paragraph{Estimation improvement results persist in a simulated MDP calibrated to real data.}\textcolor{revisionblue}{Since the ground truth is unknown here, \Cref{apx-simulator} repeats the two adaptive protocols on a simulator that we fit to the same cohort data. Fitting the simulator allows us to evaluate coverage relative to a ground-truth, although the resulting model-based effect estimates are less credible. We see that annotated-DRL intervals cover the true policy value at $0.91$--$0.96$ across budgets, and qualitatively similar estimation improvements in RMSE and interval width. 
}

\subsection{Real-world data: sequential preferences in LMArena}\label{sec:arena}

We also study annotation of human preferences using a large dataset from Arena, where users input a prompt, then compare two LLM model responses and vote for either response, a tie, or ``both bad.''\footnote{\url{https://huggingface.co/datasets/lmarena-ai/arena-human-preference-140k}} A user's prompt is routed to different potential LLM models. Arena must balance exploring many models to obtain evaluation data against the user experience, which may suffer if users read many completions from weaker models. We build a $T=2$ OPE task by keeping the first two prompts of each session. Action $A_t$ is the displayed model pair's category: two closed models, two open-weight models, or mixed. \citet{singh2026leaderboard} alleged that open-source models received differential coverage, while LMArena discussed user experience preservation as part of the reason for stratified sampling of models, rather than uniform sampling. We let stagewise rewards be the outcome  $U_t\in\{\mathrm{A},\mathrm{B},\mathrm{tie},\mathrm{both\ bad}\}$ denote the vote at stage $t$; we define the reward $R_t=\mathbf{1}\{U_t\neq\mathrm{both\ bad}\}$, which is the user's indication that at least one response is acceptable. Reward models such as Skywork-Reward-V2-Llama-3.1-8B are used as features for predicting human responses, as well as tf-idf embeddings of prompts and both responses. 


We analyze joint reward and state annotation. We evaluate the \emph{equal exposure} policy, assigning probability $1/3$ to each pair category; results for a smaller tilt toward open-weight pairs are in the Appendix. We use the Bernoulli outcome binary-link variance form.
\Cref{apx:arena} gives further sampling, fitting and evaluation details.

\textcolor{revisionblue}{\Cref{fig:arena} shows the results over $1{,}000$ replications of $21{,}685$ sessions. For scale, the full-annotation estimates are $\Phi^{\pi_b}=0.967$ and $\Phi^{\text{equal}}=0.958$, so equal exposure costs about $0.01$ (significant). Under equal exposure, adaptive annotation reduces RMSE against the full-annotation benchmark by $55$--$62\%$ at every budget and interval width by $34\%$ at a $30\%$ budget; relative to the random-annotation budget with the same interval width, it uses $37$--$48\%$ fewer labels at budgets $30$--$50\%$. Random allocation's RMSE ($0.0086$) at a $30\%$ budget  is still of the same order of magnitude of the equal-exposure effect, while the adaptive design's ($0.0033$) can enable inferential conclusions sooner.  }

\begin{figure}[t]
    \includegraphics[width=0.5\columnwidth]{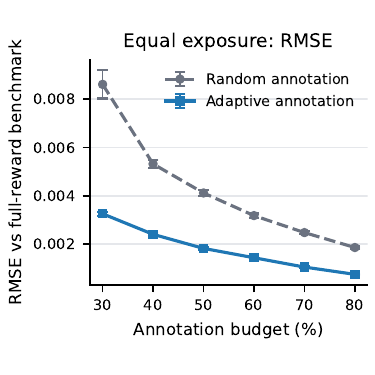}\includegraphics[width=0.5\columnwidth]{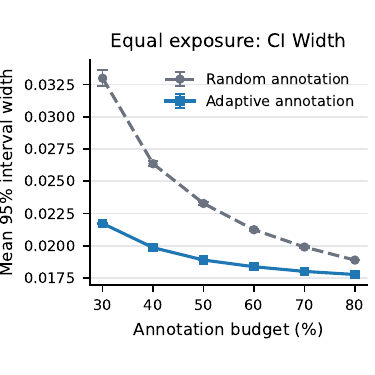}
    \caption{\textcolor{revisionblue}{LMArena results, with all labels used by the adaptive design charged to the annotation budget; $1{,}000$ Monte Carlo replications of $21{,}685$ sessions, error against each replication's full-annotation estimate. Left: RMSE; right: interval width, for the \emph{equal exposure} policy ($\pi_e$ uniform over the three model-pair categories). }}\label{fig:arena}
\end{figure}

\section{Conclusion}
\paragraph{Limitations} We would like to highlight some limitations with our method and analysis. We assume that annotation reveals ground-truth - in practice, annotations themselves can be noisy. We distinguish two settings. First, where annotations are observed with zero-mean noise, our results still hold (with larger variance). Second, the more difficult case is when there is no single ground-truth label \citep{plank2022problem}, such as toxicity assessment, in which case our method estimates the value under the distribution of annotators' judgments. In our casenote data, we have assumed sequential ignorability holds. In the appendix we include results from an MDP simulator learned from casenotes: indicating that the sample efficiency results still hold under potential violations of sequential ignorability. Due to the volume of the casenotes, we were unable to ground-truth them all; we ground-truthed the casenote with the maximum fine-tuned label in each client-stage, so the reward may miss higher-progress casenotes that were not reviewed.

\section*{Data, ethics, and availability}
\paragraph{Casenote data and annotation}
The casenote data were provided by the partner nonprofit. Master of Social Work students coded an initial sample of casenotes with the ordinal schema described in \Cref{sec:experiments}, which was used to fine-tune a local LLM (Gemma 3 12B). The always-observed silver labels are zero-shot labels from the same base model. The maximum-label casenote in each client-stage was then ground-truthed (Appendix). (All casenotes were processed only by local or privately hosted LLMs, preserving strict data privacy.) 

\paragraph{Ethics}
This study was conducted under a protocol determined exempt due to secondary analysis of de-identified data. The LMArena data are publicly released.

\paragraph{Data, Materials, and Software Availability}
Code for all methods and analyses is publicly available at \url{https://github.com/angelamzhou/optimal-sequential-annotations}. The LMArena data are publicly available (\url{https://huggingface.co/datasets/lmarena-ai/arena-human-preference-140k}). The casenote data are restricted to protect client privacy and cannot be shared. Details of the simulator are available from the authors on request.

\section*{Acknowledgments}
We are immensely grateful to our collaborating nonprofit and Lauri Goldkind for extensive discussions on understanding casenote data and its uses. We are also grateful to our Masters of Social Work (MSW) data annotators. We are grateful to the University of Southern California, Marshall School of Business Institute for Outlier Research in Business (iORB) for supporting funding. We thank Together.AI for a credit grant supporting this research.

\bibliographystyle{plainnat}
\bibliography{big-rl,sequential-annotation,active-annotation}

\clearpage
\appendix
\section{Handling Markovian and non-Markovian problem setups}\label{si-mdp-settings}
Define the history prior to $A_t$ as $\mathcal{H}_{S_t}=(S_1,A_1,R_1,\ldots,S_{t-1},A_{t-1},R_{t-
  1},S_t)$.
\begin{setting}[Markovian decision process]\label{setting-mdp}

The reward and transition law satisfy: $(R_t,S_{t+1}) \indep \mathcal{H}_{S_t}\mid S_t,A_t.$    
\end{setting}
\begin{setting}[Non-Markovian decision process]\label{setting-nmdp}
The conditional distributions of $(R_t,S_{t+1})$ can be history $\mathcal{H}_{S_t}$ dependent. 
\end{setting}

\begin{figure}[t]
  \centering
  \resizebox{0.6\linewidth}{!}{
%
%
%
%
\begin{tikzpicture}[
  >=Stealth,
  every node/.style={font=\small},
  obs/.style    ={draw, circle, inner sep=1pt, minimum size=8mm, fill=white},
  hidden/.style ={draw, circle, inner sep=1pt, minimum size=8mm, fill=gray!30},
  edge/.style   ={->, semithick},
  sametime/.style={dotted, thick},
  ci/.style     ={font=\small, color=blue!55!black}
]
  \node[obs]    (Spt)  at (0,2.0)   {$S'_t$};
  \node[obs]    (At)   at (2.2,2.0) {$A_t$};
  \node[hidden] (St)   at (1.0,0.5) {$S_t$};
  \node[hidden] (Rt)   at (3.2,0.5) {$R_t$};
  \node[obs]    (Stt)  at (0,-1.0)  {$\tilde S_t$};
  \node[obs]    (Rtt)  at (2.0,-1.0){$\tilde R_t$};

  \node[obs]    (Sptp) at (5.0,2.0)   {$S'_{t+1}$};
  \node[obs]    (Atp)  at (7.2,2.0)   {$A_{t+1}$};
  \node[hidden] (Stp)  at (6.0,0.5)   {$S_{t+1}$};
  \node[hidden] (Rtp)  at (8.2,0.5)   {$R_{t+1}$};
  \node[obs]    (Sttp) at (5.0,-1.0)  {$\tilde S_{t+1}$};
  \node[obs]    (Rttp) at (7.0,-1.0)  {$\tilde R_{t+1}$};

  \draw[edge] (Spt) -- (At);
  \draw[edge] (Spt) -- (St);
  \draw[edge] (Spt) to[bend right=18] (Stt);
  \draw[edge] (Spt) -- (Rt);
  \draw[edge] (At)  -- (Rt);
  \draw[sametime] (St) -- (Rt);      
  \draw[edge] (Rt) -- (Rtt);

  \draw[edge] (Rt)  to[bend right=18] (Stp);
  \draw[edge] (Rtt) -- (Sttp);
  \draw[edge] (Spt) to[bend left=28] (Sptp);

  \draw[edge] (Sptp) -- (Atp);
  \draw[edge] (Sptp) -- (Stp);
  \draw[edge] (Sptp) to[bend right=18] (Sttp);
  \draw[edge] (Sptp) -- (Rtp);
  \draw[edge] (Atp)  -- (Rtp);
  \draw[sametime] (Stp) -- (Rtp);
  \draw[edge] (Rtp) -- (Rttp);

  \node[hidden, scale=0.7] (legend) at (8.6,-1.6) {};
  \node[right=1.5mm of legend, font=\footnotesize] {Unobserved};

  \node[ci] at (1.0, 3.0)  {$A_t \perp S_t \mid S'_t$};
  \node[ci] at (8.6, 2.4)  {$R_t \perp S_t \mid A_t, S'_t$};
\end{tikzpicture}}
  \caption{Reward as (limited) cumulative history annotation.}
  \label{fig:bg-mdp-full}
\end{figure}

\section{Proofs of method}\label{apx-proofs-of-method}
\subsection{Proof of \Cref{prop:seq-annotation-unbiased}}
\begin{proof} \label{pf:seq-annotation-unbiased}
    Let 
    \[
        D_t^{\pi_e} = m_t^{\pi_e}(\mcF_t) - m_{t-1}^{\pi_e}(\mcF_{t-1}), \qquad \forall t=1,\dots, T.
    \]
    Since $\{(m_t^{\pi_e},\mcF_t)\}_{t=0}^{T}$ is a martingale, $\{(D_t^{\pi_e},\mcF_t)\}_{t=1}^{T}$ is a martingale difference sequence. Therefore,
    \[
        \E[D_t^{\pi_e}\mid \mcF_{t-1}]=0, \qquad \forall t=1,\dots, T.
    \]
    Building on this, we first show that every weighted annotation increment has mean zero. Since
    \begin{align*}
        P(C_{1:t}=1\mid\mcF_{t-1})&=P(C_1=1\mid\mcF_{t-1})\cdot\prod_{i=2}^{t}P(C_{1:i}=1\mid C_{1:i-1}=1,\mcF_{t-1})\\
        &= P(C_1=1\mid\mcF_0)\cdot\prod_{i=1}^t P(C_{1:i}=1\mid C_{1:i-1}=1,\mcF_{i-1}) \tag{$\because$\Cref{asn-anno-ignorability}}\\
        &=\Lambda_{1:t},
    \end{align*}
    the following holds by the iterated expectation and the conditional independence $D_t^{\pi_e}\indep C_{1:t}\mid \mcF_{t-1}$:
    \begin{align*}
    \E\left[
    \frac{C_{1:t}}{\Lambda_{1:t}}
    D_t^{\pi_e}
    \right]
    &=
    \E\left[\E[D_t^{\pi_e}\mid C_{1:t}=1, \mathcal{F}_{t-1}]\right]=\E[\E[D_t^{\pi_e}\mid \mcF_{t-1}]]=0.
    \end{align*}
    This reduces the expectation of $\tilde\Gamma_T^{\pi_e}$ to the mean of full imputation
    \[
    \E[\tilde{\Gamma}_T^{\pi_e}]
    =
    \E[m_0^{\pi_e}(\mcF_0)]
    +
    \sum_{t=1}^T
    \E\left[
    \frac{C_{1:t}}{\Lambda_{1:t}}
    D_t^{\pi_e}
    \right]
    =
    \E[m_0^{\pi_e}(\mcF_0)],
    \]
    which is the target value $\Phi^{\pi_e}$ by the tower rule
    \[
    \E[m_0^{\pi_e}(\mcF_0)]
    =
    \E[\E[\Gamma^{\pi_e}\mid \mcF_0]]
    =
    \E[\Gamma^{\pi_e}]
    =
    \Phi^{\pi_e}.
    \]
    Therefore, $\tilde\Gamma_T^{\pi_e}$ is unbiased under true annotation probability and the oracle nuisance.
    \end{proof}

\subsection{Proof of \Cref{prop-variance-decomposition}}
\begin{proof} \label{pf-variance-decomposition}
    Let
    \[
        W_t=\frac{C_{1:t}}{\Lambda_{1:t}}, \qquad D_t^{\pi_e}=m_t^{\pi_e}-m_{t-1}^{\pi_e}, \qquad \forall t=1,\dots, T.
    \]
    We proved in \cref{prop:seq-annotation-unbiased} that \(W_tD_t^{\pi_e}\) has mean zero. Also, for every \(s<t\), every cross-product term $W_sD_s^{\pi_e} W_t D_t^{\pi_e}$ has mean zero:
    \begin{align*}
    {\E}[W_sD_s^{\pi_e}W_tD_t^{\pi_e}]
    &=
    {\E}
    \left[
    \frac{C_{1:s} C_{1:t}}{\Lambda_{1:s}\Lambda_{1:t}}
    D_s^{\pi_e}D_t^{\pi_e}
    \right] \\
    & =\E\left[ \frac{C_{1:t}}{\Lambda_{1:s}\Lambda_{1:t}} D_s^{\pi_e}D_t^{\pi_e}\right] \tag{$\because$ $C_{1:s} C_{1:t} = C_{1:t}$ by forward monotonicity} \\
    &= \E\left[ \frac{D_s^{\pi_e}}{\Lambda_{1:s}} \E[D_t^{\pi_e} \mid C_{1:t}=1, \mcF_{t-1}] \right] \tag{$\because$ iterated expectation} \\
    & =0 &\tag{$\because$ $D_t^{\pi_e}\indep C_{1:t}\mid \mcF_{t-1}$, $\E[D_t^{\pi_e}\mid \mcF_{t-1}]=0$}
    \end{align*}
    Similarly, every cross-product term $m_0^{\pi_e}W_tD_t^{\pi_e}$ has mean zero:
    \begin{align*}
        \E[m_0^{\pi_e}W_t D_t^{\pi_e}]=\E\left[m_0^{\pi_e}\E[D_t^{\pi_e}\mid C_{1:t}=1, \mcF_{t-1}]\right]=0.
    \end{align*}
    Then, the variance of $\tilde\Gamma_T^{\pi_e}$ can be decomposed as follows:
    \[
    {\mathrm {Var}}(\tilde{\Gamma}_T^{\pi_e})
    =
    {\mathrm {Var}}(m_0^{\pi_e})
    +
    \sum_{t=1}^T
    \E[(W_tD_t^{\pi_e})^2].
    \]
    Note that the second term in the RHS above reduces as 
    \begin{align*}
        \E[(W_tD_t^{\pi_e})^2]=\E\left[\frac{C_{1:t}}{\Lambda_{1:t}^2}(D_t^{\pi_e})^2\right] =\E\left[\frac{\E[(D_t^{\pi_e})^2\mid C_{1:t}=1, \mcF_{t-1}]}{\Lambda_{1:t}}\ \right]=\E\left[
    \frac{\E\left[(D_t^{\pi_e})^2\mid\mcF_{t-1}\right]}{\Lambda_{1:t}}
    \right].
    \end{align*}
    Since \(m_{t-1}^{\pi_e}={\E}[m_t^{\pi_e}\mid {\cal F}_{t-1}]\), the above further reduces to
    \[
        \E\left[\frac{\Var(m_t^{\pi_e}\mid \mcF_{t-1})}{\Lambda_{1:t}}\right].
    \]
    Substituting this into the variance decomposition above yields the result.
\end{proof}

\subsection{Proof of \Cref{thm-twostage-optimal-allocation-T=2}} \label{pf:thm-twostage-optimal-allocation-T=2}
\begin{proof}
Since $\sigma_1^2,\sigma_2^2>0$ almost surely, the objective is strictly convex on positive probabilities and the constraints are linear. Hence a feasible design with nonnegative multipliers satisfying stationarity and complementary slackness is optimal \citep{luenberger1997optimization}, and by strict convexity it is unique. We construct such a design.\\

To derive them, we partially dualized the constraints and formed the Lagrangian. Let $\beta\ge0$ be the scalar multiplier for the budget constraint $\E[\lambda_1+\Lambda_{1:2}]\le 2B$; $\eta(\mcF_1)\ge0$ the $\mcF_1$-measurable multiplier for the monotonicity constraint $\Lambda_{1:2}\le\lambda_1$; and $\nu(\mcF_0)\ge0$ the $\mcF_0$-measurable multiplier for $\lambda_1\le1$. The Lagrangian of \eqref{eqn-opt_prob} for $T=2$ is
\begin{align*}
\mathcal L(\lambda_1,\Lambda_{1:2},\beta,\eta,\nu)
&=\E\left[\frac{\sigma_1^2(\mcF_0)}{\lambda_1(\mcF_0)}\right]
+\E\left[\frac{\sigma_2^2(\mcF_1)}{\Lambda_{1:2}(\mcF_1)}\right]
+\beta\{\E[\lambda_1+\Lambda_{1:2}]-2B\}\\
&\quad+\E[\eta(\mcF_1)\{\Lambda_{1:2}(\mcF_1)-\lambda_1(\mcF_0)\}]
+\E[\nu(\mcF_0)\{\lambda_1(\mcF_0)-1\}].
\end{align*}
Because $\lambda_1$ is $\mcF_0$-measurable, the tower property gives
\[
\E[\eta(\mcF_1)\lambda_1(\mcF_0)]=\E\big[\lambda_1(\mcF_0)\,\E[\eta(\mcF_1)\mid\mcF_0]\big].
\]
Collecting the terms in $\lambda_1$ and in $\Lambda_{1:2}$, we obtain
\begin{align*}
\mathcal L(\lambda_1,\Lambda_{1:2},\beta,\eta,\nu)
&=\E\left[\frac{\sigma_1^2(\mcF_0)}{\lambda_1(\mcF_0)}
+\bigl\{\beta+\nu(\mcF_0)-\E[\eta(\mcF_1)\mid\mcF_0]\bigr\}\lambda_1(\mcF_0)\right]\\
&\quad+\E\left[\frac{\sigma_2^2(\mcF_1)}{\Lambda_{1:2}(\mcF_1)}
+\{\beta+\eta(\mcF_1)\}\Lambda_{1:2}(\mcF_1)\right]
-2\beta B-\E[\nu(\mcF_0)].
\end{align*}
  For fixed multipliers $\beta$, $\eta(\mcF_1)$, and $
  \nu(\mcF_0)$,
  minimizing the two integrands pointwise over $
  \lambda_1(\mcF_0)$
  and $\Lambda_{1:2}(\mcF_1)$, respectively, gives the
  stationarity
  conditions. Define the Lagrangian dual function $g$:
\begin{align*} g(\beta,\eta,\nu)&=\min_{\substack{\lambda_1(\mcF_0) \ge0, \\0\le \Lambda_{1:2}(\mcF_1)\le 1}} \mathcal L(\lambda_1,\Lambda_{1:2}, \beta, \eta, 
  \nu) \\ &= \min_{\lambda_1(\mcF_0)\ge0} \E\left[
  \frac{\sigma_1^2(\mcF_0)}{\lambda_1(\mcF_0)}
  +\bigl\{\beta+
  \nu(\mcF_0)-\E[\eta(\mcF_1)\mid\mcF_0]\bigr\}
  \lambda_1(\mcF_0)
  \right] \\
   & + \min_{0\le\Lambda_{1:2}(\mcF_1)\le1} \E\left[
  \frac{\sigma_2^2(\mcF_1)}{\Lambda_{1:2}(\mcF_1)}
  +\{\beta+\eta(\mcF_1)\}\Lambda_{1:2}(\mcF_1)
  \right]
    -2\beta B-\E[\nu(\mcF_0)].
  \end{align*}
The dual function separates into a minimization over $\lambda_1$ and one over $\Lambda_{1:2}$, and each can be solved pointwise in $\mcF_0$ and $\mcF_1$, respectively. Taking derivatives with respect to $\lambda_1(\mcF_0)$ and $\Lambda_{1:2}(\mcF_1)$ gives the first-order conditions:
\begin{align}
    \frac{\partial \mathcal{L}}{\partial \lambda_1(\mcF_0)} &= -\frac{\sigma_1^2(\mcF_0)}{\lambda_1(\mcF_0)^2}
+\beta-\E[\eta(\mcF_1)\mid\mcF_0]+\nu(\mcF_0)=0
    \label{eqn-apx-stationarity-1}
    \\
    \frac{\partial \mathcal{L}}{\partial \Lambda_{1:2}(\mcF_1)} &= -\frac{\sigma_2^2(\mcF_1)}{\Lambda_{1:2}(\mcF_1)^2}+\beta+\eta(\mcF_1)=0
        \label{eqn-apx-stationarity-2}
\end{align}

Together with feasibility and nonnegative multipliers, complementary slackness requires
\begin{align*}
\eta(\mcF_1)\{\Lambda_{1:2}(\mcF_1)-\lambda_1(\mcF_0)\}&=0\quad\text{a.s.},\\
\nu(\mcF_0)\{\lambda_1(\mcF_0)-1\}&=0\quad\text{a.s.},\\
\beta\{\E[\lambda_1+\Lambda_{1:2}]-2B\}&=0.
\end{align*}

First we derive a solution candidate for the interior case, where $\eta = \nu = 0.$ Then, the stationarity conditions reduce to: 
\begin{align*}
    0 &= -\frac{\sigma_1^2(\mcF_0) }{\lambda_1^2(\mcF_0)} + \beta
    \implies \lambda_1^*(\mcF_0) = \frac{  \sqrt{\sigma_1^2(\mcF_0)} }{\sqrt{\beta}} \\
    0 &= -\frac{ \sigma_2^2(\mcF_1) }{\Lambda_{1:2}^2(\mcF_1)} + \beta \implies       \Lambda_{1:2}^*(\mcF_1) =
    \frac{  \sqrt{\sigma_2^2(\mcF_1)} }{\sqrt{\beta}}
\end{align*}

The optimal solutions satisfy the budget equality constraint, plugging in the first-order solutions to the budget constraint, we obtain 
$$ \frac{1}{\sqrt{\beta}} \left(\E\left[  \sqrt{\sigma_1^2(\mcF_0)} + \sqrt{\sigma_2^2(\mcF_1)} \right]\right) = 2B \quad \Longrightarrow \quad \beta=\frac{1}{4B^2}\left(\E\left[\sqrt{\sigma_1^2(\mcF_0)}+\sqrt{\sigma_2^2(\mcF_1)}\right]\right)^2
$$
The candidate is an interior solution if
\[
0<\frac{\sqrt{\sigma_2^2(\mcF_1)}}{\sqrt\beta}
<\frac{\sqrt{\sigma_1^2(\mcF_0)}}{\sqrt\beta}<1
\qquad\text{a.s.}
\]
Then this solution is feasible, and KKT conditions hold with $\eta = \nu = 0.$

Next, consider the boundary case where the monotonicity constraint \(\Lambda_{1:2}(\mcF_1)\leq \lambda_1(\mcF_0)\) binds. For fixed $\beta > 0$, we solve the KKT conditions backwards: first express $\Lambda_{1:2}$ and $\eta$ in terms of $\lambda_1$, then substitute into the first-stage stationarity condition. 

Holding $\lambda_1 > 0$ fixed, second-stage stationarity \eqref{eqn-apx-stationarity-2} gives
\[
\Lambda_{1:2}(\mcF_1)=\sqrt{\frac{\sigma_2^2(\mcF_1)}{\beta+\eta(\mcF_1)}}.
\]

There are two cases. If $\sqrt{\sigma_2^2(\mcF_1)/\beta} \leq \lambda_1(\mathcal{F}_0)$, $(\lambda_1, \Lambda_{1:2})=(\lambda_1, \sqrt{\sigma_2^2({\mcF_1})/\beta})$ is
  feasible
  and satisfies the stationarity and the complementary slackness
  with
  $\eta(\mcF_1)=0$.

  Otherwise, the second-stage stationarity condition \eqref{eqn-apx-stationarity-2} gives
  $\eta(\mcF_1)=\sigma_2^2(\mcF_1)/
  \lambda_1(\mcF_0)^2-\beta>0$. Then by the complementary slackness, we have $\lambda_1=\Lambda_{1:2}$ a.s.. Combining these cases, we obtain: 
  \[
\eta(\mcF_1)
=\left(\frac{\sigma_2^2(\mcF_1)}{\lambda_1^*(\mcF_0)^2}-\beta\right)_+ \quad \Longrightarrow \quad \Lambda_{1:2}^*(\mcF_1)
=\min\left\{\lambda_1^*(\mcF_0),\sqrt{\frac{\sigma_2^2(\mcF_1)}{\beta}}\right\}.
\]

Therefore, we have solved the second-stage conditions for
  $\Lambda_{1:2}$ and $\eta$. Next, let us consider the cases for the upper-bound multiplier, $\nu$ linked to the upper-bound constraint $\lambda_1(\mathcal F_0)\le1$. Consider the case with $\nu=0$ (inactive upper bound constraint). 

  First-stage stationarity with $
  \nu=0$
  requires solving
  \[
  \beta
  =
  \frac{\sigma_1^2(\mcF_0)}{\lambda_1^*(\mcF_0)^2}
  +
  \E\left[
  \left\{
  \frac{\sigma_2^2(\mcF_1)}{\lambda_1^*(\mcF_0)^2}
  -\beta
  \right\}_+
  \;\middle|\;\mcF_0
  \right].
  \]

The second-stage allocation is clipped at $
  \lambda_1(\mcF_0)$
  on histories where
  $\sqrt{\sigma_2^2(\mcF_1)/\beta}>\lambda_1(\mcF_0)$.
 
  Increasing $\lambda_1$ relaxes the second-stage
  constraint on
  those histories. The optimal first-stage allocation
  therefore
  depends on the conditional benefit of enabling further
  annotation.
  Consequently, the interior
  formula
  $\lambda_1^*=\sqrt{\sigma_1^2/\beta}$ no longer applies
  when
  $\E[\eta(\mcF_1)\mid\mcF_0]>0$.
  
 With the first-stage upper bound inactive, reparametrize
  stationarity using $Z(\mcF_0)=\beta\lambda_1(\mcF_0)^2$.
  This gives a scalar equation conditional on $\mcF_0$ that
  does not depend on $\beta$. After solving for $Z$,
  recover
  the allocation probabilities as functions of $\beta$ and
  choose $\beta$ to satisfy the global budget constraint.

  We multiply by $\lambda_1(\mcF_0)^2$, simplify, and therefore reparametrize:
\begin{align*}
\beta\lambda_1(\mcF_0)^2
&=\sigma_1^2(\mcF_0)
+\lambda_1(\mcF_0)^2\E[\eta(\mcF_1)\mid\mcF_0]\\
&=\sigma_1^2(\mcF_0)
+\E\left[\lambda_1(\mcF_0)^2
\left(\frac{\sigma_2^2(\mcF_1)}{\lambda_1(\mcF_0)^2}-\beta\right)_+
\,\middle|\,\mcF_0\right]\\
&=\sigma_1^2(\mcF_0)
+\E\left[\{\sigma_2^2(\mcF_1)-\beta\lambda_1(\mcF_0)^2\}_+
\mid\mcF_0\right].
\end{align*}
  Therefore we obtain a scalar fixed-point equation in $Z(\mathcal{F}_0)$: 
\[
Z(\mcF_0)
=
\sigma_1^2(\mcF_0)
+
\E[
\{\sigma_2^2(\mcF_1)-Z(\mcF_0)\}_+
\mid \mcF_0].
\]

Define
\[
H(z;\mcF_0)=\sigma_1^2(\mcF_0)-z+\E\big[\{\sigma_2^2(\mcF_1)-z\}_+\mid\mcF_0\big],\qquad z\ge0.
\]
The map $z\mapsto H(z;\mcF_0)$ is continuous and strictly decreasing, with slope at most $-1$, since $z\mapsto\{\sigma_2^2-z\}_+$ is nonincreasing. Also $H(0;\mcF_0)>0$ and $H(z;\mcF_0)\to-\infty$ as $z\to\infty$. Hence the fixed-point equation $H(Z;\mcF_0)=0$ has a unique root $Z(\mcF_0)>0$, and $H(z;\mcF_0)\ge0$ if and only if $z\le Z(\mcF_0)$. With $\nu=0$, the first-stage allocation is $\lambda_1^*=\sqrt{Z/\beta}$, which is feasible exactly when $Z(\mcF_0)\le\beta$.

Now suppose the upper bound binds, $\lambda_1(\mcF_0)=1$, so that $\nu(\mcF_0)\ge0$ is permitted. The second-stage solution above, with $\lambda_1=1$, gives
\[
\eta(\mcF_1)=\{\sigma_2^2(\mcF_1)-\beta\}_+,\qquad
\Lambda_{1:2}^*(\mcF_1)=\min\left\{1,\sqrt{\frac{\sigma_2^2(\mcF_1)}{\beta}}\right\}.
\]
First-stage stationarity \eqref{eqn-apx-stationarity-1} at $\lambda_1=1$ then determines the multiplier:
\[
\nu(\mcF_0)=\sigma_1^2(\mcF_0)+\E\big[\{\sigma_2^2(\mcF_1)-\beta\}_+\mid\mcF_0\big]-\beta=H(\beta;\mcF_0).
\]
Dual feasibility $\nu\ge0$ holds if and only if $\beta\le Z(\mcF_0)$, which is exactly the case in which $\sqrt{Z/\beta}\ge1$. The two cases therefore partition the histories, and in both
\[
\lambda_1^*(\mcF_0)=\min\left\{1,\sqrt{\frac{Z(\mcF_0)}{\beta}}\right\},
\qquad
\Lambda_{1:2}^*(\mcF_1)
=
\min\left\{
\lambda_1^*(\mcF_0),
\sqrt{\frac{\sigma_2^2(\mcF_1)}{\beta}}
\right\},
\]
with $\eta(\mcF_1)=\{\sigma_2^2(\mcF_1)/\lambda_1^*(\mcF_0)^2-\beta\}_+$ and $\nu(\mcF_0)=\{H(\beta;\mcF_0)\}_+$.
 Once $Z$ is computed, a scalar search over $\beta$ matches $\E[\lambda_1^*+\Lambda_{1:2}^*]=2B$.
\end{proof}

\subsection{Two-stage projection increments}\label{si-twostage-increments}
\begin{assumption}[Proxy reward exclusion restriction]\label{asn-annotation-exclusion-restriction}
$\E[ R_t \mid \tilde{R}_t, S_t, A_t] = \E[ R_t \mid \tilde{R}_t, S_t', A_t] $
\end{assumption}
To provide more intuition, we first consider the two-stage case, where $T=2$. The always-observed data is $\mathcal{O}_0 = \{ (S_1', A_1, \tilde{R}_1, \tilde{S}_2, A_2, \tilde{R}_2) \}$. 


For $T=2$, the first annotation reveals $(R_1,S_2)$,
  while the
  second reveals $R_2$. For an NMDP setting, starting from
  $$ m_0^{\pi_e}(\mcF_0)
=
V_1^{\pi_e}+
  \sum_{t=1}^T\gamma^{t-1}\mu_t
  \bigl(b_t(\mcF_0)+\gamma\tilde V_{t+1}^{\pi_e}(\mcF_0)
  -\tilde Q_t^{\pi_e}(\mcF_0)\bigr)
  ,$$
  with $\tilde Q_1^{\pi_e}(\mcF_0)=Q_1^{\pi_e}$ and $\tilde V_{T+1}^{\pi_e}(\mcF_0)=0$, the additional incremental projections are
  \begin{align*}
  m_1^{\pi_e}(\mcF_1)-m_0^{\pi_e}(\mcF_0)
  &= \mu_1
     \Bigl\{R_1-b_1(\mcF_0)
     +\gamma\bigl(V_2^{\pi_e}-\tilde
     V_2^{\pi_e}(\mcF_0)\bigr)\Bigr\}
     \\
  &\quad
  +\gamma\Bigl\{
     \mu_2\bigl(b_2(\mcF_1)-Q_2^{\pi_e}\bigr)
     -\E\bigl[\mu_2
        \{R_2-Q_2^{\pi_e}\}
        \mid\mcF_0\bigr]\Bigr\},\\
  m_2^{\pi_e}(\mcF_2)-m_1^{\pi_e}(\mcF_1)
  &=\gamma\mu_2(R_2-b_2(\mcF_1)).
  \end{align*}
By \Cref{asn-annotation-exclusion-restriction}, $\E[R_{t+1}\mid \mcF_t]-\E[R_{t+1}\mid \mcF_{t-1}]=0.$

\subsection{Generalization of \Cref{thm-twostage-optimal-allocation-T=2} to $T>2$}

In the main text, we give the solution for $T=2$. The key structure of the solution can be generalized to $T>2$ via backwards induction, although the annotation probabilities no longer have such a simple closed form. The key idea is that the optimization has one global budget constraint, and additional constraints on the annotation probabilities that enforce forward monotonicity. The first global constraint can be handled with a Lagrangian relaxation as in the $T=2$ case, leading to a conditionally-separable optimization problem. Then the one scalar value (the dual variable) can be optimized over. The forward monotone constraints can be handled with backwards recursion. 

As in \eqref{eqn-opt_prob}, write $\sigma_t^2=\Var(m_t^{\pi_e}(\mcF_t)\mid\mcF_{t-1})$ and $\Lambda_{1:t}=\Lambda_{1:t}(\mcF_{t-1})$. Thus $\sigma_t^2$ includes the density ratios and discount factors in the projection increment; division by $\Lambda_{1:t}$ accounts for annotation sampling. For example, under reward-only annotation, $\sigma_t^2=\gamma^{2(t-1)}\mu_t^2\Var(R_t\mid\mcF_{t-1})$.

\begin{theorem}[Optimal allocation probabilities for general $T$]\label{thm-general-t-optimal-allocation}
Consider \eqref{eqn-opt_prob} for fixed $T$ and $0<B<1$, with bounded variance coefficients $\sigma_t^2$. Set $\Lambda_{1:0}=1$ and interpret $0/0=0$ and $a/0=+\infty$ for $a>0$. For each $z\ge0$, define the backward functions
\begin{equation}\label{eq:general-t-backward}
\begin{aligned}
H_T(z)&=\sigma_T^2-z,\\
H_t(z)&=\sigma_t^2-z+
\E\big[\{H_{t+1}(z)\}_+\mid\mcF_{t-1}\big],
\qquad t=T-1,\ldots,1.
\end{aligned}
\end{equation}
Here $\{u\}_+=\max\{u,0\}$. Let $z_t^*$ be the unique nonnegative root of $H_t(z_t^*)=0$; it is $\mcF_{t-1}$-measurable. If $BT<\sum_{t=1}^T\Pr(z_t^*>0)$, choose $\beta>0$ so that the recursion
\begin{equation}\label{eq:general-t-allocation}
\Lambda_{1:0}^*=1,\qquad
\Lambda_{1:t}^*
=\min\left\{\Lambda_{1:t-1}^*,\sqrt{\frac{z_t^*}{\beta}}\right\},
\quad t=1,\ldots,T,
\end{equation}
satisfies $\sum_{t=1}^T\E[\Lambda_{1:t}^*]=BT$. Such a $\beta$ exists, and \eqref{eq:general-t-allocation} solves \eqref{eqn-opt_prob}. The conditional annotation probabilities are $\lambda_t^*=\Lambda_{1:t}^*/\Lambda_{1:t-1}^*$ whenever the denominator is positive, and may be set to zero otherwise. If instead $BT\ge\sum_{t=1}^T\Pr(z_t^*>0)$, the design $\Lambda_{1:t}^*=\mathbf 1\{z_t^*>0\}$ is optimal and the budget need not bind. If every $\sigma_t^2>0$ almost surely, the budget binds for every $B\in(0,1)$ and the optimal design is unique up to null sets.
\end{theorem}

When $BT<\sum_t\Pr(z_t^*>0)$, compute the roots $z_t^*$ by scalar bisection, then choose $\beta$ by a scalar search over $\sum_t\E[\Lambda_{1:t}^*]$. For $T=2$, $z_2^*=\sigma_2^2$. 

With different surrogate variances (nonnegative and bounded) in place of $\sigma_t^2$, the result follows with respect to the corresponding surrogate objective.

\begin{proof}[Proof of \Cref{thm-general-t-optimal-allocation}]
For the constraint $\sum_{t=1}^T\E[\Lambda_{1:t}]\le BT$, define the Lagrangian
\[
\mathcal L(\{\Lambda_{1:t}\},\beta)
=\sum_{t=1}^T\E\left[\frac{\sigma_t^2}{\Lambda_{1:t}}+\beta\Lambda_{1:t}\right]-\beta BT,
\qquad \beta\ge0.
\]
Fix $\beta>0$. We first minimize $\mathcal L$ over the measurable, nested probabilities, without the budget constraint. For $x\ge0$, define the conditional optimization problem
\[
\begin{aligned}
L_t(x)=\min_{\{\Lambda_{1:j}\}_{j=t}^T}\quad
&\E\left[\sum_{j=t}^T\left\{\frac{\sigma_j^2}{\Lambda_{1:j}}+\beta\Lambda_{1:j}\right\}\,\middle|\,\mcF_{t-1}\right]\\
\text{subject to}\quad
&\Lambda_{1:j}\text{ is }\mcF_{j-1}\text{-measurable},\qquad j=t,\ldots,T,\\
&0\le\Lambda_{1:t}\le x,\qquad
0\le\Lambda_{1:j}\le\Lambda_{1:j-1},\quad j=t+1,\ldots,T.
\end{aligned}
\]
The minimum is conditional on $\mcF_{t-1}$; the calculation below constructs allocations that attain it. Here $x$ represents $\Lambda_{1:t-1}$, and $\beta$ is fixed. With $L_{T+1}\equiv0$, the tower property gives
\begin{equation}\label{eq:general-t-bellman}
L_t(x)=\min_{0\le u\le x}\left\{
\frac{\sigma_t^2}{u}+\beta u+
\E[L_{t+1}(u)\mid\mcF_{t-1}]
\right\}.
\end{equation}
Since the positive-part map is $1$-Lipschitz, backward induction in \eqref{eq:general-t-backward} gives
\[
|H_t(z)-H_t(z')|\le (T-t+1)|z-z'|,\qquad z,z'\ge0.
\]
Thus $H_t$ is continuous. Since $H_t(0)=\E[\sum_{j=t}^T\sigma_j^2\mid\mcF_{t-1}]\ge0$ and $H_t(z)<0$ for sufficiently large $z$, the intermediate value theorem gives a nonnegative root. The recursion makes $H_t$ strictly decreasing, so this root is unique. Hence $z_t^*$ is well defined. On $\{z_t^*=0\}$, $\sigma_j^2=0$ for every $j\ge t$ almost surely, so zero allocations give $L_t(x)=0$.

We prove by backward induction that, almost surely,
\begin{equation}\label{eq:general-t-marginal}
\frac{d}{dx}L_t(x)=-\frac{\{H_t(\beta x^2)\}_+}{x^2},\qquad x>0.
\end{equation}
At $t=T$, the minimizer of $\sigma_T^2/u+\beta u$ over $0\le u\le x$ is $u=\min\{x,\sqrt{\sigma_T^2/\beta}\}$, which gives \eqref{eq:general-t-marginal}. For the induction step, the derivative of the objective in \eqref{eq:general-t-bellman} is
\[
-\frac{\sigma_t^2}{u^2}+\beta
-\frac{\E[\{H_{t+1}(\beta u^2)\}_+\mid\mcF_{t-1}]}{u^2}
=-\frac{H_t(\beta u^2)}{u^2}.
\]
Boundedness of the variance coefficients justifies differentiation under the conditional expectation for $u>0$. If $z_t^*>0$, the derivative is negative for $u<\sqrt{z_t^*/\beta}$ and positive for $u>\sqrt{z_t^*/\beta}$. The minimizer is therefore $u=\min\{x,\sqrt{z_t^*/\beta}\}$; substituting it gives \eqref{eq:general-t-marginal}. If $z_t^*=0$, the minimum is zero at $u=0$, and the same derivative formula holds. Applying these measurable minimizers successively from $x=1$ gives \eqref{eq:general-t-allocation} and
\[
\mathcal L(\{\Lambda_{1:t}^*\},\beta)=\E[L_1(1)]-\beta BT
\le\mathcal L(\{\Lambda_{1:t}\},\beta)
\]
for every measurable design satisfying $0\le\Lambda_{1:t}\le\Lambda_{1:t-1}$ and $\Lambda_{1:0}=1$.

The function $\beta\mapsto\sum_t\E[\Lambda_{1:t}^*]$ is continuous and nonincreasing on $(0,\infty)$. Since $z_t^*=0$ implies $z_{t+1}^*=0$ almost surely, bounded convergence gives limits of zero as $\beta\to\infty$ and $\sum_t\Pr(z_t^*>0)$ as $\beta\downarrow0$. Thus a multiplier satisfying $\sum_t\E[\Lambda_{1:t}^*]=BT$ exists whenever $0<BT<\sum_t\Pr(z_t^*>0)$. For this multiplier and any budget-feasible design,
\[
\sum_t\E\left[\frac{\sigma_t^2}{\Lambda_{1:t}^*}\right]
=\mathcal L(\{\Lambda_{1:t}^*\},\beta)
\le\mathcal L(\{\Lambda_{1:t}\},\beta)
\le\sum_t\E\left[\frac{\sigma_t^2}{\Lambda_{1:t}}\right].
\]
The equality uses $\sum_t\E[\Lambda_{1:t}^*]=BT$; the last inequality uses $\beta\ge0$ and $\sum_t\E[\Lambda_{1:t}]\le BT$. When $BT\ge\sum_t\Pr(z_t^*>0)$, the feasible design $\Lambda_{1:t}^*=\mathbf1\{z_t^*>0\}$ attains the lower bound $\sum_t\E[\sigma_t^2]$. Finally, if all $\sigma_t^2>0$ almost surely, then $\sum_t\Pr(z_t^*>0)=T$, the budget binds, and strict convexity of $\sum_t\E[\sigma_t^2/\Lambda_{1:t}]$ gives uniqueness up to null sets.
\end{proof}

\subsection{Further simplified estimator via local approximation and robustness guarantees}\label{si-local-approx}
For MDPs with joint reward and state annotation
(\Cref{setting-intermediate}), annotation $t$ reveals $(R_t,S_{t+1})$.
The exact projection increment is
\begin{align}
m_t^{\pi_e}(\mcF_t)
-m_{t-1}^{\pi_e}(\mcF_{t-1})
&=\gamma^{t-1}\Bigl\{
  \mu_t(R_t+\gamma V_{t+1}^{\pi_e})
  -\E[\mu_t(R_t+\gamma V_{t+1}^{\pi_e})
       \mid\mcF_{t-1}]
  \Bigr\} \label{eqn-mdp-increment}
  \\
&\quad
+\sum_{j=t+1}^{T}\gamma^{j-1}\Bigl\{
  \E[\mu_j
      (R_j+\gamma V_{j+1}^{\pi_e}-Q_j^{\pi_e})
      \mid\mcF_t]
  -\E[\mu_j^{\pi_e}
      (R_j+\gamma V_{j+1}^{\pi_e}-Q_j^{\pi_e})
      \mid\mcF_{t-1}]
  \Bigr\}. \nonumber
\end{align}
Here 
$V_{T+1}^{\pi_e}=0$, and the sum is empty when $t=T$.
While $\mcF_{t-1}$ reveals 
$\mu_tQ_t^{\pi_e}$, the other terms in the last 
 sum record how the new annotation updates predictions
of future weighted residuals.

The local approximation combines two approximations. The first approximation can be viewed as using a misspecified predictor, 
$$ \text{estimating } \E[ \mu_{t+1}^{\pi_e}Q_{t+1}^{\pi_e}\mid\mcF_{t-1} ]\text{ with }
       \E[ \mu_{t+1}^{\pi_e}(S_{t+1},A_{t+1}) \mid \mcF_{t-1}]\,
  \E[Q_{t+1}^{\pi_e}\mid\mcF_{t-1}].$$
The second simplification omits updates to future value predictions induced by annotation $t$, which reveals $(R_t,S_{t+1})$:
\[
\E[Q_{t'}^{\pi_e}\mid\mcF_t]
=\E[Q_{t'}^{\pi_e}\mid\mcF_{t-1}],
\qquad t'\ge t+2.
\]
When the corresponding density ratio is
  $\mcF_0$-measurable, such omitted updates have
  mean zero by iterated expectation, e.g. $\E[\mu_2^{\pi_e}\{b_2(\mcF_1)-b_2(\mcF_0)\}]=0$, though omission can change the variance.

Our estimator is fairly robust to \textit{approximate} projections onto $\mcF_{t-1}$, which enables such approximations, so long as standard DRL product-error rates are satisfied. We summarize this in the following, where $\hat{\bm\eta}_m = \{ \{\hat\mu_t,\hat Q_t^{\pi_e},\hat b_t\}_{t=1}^T,
\{\hat{\tilde V}_t^{\pi_e},\hat{\tilde Q}_t^{\pi_e}\}_{t=2}^T\}$.

\begin{proposition}[Bias under approximate projections]
\label{prop:approximate-projection-bias}
Consider a Markov decision process and suppose that \Cref{setting-mdp} and \cref{asn-seq-igno-stilde,asn-seq_overlap,asn-anno-ignorability}
hold. Condition on an independent training sample. With the actual annotation probabilities $\bm\Lambda$, and  $\hat m_T=\hat\Gamma_T^{\pi_e}$, then
\begin{align*}
\E\big[\tilde\Gamma_T^{\pi_e}(O;\bm\Lambda,\hat{\bm\eta}_m)\big]-\Phi^{\pi_e}
&=\E_{\pi_b}[\hat m_T^{\pi_e}-m_T^{\pi_e}]
=\sum_{t=1}^T\gamma^{t-1}\E_{\pi_b}\big[
(\hat\mu_t-\mu_t)\{\gamma(\hat V_{t+1}^{\pi_e}-V_{t+1}^{\pi_e})-(\hat Q_t^{\pi_e}-Q_t^{\pi_e})\}\big].
\end{align*}
\end{proposition}

\az{$m_t^\dagger$ is plim of $\hat m_t$ }Let $m_t^\dagger$ be the frozen predictions with full-data endpoint $m_T^\dagger=\Gamma_T^{\pi_e}$. Let $m_t^{\pi_e}=\E[\Gamma_T^{\pi_e}\mid\mcF_t]$ be the exact projections, and let $\bm\eta_m^\dagger$ collect the frozen predictions. 

\begin{proposition}[Surrogate variance]\label{prop-surrogate-variance}

Assume annotation ignorability, overlap and finite second moments. Let $m_t=\E[\Gamma_T^{\pi_e}\mid\mcF_t]$ and let $m_t^\dagger$ be $\mcF_t$-measurable frozen predictions with $m_T^\dagger=\Gamma_T^{\pi_e}$. Let $\bm\eta_m$ and $\bm\eta_m^\dagger$ collect these respective predictions. Using the realized annotation probabilities, and set $\Lambda_{1:0}=1$. The budget-constraint class is fixed throughout.

\textit{(i) Surrogate variance.}
   With finite second moments and $\Lambda_{1:0}=1$,
\begin{equation}
\begin{aligned}
\Var(\tilde\Gamma_T^{\pi_e}(O;\bm\Lambda,\bm\eta_m^\dagger))
&={\Var(m_0^{\pi_e})+\sum_{t=1}^{T}\E\!\left[
\frac{\Var(m_t^{\pi_e}\mid\mcF_{t-1})}{\Lambda_{1:t}}\right]}
+\underbrace{\sum_{t=1}^{T}\E\!\left[
\left(\frac{1}{\Lambda_{1:t}}-\frac{1}{\Lambda_{1:t-1}}\right)
(m_{t-1}^\dagger-m_{t-1}^{\pi_e})^2\right]}_{\text{variance cost of freezing}}.
\end{aligned}
\label{eq:approx-variance}
\end{equation}
\textit{(ii) Optimization guarantee.} Suppose the second term above, the variance cost, is at most $\varepsilon$ uniformly over feasible designs. If $\bm\Lambda^\ast$ minimizes the surrogate variance in part~(i), then
\[
\Var(\tilde\Gamma_T^{\pi_e}(O;\bm\Lambda^\ast,\bm\eta_m^\dagger))
-\inf_{\bm\Lambda}\Var(\tilde\Gamma_T^{\pi_e}(O;\bm\Lambda,\bm\eta_m^\dagger))
\le\varepsilon.
\]
If an estimated or further simplified variance criterion differs uniformly from the surrogate variance by at most $\delta$, its minimizer $\hat{\bm\Lambda}$ instead satisfies the same bound with $\varepsilon+2\delta$. The guarantee holds on the event of these uniform bounds and also controls the shortfall in maximized variance reduction relative to a fixed reference variance.

\textit{(iii) Approximation bound.} Suppose every feasible design satisfies $\Lambda_{1:T}\ge\underline\Lambda>0$ almost surely. Then
\begin{equation}
\begin{aligned}
0&\le\sup_{\bm\Lambda}\Bigl\{
\Var(\tilde\Gamma_T^{\pi_e}(O;\bm\Lambda,\bm\eta_m^\dagger))
-\Var(\tilde\Gamma_T^{\pi_e}(O;\bm\Lambda,\bm\eta_m))\Bigr\}
\le (\underline\Lambda^{-1}-1)
\sum_{t=1}^{T}\|m_{t-1}^\dagger-m_{t-1}^{\pi_e}\|_2^2.
\end{aligned}
\label{eq:freezing-bound}
\end{equation}
 
\end{proposition}

Our approximation posits that intermediate annotation updates $S_{t'}, t'<t$ are less informative than revealing the true conditioning state $S_t$ itself. 
$$ \E\!\left[\{\E[V_t^{\pi_e}\mid\mcF_k]-\E[V_t^{\pi_e}\mid\mcF_{k-1}]\}^{2}\right]
\le
\E\!\left[\{V_t^{\pi_e}(S_t)-\E[V_t^{\pi_e}\mid\mcF_{t-2}]\}^{2}\right], k < t-2.
$$

\begin{proof}[Proof of \Cref{prop-surrogate-variance}]
    \textbf{Proof of the proposition.}
For part~(i), set $C_{1:0}=1$. Summation by parts gives
\[
\tilde\Gamma(O;\bm\Lambda,\bm\eta_m^\dagger)-\Gamma_T^{\pi_e}
=\sum_{t=1}^{T}\left(\frac{C_{1:t}}{\Lambda_{1:t}}
-\frac{C_{1:t-1}}{\Lambda_{1:t-1}}\right)
(\Gamma_T^{\pi_e}-m_{t-1}^\dagger).
\]
Conditional on the complete trajectory, the weight increments are orthogonal, centered, and have variances $1/\Lambda_{1:t}-1/\Lambda_{1:t-1}$. Total variance and conditional squared-error decomposition around $m_{t-1}$ give part~(i), using Proposition~2 for the exact projections.

For part~(ii), nonnegative variance cost and optimality give
\[
\begin{aligned}
\Var(\tilde\Gamma(O;\bm\Lambda^\ast,\bm\eta_m^\dagger))
&\le\Var(\tilde\Gamma(O;\bm\Lambda^\ast,\bm\eta_m))+\varepsilon=\inf_{\bm\Lambda}\Var(\tilde\Gamma(O;\bm\Lambda,\bm\eta_m))+\varepsilon\\
&\le\inf_{\bm\Lambda}\Var(\tilde\Gamma(O;\bm\Lambda,\bm\eta_m^\dagger))+\varepsilon.
\end{aligned}
\]
Uniform criterion error $\delta$ adds at most $2\delta$: compare the criteria once at the selected design and once at the surrogate-optimal design.

For part~(iii), decreasing prefix probabilities imply $0\le1/\Lambda_{1:t}-1/\Lambda_{1:t-1}\le\underline\Lambda^{-1}-1$; substitute into part~(i). 
\end{proof}

\begin{remark}[Modifications for MDPs with joint states and rewards]
      Finally, in the MDP setting with joint states and rewards, when $T>2$, we require the following adjustments since $\mu_t(S_t,A_t)$ is not $\mcF_0$-measurable. We add a small correction to the estimator that restores telescoping. 
    
    \begin{equation}\tilde{\Gamma}_{T,M}^{\pi_e}=\hat{m}_0^{\pi_e}+\sum_{t=1}^{T-1} \frac{C_{1: t}}{\Lambda_{1: t}}\left(\hat{m}_t^{\pi_e}-\hat{m}_{t-1}^{\pi_e}\right)+\frac{C_{1: T}}{\Lambda_{1: T}}\left(\hat{\Gamma}_T^{\pi_e}-\hat{m}_{T-1}^{\pi_e}\right)
    \label{eqn-mdp-estimator}
    \end{equation}

  \end{remark}
\newpage
\section{Proofs of Estimation Results}\label{apx-proofs-of-estimation-results}

\begin{figure}[t]
    \centering
    \includegraphics[width=\columnwidth]{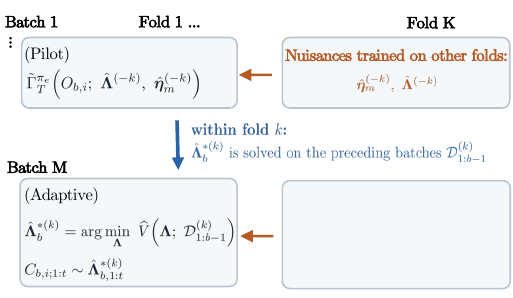}
    \caption{Batch-adaptive experiment with cross-fitting (CSBAE). Data are partitioned into $M$ batches and $K$ folds. The annotation probability for each later batch is solved using only the preceding batches of that fold. The nuisances $\hat{\bm\eta}_{m}^{(-k)}$ and $\hat{\bm\Lambda}^{(-k)}$ are trained on the other folds and plugged into $\tilde\Gamma_T^{\pi_e}$ for the units in fold $k$.}\label{fig:csbae}
\end{figure}

\paragraph{Assumptions omitted from main text}

\begin{assumption}[Bounded rewards]\label{asn-bdd-reward}
The reward satisfies $0\leq R_t \leq R_{\max}$ for all $t=1, \dots, T$.
\end{assumption}
\begin{assumption}[Bounded estimators under the batch-adaptive setting]
\label{asn-bdd_estimators_adaptive}
    The estimators are globally bounded: $c_\lambda\le\hat\lambda_{b,t}^{(k)}\le1$, $c_{\hat\Lambda}\le\hat\Lambda_{1:t}^{(-k)}\le1$, $0\le\hat\mu_{t}^{(-k)}\le C_{\hat\mu}$, and $\limsup_{n\to\infty}\|\hat f_{t}^{(-k)}\|_\infty<\infty$ for $\hat f\in\{\hat Q,\hat b,\hat{\tilde Q}\}$, $b=1,\dots,M$, $k=1,\dots,K$, $t=1,\dots,T$.
\end{assumption}

\subsection{Asymptotic normality under the non-adaptive experiment}

We begin with the non-adaptive setting, in which the annotation probability remains fixed across batches, i.e., $\Lambda_{b,1:t}=\Lambda_{1:t}$ for all $b$, because it serves as an oracle benchmark that the batch-adaptive estimator seeks to achieve asymptotically. In this setting, the feasible estimator takes the form of
\begin{equation}\label{eqn:pooled_non_adaptive}
    \hat\psi_{\mathrm{nad}}=\frac{1}{n}\sum_{k=1}^{K}\sum_{i\in\mathcal{D}^{(k)}}
    \tilde\Gamma_T^{\pi_e}\Big(O_{i};\hat{\bm\Lambda}^{(-k)},\hat{\bm\eta}^{(-k)}_{m}\Big),
\end{equation}
where $\tilde\Gamma_T^{\pi_e}$ is the estimator in \eqref{eqn-estimator} and $\mathcal{D}^{(k)}$ denotes the data in fold $k$. Asymptotic normality holds for $\hat\psi_{\mathrm{nad}}$ under the following assumptions.

\begin{assumption}[Bounded estimators]\label{asn-bdd-estimator}
The estimators are globally bounded: $c_\lambda\le\hat\lambda_t^{(-k)}\le1$, $0\le\hat\mu_t^{(-k)}\le C_{\hat\mu}$, and $\limsup_{n\to\infty} \|\hat f_t^{(-k)}\|_\infty < \infty$ for $\hat f \in \{\hat Q, \hat b, \hat{\tilde{Q}}\}$, $k=1,\dots, K$, $t=1,\dots, T$. 
\end{assumption}
\begin{remark}
    The \cref{asn-bdd-estimator} ensures that every projection estimates $\hat m_t$ is uniformly bounded.
\end{remark}

\begin{assumption}[Rate assumption]\label{asn-rate}
With the conventions $\Delta Q_{T+1}^{(-k)}=\Delta\tilde Q_{T+1}^{(-k)}=\Delta\mu_{T+1}^{(-k)}:=0$, the following product error rates hold for all $1\leq k\leq K$ and $1\leq t\leq T$:
\begin{align*}
    &\textup{(R1)}\qquad \|\Delta\mu_t^{(-k)}\|_2\Big(\|\Delta Q_t^{(-k)}\|_2+\|\Delta Q_{t+1}^{(-k)}\|_2\Big)=o_p(n^{-1/2}),\\
    &\textup{(R2)}\qquad \sum_{j\leq t}\|\Delta\lambda_j^{(-k)}\|_2\Big(\|\Delta b_t^{(-k)}\|_2+\|\Delta Q_{t+1}^{(-k)}\|_2 +\|\Delta\mu_t^{(-k)}\|_2+\|\Delta\mu_{t+1}^{(-k)}\|_2\Big)=o_p(n^{-1/2}).
\end{align*}
We refer to (R1) as the double reinforcement learning (DRL) product error and to (R2) as the annotation product error.
\end{assumption}
\begin{theorem}[Asymptotic normality under the non-adaptive setting]\label{thm:clt_non-adaptive}
Suppose \cref{asn-seq-igno-stilde,asn-anno-ignorability,asn-seq_overlap,asn-bdd-estimator,asn-rate} and bounded rewards (\Cref{asn-bdd-reward}) hold, together with $\|\hat\lambda_t^{(-k)}-\lambda_t\|_2=o_p(1)$, $\|\hat\mu_t^{(-k)}-\mu_t\|_2=o_p(1)$, $\|\hat Q_t^{(-k)}-Q_t\|_2=o_p(1)$, $\|\hat b_t^{(-k)}-b_t\|_2=o_p(1)$, and $\|\hat{\tilde Q}_t^{(-k)}-\tilde Q_t\|_2=o_p(1)$ for all $1\le k\le K$ and $1\le t\le T$. Then
\[
    \sqrt{n}\bigl(\hat\psi_{\mathrm{nad}}-\Phi^{\pi_e}\bigr)
    \Longrightarrow
    \mathcal N(0,\sigma^2), \qquad \text{where} \qquad \sigma^2 = \Var(m_0) + \sum_{t=1}^{T}
\E\left[ \frac{\Var(m_t\mid\mcF_{t-1})}{\Lambda_{1:t}} \right],
\]
and $\sigma^2 \in(0,\infty)$.
\end{theorem}

\subsection{Proofs of \Cref{thm:clt_non-adaptive}} \label{pf-clt_non-adaptive}

Recall that the non-adaptive estimator can be written as 
\[
    \hat\psi_{\mathrm{nad}}=\sum_{k=1}^{K}\frac{n_k}{n}\cdot \frac{1}{n_k}\sum_{i\in\mathcal{D}^{(k)}}
    \tilde\Gamma\left(O_{i};\hat{\bm{\Lambda}}^{(-k)},\hat{\bm{\eta}}^{(-k)}_{m}\right),
\]
where $n_k=|\mathcal{D}^{(k)}|$, thus $\sum_{k=1}^{K} n_k=n$. Therefore,
\begin{align*}
    \sqrt{n}(\hat{\psi}_{\mathrm{nad}}-\Phi^{\pi_e})=\sum_{k=1}^{K} \sqrt{\frac{n_k}{n}}\cdot \sqrt{n_k}\left(\E_{n_k}\left[\tilde\Gamma(O; \hat{\bm\Lambda}^{(-k)}, \hat{\bm\eta}_m^{(-k)})\right]-\Phi^{\pi_e}\right),
\end{align*}
where we use $\E_{n_k}$ to refer empirical mean over the data in $\mathcal{D}^{(k)}$, for $k=1,\cdots, K$.
Then it suffices to probe the asymptotic behavior of the $k$th fold term, as the same logic applies to the other terms. Note that the $k$th term allows the following decomposition
\begin{align}
    \sqrt{n_k}\left(\E_{n_k}[\tilde{\Gamma} (O; \hat{\bm\Lambda}^{(-k)}, \hat{\bm\eta}_m^{(-k)})] -\Phi^{\pi_e}\right)  &=   \underbrace{\sqrt{n_k} (\E_{n_k}-\E) \left[ \tilde{\Gamma}(O; \hat{\bm\Lambda}^{(-k)}, \hat{\bm\eta}_m^{(-k)}) - \tilde{\Gamma}(O; \bm\Lambda, \bm\eta_m) \big\vert \mathcal{D}^{(-k)} \right]}_{\text{empirical process term}}\label{eqn:term_empricial_process_non_adaptive}\\
    &+ \underbrace{ \sqrt{n_k} (\E_{n_k}-\E) \left[ \tilde{\Gamma}(O; \bm\Lambda, \bm\eta_m)\big\vert \mathcal{D}^{(-k)}\right]}_{\text{oracle term}} \label{eqn:term_oracle_non_adaptive}\\
    &+ \underbrace{ \sqrt{n_k} \left( \E[\tilde{\Gamma} (O; \hat{\bm\Lambda}^{(-k)}, \hat{\bm\eta}_m^{(-k)})\big\vert \mathcal{D}^{(-k)}] -\Phi^{\pi_e} \right)}_{\text{drift term}},
    \label{eqn:drift_non_adaptive}
\end{align}
We established the asymptotic behavior of the empirical process, oracle, and drift terms in \cref{prop-emp-process-term-non-adaptive,prop-oracle-term-non-adaptive,prop-drift-term-non-adaptive}, respectively. \Cref{prop-emp-process-term-non-adaptive} shows that the empirical process term is $o_p(1)$ via Chebyshev's inequality, by conditioning on $\mathcal D^{(-k)}$ and showing that the conditional mean of the term is zero while its conditional variance goes to zero in probability. \Cref{prop-oracle-term-non-adaptive} shows that the oracle term converges in distribution to $\mathcal N(0,\sigma^2)$ by the central limit theorem. \Cref{prop-drift-term-non-adaptive} shows that the drift term is $o_p(1)$ under the assumed product error rates. Combining these three results yields
\[
    \sqrt{n_k}\left(\E_{n_k}[\tilde{\Gamma} (O; \hat{\bm\Lambda}^{(-k)}, \hat{\bm\eta}_m^{(-k)})] -\Phi^{\pi_e}\right) \Longrightarrow \mathcal{N}(0, \sigma^2)
\]
by Slutsky's theorem. Since each fold is independent, combining all terms together yields
\[
    \sqrt{n}(\hat{\psi}_{\mathrm{nad}}-\Phi^{\pi_e})=\sum_{k=1}^{K} \sqrt{\frac{n_k}{n}}\cdot \sqrt{n_k}\left(\E_{n_k}\left[\tilde\Gamma(O; \hat{\bm\Lambda}^{(-k)}, \hat{\bm\eta}_m^{(-k)})\right]-\Phi^{\pi_e}\right) \Longrightarrow\mathcal{N}\left(0,\sum_{k=1}^{K}\frac{n_k}{n}\sigma^2\right) =\mathcal{N}(0, \sigma^2).
\]
Before delving into the proofs of these three propositions, we introduce useful decompositions and reduction that appear frequently throughout the proofs.


\begin{lemma}[Useful decompositions of $\tilde\Gamma(O;\hat{\bm\Lambda},\hat{\bm\eta}_m)-\tilde\Gamma( O;\bm\Lambda,\bm\eta_m)$]\label{lem:diff_decomposition}
Let $\Delta(\cdot)=\hat{(\cdot)}-(\cdot)$ be the nuisance estimation error. Then, the difference between the
feasible estimator and the feasible oracle estimator decomposes as
\begin{equation*}\label{eqn:decomp-top}
\begin{aligned}
    \tilde\Gamma(O;\hat{\bm\Lambda},\hat{\bm\eta}_m)-\tilde\Gamma( O;\bm\Lambda,\bm\eta_m)
    =\Delta m_0
    &-\sum_{t=1}^{T}\frac{C_{1:t}}{\hat{\Lambda}_{1:t}\Lambda_{1:t}}\Delta\Lambda_{1:t} \Delta\big(m_t-m_{t-1}\big)
    -\sum_{t=1}^{T}\frac{C_{1:t}}{\hat{\Lambda}_{1:t}\Lambda_{1:t}}\Delta\Lambda_{1:t} \big(m_t-m_{t-1}\big) \\
    &+\sum_{t=1}^{T}\frac{C_{1:t}}{\Lambda_{1:t}}\Delta\big(m_t-m_{t-1}\big).
\end{aligned}
\end{equation*}
The three summands $\Delta m_0$, $(m_t-m_{t-1})$, and $\Delta(m_t-m_{t-1})$ decompose further as
follows.\\

\noindent(i) Full imputation estimation error $\Delta m_0$:
\begin{equation*}\label{eqn:decomp-m0} 
\begin{aligned}
    \Delta m_0
    =\Delta V_1
    &+\sum_{t=1}^{T}\gamma^{t-1}
        \Delta\mu_t\big(\Delta b_t+\gamma\Delta \tilde V_{t+1}-\Delta\tilde Q_{t}\big)
+\sum_{t=1}^{T}\Delta\mu_t\big(b_t+\gamma\tilde V_{t+1}-\tilde Q_t\big) \\
    &+\sum_{t=1}^{T} \mu_t\big(\Delta b_t+\gamma\Delta\tilde V_{t+1}-\Delta\tilde Q_t\big), \\
\end{aligned}
\end{equation*}

\noindent(ii) Oracle projection increment $m_t-m_{t-1}$:
\begin{equation*}\label{eqn:decomp-increment} 
\begin{aligned}
    (m_t-m_{t-1})
    &=\gamma^{t-1}\mu_t\big\{(R_t-b_t)+\gamma(V_{t+1}-\tilde V_{t+1})\big\} 
    +\gamma^{t}\,\mathbf{1}\{t<T\}\,\mu_{t+1}\big(\tilde Q_{t+1}-Q_{t+1}\big), \\
\end{aligned}
\end{equation*}

\noindent(iii) Increment estimation error $\Delta (m_t-m_{t-1})$:
\begin{equation*} \label{eqn:decomp-increment-error}
\begin{aligned} 
    \Delta(m_t-m_{t-1})
    &=\gamma^{t-1}\Delta\mu_t\big(-\Delta b_t +\gamma\Delta V_{t+1}-\gamma\Delta\tilde V_{t+1}\big) 
    +\gamma^{t-1}\Delta\mu_t\big\{(R_t-b_t)+\gamma(V_{t+1}-\tilde V_{t+1})\big\} 
    \\
    &
    +\gamma^{t-1}\mu_t(-\Delta b_t+\gamma\Delta V_{t+1}-\gamma\Delta\tilde V_{t+1}) 
    +\mathbf{1}\{t<T\} \gamma^{t}
        \Delta\mu_{t+1}\big(\Delta\tilde Q_{t+1}-\Delta Q_{t+1}\big) 
        \\
    &
    +\mathbf{1}\{t<T\}\gamma^t \Delta\mu_{t+1}\big(\tilde Q_{t+1}-Q_{t+1}\big)
    +\mathbf{1}\{t<T\}\gamma^t \mu_{t+1}\big(\Delta\tilde Q_{t+1}-\Delta Q_{t+1}\big).
\end{aligned}
\end{equation*}
\end{lemma}

\begin{lemma}[Useful decomposition of the mean of $\tilde \Gamma(O;\hat{\bm\Lambda},\hat{\bm\eta}_m)-\tilde\Gamma(O; \bm\Lambda, \bm\eta_m)$] \label{lem:mean_diff_decomposition}
The mean of the difference between the feasible estimator and the feasible oracle estimator decomposes as
\begin{equation} \label{eqn:drift_mean}
    \begin{aligned}
        \E[\tilde{\Gamma}(O;\hat{\bm\Lambda},\hat{\bm\eta}_m) - \tilde\Gamma(O;\bm\Lambda, \bm\eta_m)] &= \sum_{t=1}^{T}\E\left[\gamma^{t-1}\Delta\mu_t(\gamma\Delta V_{t+1}-\Delta Q_t)\right]\\
        &+\sum_{t=1}^{T} \gamma^{t-1} \E\left[\frac{C_{1:t}}{\hat{\Lambda}_{1:t}\Lambda_{1:t}}\Delta\Lambda_{1:t} \Delta\mu_t(\Delta b_t -\gamma \Delta V_{t+1} + \gamma\Delta \tilde{V}_{t+1})\right]\\
        &+\sum_{t=1}^{T}\gamma^{t-1}\E\left[\frac{ C_{1:t}}{\hat{\Lambda}_{1:t}\Lambda_{1:t}}\Delta\Lambda_{1:t}\mu_t(\Delta b_t -\gamma \Delta V_{t+1}+\gamma\Delta \tilde{V}_{t+1})\right]\\
        &+\sum_{t=1}^{T}\gamma^{t-1}\E\left[\frac{C_{1:t}}{\hat{\Lambda}_{1:t}\Lambda_{1:t}}\Delta\Lambda_{1:t} \Delta\mu_t\{(b_t - R_t)+\gamma(\tilde{V}_{t+1}- V_{t+1})\}\right]  \\
        &+\sum_{t=1}^{T-1}\gamma^t\E\left[\frac{ C_{1:t}}{\hat{\Lambda}_{1:t}\Lambda_{1:t}}\Delta\Lambda_{1:t}\Delta \mu_{t+1}(\Delta Q_{t+1}- \Delta\tilde{Q}_{t+1})\right] \\
        &+\sum_{t=1}^{T-1}\gamma^t\E\left[\frac{ C_{1:t}}{\hat{\Lambda}_{1:t}\Lambda_{1:t}}\Delta\Lambda_{1:t}\mu_{t+1}(\Delta Q_{t+1}- \Delta \tilde Q_{t+1})\right]  \\
        &+\sum_{t=1}^{T-1}\gamma^{t} \E\left[\frac{C_{1:t}}{\hat{\Lambda}_{1:t}\Lambda_{1:t}}\Delta\Lambda_{1:t} \Delta \mu_{t+1}(Q_{t+1}- \tilde Q_{t+1})\right] . 
    \end{aligned}
\end{equation}
\end{lemma}

\begin{lemma}[Nuisance error reductions, from $V$ to $Q$]\label{lem-v-q-error}
For $t=1,\dots, T$, the following inequality holds
\begin{align*} 
         &\|\Delta V_t^{(-k)}\|_{2}^2 \leq \epsilon_b^{-1}\|\Delta Q_t^{(-k)}\|_{2}^2,\\
         & \|\Delta \tilde V_{t}^{(-k)}\|_{2}^2 \leq \epsilon_b^{-1} \|\Delta Q_{t}^{(-k)}\|_{2}^{2}, \\
        &  \|\Delta\tilde Q_t^{(-k)}\|_{2}^{2} \leq  \|\Delta Q_t^{(-k)}\|_{2}^{2}.
\end{align*}
\end{lemma}
\noindent We defer the proofs of these lemmas to the end of the section. Building on these lemmas, the following three propositions articulate the asymptotic behavior of the three terms.

\begin{proposition}[Empirical process term]\label{prop-emp-process-term-non-adaptive}
    The empirical process term is $o_p(1)$, i.e., 
    \[
        \sqrt{n_k} (\E_{n_k}-\E) \left[ \tilde{\Gamma}(O; \hat{\bm\Lambda}^{(-k)}, \hat{\bm\eta}_m^{(-k)}) - \tilde{\Gamma}(O; \bm\Lambda, \bm\eta_m) \big\vert \mathcal{D}^{(-k)} \right]=o_p(1).   
    \]
\end{proposition}
\begin{proof}[Proof of \Cref{prop-emp-process-term-non-adaptive}]
    Let $\varphi= \tilde{\Gamma}(O; \hat{\bm\Lambda}^{(-k)}, \hat{\bm\eta}_m^{(-k)}) - \tilde{\Gamma}(O; \bm\Lambda, \bm\eta_m)$. Then it suffices to show that for every $\epsilon>0$,
    \begin{equation} \label{eqn:conv_ip}
         \lim_{\substack{n\to \infty \\ n_k/n\to1/K}} P\big[\,|\sqrt{n_k}(\E_{n_k}[\varphi]-\E[\varphi \mid \mathcal{D}^{(-k)}])|>\epsilon \mid \mathcal D^{(-k)}\,\big]=0
        \quad\text{in probability}.
    \end{equation}
    The reason is that LHS is bounded ($[0,1]$-valued), thus the condition \eqref{eqn:conv_ip} implies
    \[
         \lim_{\substack{n\to \infty, \\ n_k/n\to1/K}} P\big[|\sqrt{n_k}(\E_{n_k}[\varphi]-\E[\varphi\mid\mathcal{D}^{(-k)}])|>\epsilon\big]
         =0
    \]
    by the bounded convergence theorem, concluding the proof. To prove \eqref{eqn:conv_ip}, we show that the conditional mean is zero and the conditional variance is $o_p(1)$. \\
    
    \noindent First, note that the conditional mean of $\sqrt{n_k}(\E_{n_k}[\varphi]-\E[\varphi \mid \mathcal{D}^{(-k)}])$ is zero, since $\hat{\bm\Lambda}^{(-k)}$ and $\hat{\bm \eta}_m^{(-k)}$ is fixed given $\mathcal{D}^{(-k)}$, i.e.,
    \[
        \E[\sqrt{n_k}(\E_{n_k}[\varphi] - \E[\varphi \mid\mathcal{D}^{(-k)}]) \mid \mathcal{D}^{(-k)}]=0.
    \]
    Next, as a consequence of conditional mean being zero, its conditional variance equals its conditional second moment, which in turn reduces to the conditional variance of $\varphi$ itself:
    \begin{align*}
        \Var\big[\sqrt{n_k}\big(\E_{n_k}[\varphi]-\E[\varphi\mid\mathcal D^{(-k)}]\big)\mid \mathcal D^{(-k)}\big]
        &= \E\big[\big\{\sqrt{n_k}\big(\E_{n_k}[\varphi]-\E[\varphi\mid\mathcal D^{(-k)}]\big)\big\}^2\mid \mathcal D^{(-k)}\big] \\
        &= n_k\cdot\E\big[\big(\E_{n_k}[\varphi]-\E[\varphi\mid\mathcal D^{(-k)}]\big)^2 \mid \mathcal D^{(-k)}\big] \\
        &= \Var\big[\varphi\mid \mathcal D^{(-k)}\big],
    \end{align*}
    where the last equality holds because the $n_k$ evaluation-fold observations are i.i.d.\ given $\mathcal D^{(-k)}$, so that $\Var[\E_{n_k}[\varphi]\mid\mathcal D^{(-k)}] = \Var[\varphi\mid\mathcal D^{(-k)}]/n_k$.
    Using the slight modification of the $\varphi$-decomposition derived in \cref{lem:diff_decomposition}:
    \begin{align*}
        \varphi
        =\Delta m_0^{(-k)}
        +\sum_{t=1}^{T}\frac{C_{1:t}}{\hat{\Lambda}_{1:t}^{(-k)}\Lambda_{1:t}}\Delta\Lambda_{1:t}^{(-k)}\big(\hat m_t^{(-k)}-\hat m_{t-1}^{(-k)}\big) +\sum_{t=1}^{T}\frac{C_{1:t}}{\Lambda_{1:t}}\Delta\big(m_t^{(-k)}-m_{t-1}^{(-k)}\big),
    \end{align*}
    one can deduce the upper bound of the conditional variance as follows:
    \begin{align*}
        \Var[\varphi\mid \mathcal{D}^{(-k)}]
        &\le \E\left[\varphi^2 \big\vert \mathcal{D}^{(-k)}\right] \\
        &\leq 3\E\bigl[(\Delta m_0^{(-k)})^2\mid\mathcal{D}^{(-k)}\bigr]\\
        &+ 3T \sum_{t=1}^{T}\E\left[\frac{C_{1:t}}{\bigl(\hat\Lambda_{1:t}^{(-k)}\Lambda_{1:t}\bigr)^2}\bigl(\Delta\Lambda_{1:t}^{(-k)}\bigr)^2(\hat m_t^{(-k)}-\hat m_{t-1}^{(-k)})^2 \bigg\vert \mathcal{D}^{(-k)}\right] \\
        &+3T\sum_{t=1}^{T}\E\left[\frac{C_{1:t}}{\Lambda_{1:t}^{2}}\bigl(\Delta(m_t^{(-k)}-m_{t-1}^{(-k)})\bigr)^2\bigg\vert \mathcal{D}^{(-k)}\right],
    \end{align*}
    where the second inequality is due to $(a+b+c)^2 \leq 3(a^2+b^2+c^2)$ and the Cauchy--Schwarz inequality. Since the oracle annotation probability $\Lambda_{1:t}$, the increment estimation error $m_t-m_{t-1}$ and their estimates $\hat\Lambda_{1:t}^{(-k)}$, $(\hat m_t^{(-k)}-m_{t-1}^{(-k)})$ are bounded by \cref{asn-seq_overlap,asn-bdd-estimator}, one can simplify the RHS as 
    \begin{align*}
        Var[\varphi\mid \mathcal{D}^{(-k)}]&=\mathcal{O}\left(\E\bigl[(\Delta m_0^{(-k)})^2 \mid \mathcal{D}^{(-k)} \bigr]+ \sum_{t=1}^{T}\E[(\Delta\Lambda_{1:t}^{(-k)})^2\mid \mathcal{D}^{(-k)}]+\sum_{t=1}^{T}\E\left[(\Delta(m_t^{(-k)}-m_{t-1}^{(-k)}))^2\mid\mathcal{D}^{(-k)}\right]\right).
    \end{align*}
    We bound each terms in order, entirely at the level of the primitive nuisances. \\

    \noindent\emph{(i) $\E\bigl[(\Delta m_{0}^{(-k)})^2\mid \mathcal{D}^{(-k)}\bigr]$}: Using the slight modification of $\Delta m_0^{(-k)}$--decomposition derived in \cref{lem:diff_decomposition}
    \begin{align*}
        \Delta m_0^{(-k)}&=\Delta V_1^{(-k)}+\sum_{t=1}^{T}\gamma^{t-1}\Delta \mu_t^{(-k)}(\hat b_t^{(-k)}+\gamma\tilde V_{t+1}^{(-k)}-\tilde Q_t^{(-k)}) \\
        &\qquad \ \quad + \sum_{t=1}^{T}\gamma^{t-1}\mu_t(\Delta b_t^{(-k)}+\gamma\Delta\tilde V_{t+1}^{(-k)}-\Delta\tilde Q_t^{(-k)}),
    \end{align*}
    applying the inequality $(\sum_{i=1}^{k}a_i)^2\leq k(\sum_{i=1}^{k}a_i^2)$ and the Cauchy-Schwarz inequality over the sum yields
    \begin{align*}
        \E\Bigl[\bigl(\Delta m_0^{(-k)}\bigr)^2\mid \mathcal{D}^{(-k)}\Bigr] &\leq 3\E\left[\bigl(\Delta V_1^{(-k)}\bigr)^2\mid \mathcal{D}^{(-k)}\right] \\
        &+3T\sum_{t=1}^{T}\gamma^{2(t-1)}\E\left[\bigl(\Delta\mu_t^{(-k)}\bigr)^2(\hat b_t^{(-k)}+\gamma\tilde V_{t+1}^{(-k)}-\tilde Q_t^{(-k)})^2\mid \mathcal{D}^{(-k)}\right] \\
        &+9T\sum_{t=1}^{T}\gamma^{2(t-1)}\E\left[\mu_t^2\bigl(|\Delta b_t^{(-k)}|^2+\gamma|\Delta\tilde V_{t+1}^{(-k)}|^2 +|\Delta \tilde Q_t^{(-k)}|^2
        \bigr)\mid \mathcal{D}^{(-k)}\right].
    \end{align*}
    Since $b_t^{(-k)}+\gamma\tilde V_{t+1}^{(-k)}-\tilde Q_t^{(-k)}$ and $\mu_t$ are bounded above by \cref{asn-bdd-estimator,asn-seq_overlap}, one can simplify the RHS as
    \begin{align*}
        \E\bigl[(\Delta m_0^{(-k)})^2\mid \mathcal{D}^{(-k)}\bigr] 
        &=\mathcal{O}\left(\E\bigl[(\Delta V_1^{(-k)})^2\mid \mathcal{D}^{(-k)}\bigr]+\sum_{t=1}^{T}\E\bigl[(\Delta \mu_t^{(-k)})^2\mid \mathcal{D}^{(-k)}\bigr]\right. \\
        &\left.\qquad\qquad+\sum_{t=1}^{T}\E\bigl[(\Delta b_t^{(-k)})^2+(\Delta \tilde V_{t+1}^{(-k)})^2+(\Delta\tilde Q_t^{(-k)})^2\mid \mathcal{D}^{(-k)}\bigr]\right).
    \end{align*}
    By the reduction from \cref{lem-v-q-error}, the RHS further reduces to
    \begin{align*}
        \E\bigl[(\Delta m_0^{(-k)})^2\mid \mathcal{D}^{(-k)}\bigr] 
        &=\mathcal{O}\left(\E\bigl[(\Delta Q_1^{(-k)})^2\mid \mathcal{D}^{(-k)}\bigr]+\sum_{t=1}^{T}\E\bigl[(\Delta \mu_t^{(-k)})^2\mid \mathcal{D}^{(-k)}\bigr]\right. \\
        &\left.\qquad\qquad+\sum_{t=1}^{T}\E\bigl[(\Delta b_t^{(-k)})^2+(\Delta Q_{t+1}^{(-k)})^2+(\Delta Q_t^{(-k)})^2\mid \mathcal{D}^{(-k)}\bigr]\right),
    \end{align*}
    which is $o_p(1)$ by the $\mathcal{L}_2$-consistency hypothesis. \\

    \noindent\emph{(ii) $\E\bigl[(\Delta\Lambda_{1:t}^{(-k)})^2 \mid \mathcal{D}^{(-k)}\bigr]$}: Note that
    \[
        \E\bigl[(\Delta\Lambda_{1:t}^{(-k)})^2\mid\mathcal{D}^{(-k)}\bigr]\leq \E\left[\left(\sum_{t=1}^{T}\Delta\lambda_t^{(-k)}\right)^2 \bigg\vert\mathcal{D}^{(-k)}\right] \leq T\sum_{t=1}^{T}\E\left[(\Delta\lambda_{t}^{(-k)})^2 \mid \mathcal{D}^{(-k)}\right]=o_p(1)
    \]
    where the first inequality is due to \cref{lem:prod_comparison}, the second inequality is due to the Cauchy-Schwarz inequality and the final equality holds under the consistency hypothesis. \\

    \noindent\emph{(iii) $\sum_{t=1}^{T}\E\left[(\Delta(m_t^{(-k)}-m_{t-1}^{(-k)}))^2\mid\mathcal{D}^{(-k)}\right]$:} Based on the decomposition below derived in \cref{lem:diff_decomposition}
    \begin{align*}
        \Delta(m_t^{(-k)}-m_{t-1}^{(-k)}) &= \gamma^{t-1}\Delta\mu_t^{(-k)}(-\Delta b_t^{(-k)}+\gamma\Delta V_{t+1}^{(-k)}-\gamma\Delta\tilde V_{t+1}^{(-k)}) \\
        &+\gamma^{t-1}\Delta \mu_t^{(-k)}\{(R_t-b_t)+\gamma(V_{t+1}-\tilde V_{t+1})\} \\
        &+\gamma^{t-1}\mu_t(-\Delta b_t^{(-k)}+\gamma\Delta V_{t+1}^{(-k)}-\gamma\Delta\tilde V_{t+1}^{(-k)}) \\
        &+\mathbf{1}\{t<T\}\gamma^{t}\Delta\mu_{t+1}^{(-k)}(\Delta\tilde Q_{t+1}^{(-k)}-\Delta Q_{t+1}^{(-k)}) \\
        &+\mathbf{1}\{t<T\}\gamma^t\Delta\mu_{t+1}^{(-k)}(\tilde Q_{t+1}-Q_{t+1}) \\
        &+\mathbf{1}\{t<T\}\gamma^t\mu_{t+1}(\Delta\tilde Q_{t+1}^{(-k)}-\Delta Q_{t+1}^{(-k)}),
    \end{align*}
    one can deduce the following by applying $(\sum_{i=1}^{k}a_i)^2 \leq k\sum_{i=1}^{k}a_i^2$ and the Cauchy--Schwarz inequality:
    \begin{align*}
        \sum_{t=1}^{T}\E\bigl[(\Delta(m_t^{(-k)}-m_{t-1}^{(-k)})^2 \mid \mathcal{D}^{(-k)}\bigr]&\leq 18\sum_{t=1}^{T}\E\bigl[|\Delta\mu_t^{(-k)}|^2\bigl(|\Delta b_t^{(-k)}|^2+|\Delta V_{t+1}|^2+|\Delta \tilde V_{t+1}|^2\bigr)\mid\mathcal{D}^{(-k)}\bigr] \\
        &+6\sum_{t=1}^{T}\E\bigl[\{(R_t-b_t)+\gamma(V_{t+1}-\tilde V_{t+1})\}^2|\Delta\mu_t^{(-k)}|^2\mid \mathcal{D}^{(-k)}\bigr] \\
        &+18\sum_{t=1}^{T}\E\bigl[\mu_t^2\bigl(|\Delta b_t^{(-k)}|^2+|\Delta V_{t+1}^{(-k)}|^2+|\Delta \tilde V_{t+1}^{(-k)}|^2\bigr)\mid \mathcal{D}^{(-k)}\bigr] \\
        &+12\sum_{t=1}^{T-1}\E\bigl[|\Delta \mu_{t+1}^{(-k)}|^2\bigl(|\Delta \tilde Q_{t+1}^{(-k)}|^2+|\Delta Q_{t+1}^{(-k)}|^2\bigr)\mid \mathcal{D}^{(-k)}\bigr] \\
        &+6\sum_{t=1}^{T-1}\E\bigl[(\tilde Q_{t+1}-Q_{t+1})^2|\Delta \mu_{t+1}^{(-k)}|^2\mid \mathcal{D}^{(-k)}\bigr] \\
        &+12\sum_{t=1}^{T-1}\E\bigl[\mu_{t+1}^2 \bigl(|\Delta \tilde Q_{t+1}^{(-k)}|^2+|\Delta Q_{t+1}^{(-k)}|^2\bigr)\mid \mathcal{D}^{(-k)}\bigr].
    \end{align*}
    Peeling the uniformly bounded factors under the boundness assumption as before, we have
    \begin{align*}
        \sum_{t=1}^{T}\E\left[(\Delta(m_t^{(-k)}-m_{t-1}^{(-k)}))^2\mid \mathcal{D}^{(-k)}\right] &=\mathcal{O}\left(\sum_{t=1}^{T}\E[(\Delta\mu_t^{(-k)})^2\mid\mathcal{D}^{(-k)}]\right.\\
        &\qquad \left.+\sum_{t=1}^{T}\E[(\Delta b_t^{(-k)})^2+(\Delta V_{t+1}^{(-k)})^2+(\Delta \tilde V_{t+1}^{(-k)})^2\mid\mathcal{D}^{(-k)}]\right. \\
        &\qquad\left. +\sum_{t=1}^{T}\E[(\Delta\tilde Q_{t+1}^{(-k)})^2+(\Delta Q_{t+1}^{(-k)})^2 \mid \mathcal{D}^{(-k)}]\right).
    \end{align*}
    The RHS further simplifies under the reduction from \cref{lem-v-q-error} as
    \begin{align*}
        \sum_{t=1}^{T}\E\left[(\Delta(m_t^{(-k)}-m_{t-1}^{(-k)}))^2\mid \mathcal{D}^{(-k)}\right] &=\mathcal{O}\left(\sum_{t=1}^{T}\E[(\Delta\mu_t^{(-k)})^2 +(\Delta b_t^{(-k)})^2+(\Delta Q_{t+1}^{(-k)})^2\mid\mathcal{D}^{(-k)}]\right),
    \end{align*}
    which is $o_p(1)$ by the $\mathcal{L}_2$-consistency hypothesis.\\

    \noindent Combining \emph{(i)--(iii)}, each of which is $o_p(1)$ under the consistency hypotheses of \cref{thm:clt_non-adaptive}, gives $\E[\varphi^2\mid\mathcal D_2]=o_p(1)$. By the conditional Chebyshev inequality,
    \[
        P\big[\,|\sqrt{n_k}(\E_{n_k}[\varphi] - \E[\varphi\mid \mathcal{D}^{(-k)}])| > \epsilon \mid \mathcal{D}^{(-k)}\big]
        \leq \frac{\Var[\sqrt{n_k}(\E_{n_k}-\E)[\varphi]\mid \mathcal{D}^{(-k)}]}{\epsilon^2}
        \le \frac{\E[\varphi^2 \mid \mathcal{D}^{(-k)}]}{\epsilon^2}=o_p(1),
    \]
    which is the conditional statement \eqref{eqn:conv_ip} required above. This concludes the proof.
\end{proof}

\begin{proposition}[Oracle term]\label{prop-oracle-term-non-adaptive}
The oracle term converges to $\mathcal{N}(0, \sigma^2)$ in distribution, where $\sigma^2=\Var(m_0)+\sum_{t=1}^{T}\E[\Var(m_t\mid\mcF_{t-1})/\Lambda_{1:t}]$, i.e., 
\[
    \sqrt{n_k}(\E_{n_k}-\E)\left[\tilde \Gamma(O; \bm\Lambda, \bm\eta_m) \mid \mathcal{D}^{(-k)}\right] \Longrightarrow \mathcal{N}(0, \sigma^2), \quad \sigma^2=\Var(m_0)+\sum_{t=1}^{T}\E\left[\frac{\Var(m_t\mid\mcF_{t-1})}{\Lambda_{1:t}}\right].
\]
\end{proposition}
\begin{proof}[Proof of \cref{prop-oracle-term-non-adaptive}]
    Note that $\tilde\Gamma(O;\bm\Lambda, \bm\eta_m)$ is unbiased due to \cref{prop:seq-annotation-unbiased}:
    \begin{align*}
        \E[\tilde \Gamma(O;\bm\Lambda, \bm\eta_m)\mid \mathcal{D}^{(-k)}]&=\E[\tilde \Gamma(O;\bm\Lambda, \bm\eta_m)]=\Phi^{\pi_e}.
    \end{align*}
    Then, the summands $\tilde\Gamma(O_i;\bm\Lambda, \bm\eta_m)$, $i\in\mathcal{D}_1$ are i.i.d. and uniformly bounded, the CLT yields
    \[
        \sqrt{n_k}(\E_{n_k}-\E)\left[\tilde \Gamma(O; \bm\Lambda, \bm\eta_m) \mid \mathcal{D}^{(-k)}\right] \Longrightarrow \mathcal{N}(0, \sigma^2), \quad \sigma^2=\Var(m_0)+\sum_{t=1}^{T}\E\left[\frac{\Var(m_t\mid\mcF_{t-1})}{\Lambda_{1:t}}\right].
    \]
\end{proof}

\begin{proposition}[Drift term]\label{prop-drift-term-non-adaptive}
The drift term is $o_p(1)$, i.e.,
\begin{align*}
    \sqrt{n_k} \left( \E[\tilde{\Gamma} (O; \hat{\bm\Lambda}^{(-k)}, \hat{\bm\eta}_m^{(-k)})\big\vert \mathcal{D}^{(-k)}] -\Phi^{\pi_e} \right)=o_p(1).
\end{align*}
\end{proposition}
\begin{proof}[Proof of \cref{prop-drift-term-non-adaptive}] 
    Since $\tilde\Gamma(O;\bm\Lambda,\bm\eta_m)$ is unbiased (\cref{prop:seq-annotation-unbiased}) and is independent of $\mathcal{D}^{(-k)}$, one can rewrite the drift term as follows:
    \[
        \sqrt{n_k}\left(\E[\tilde\Gamma(O;\hat{\bm\Lambda}^{(-k)},\hat{\bm\eta}_m^{(-k)}\mid\mathcal{D}^{(-k)}]-\Phi^{\pi_e}\right)=\sqrt{n_k}\left(\E[\tilde\Gamma(O;\hat{\bm\Lambda}^{(-k)},\hat{\bm\eta}_m^{(-k)})-\tilde\Gamma(O;\bm\Lambda,\bm\eta_m)\mid\mathcal{D}^{(-k)}]\right).
    \]
    This allows to bring the decomposition derived in \cref{lem:mean_diff_decomposition}, which gives rise to
    \begin{align}
        &\sqrt{n_k}\left(\E[\tilde\Gamma(O;\hat{\bm\Lambda}^{(-k)},\hat{\bm\eta}_m^{(-k)})-\tilde\Gamma(O;\bm\Lambda,\bm\eta_m)\mid\mathcal{D}^{(-k)}]\right) \notag \\
        &\qquad=\sum_{t=1}^{T}\sqrt{n_k}\left(\E\left[\gamma^{t-1}\Delta\mu_t^{(-k)}(\gamma\Delta V_{t+1}^{(-k)}-\Delta Q_t^{(-k)})\mid\mathcal{D}^{(-k)}\right]\right) \label{eqn:non_adaptive_drift_term1}\\
        &\qquad+\sum_{t=1}^{T} \gamma^{t-1} \sqrt{n_k}\left(\E\left[\frac{C_{1:t}}{\hat{\Lambda}_{1:t}^{(-k)}\Lambda_{1:t}}\Delta\Lambda_{1:t}^{(-k)} \Delta\mu_t^{(-k)}(-\Delta b_t^{(-k)} +\gamma \Delta V_{t+1}^{(-k)}-\gamma\Delta \tilde{V}_{t+1}^{(-k)})\mid\mathcal{D}^{(-k)}\right]\right) \label{eqn:non_adaptive_drift_term2}\\
        &\qquad+\sum_{t=1}^{T}\gamma^{t-1}\sqrt{n_k}\left(\E\left[\frac{ C_{1:t}}{\hat{\Lambda}_{1:t}^{(-k)}\Lambda_{1:t}}\Delta\Lambda_{1:t}^{(-k)}\mu_t(-\Delta b_t^{(-k)} +\gamma \Delta V_{t+1}^{(-k)}-\gamma\Delta \tilde{V}_{t+1}^{(-k)})\mid\mathcal{D}^{(-k)}\right]\right) \label{eqn:non_adaptive_drift_term3}\\
        &\qquad+\sum_{t=1}^{T}\gamma^{t-1}\sqrt{n_k}\left(\E\left[\frac{C_{1:t}}{\hat{\Lambda}_{1:t}^{(-k)}\Lambda_{1:t}}\Delta\Lambda_{1:t}^{(-k)} \Delta\mu_t^{(-k)}\{(R_t - b_t)+\gamma(V_{t+1}-\tilde{V}_{t+1})\}\mid\mathcal{D}^{(-k)}\right]\right)\label{eqn:non_adaptive_drift_term4} \\
        &\qquad+\sum_{t=1}^{T-1}\gamma^t\sqrt{n_k}\left(\E\left[\frac{ C_{1:t}}{\hat{\Lambda}_{1:t}^{(-k)}\Lambda_{1:t}}\Delta\Lambda_{1:t}^{(-k)}\Delta \mu_{t+1}^{(-k)}(\Delta\tilde{Q}_{t+1}^{(-k)}-\Delta Q_{t+1}^{(-k)})\mid\mathcal{D}^{(-k)}\right]\right) \label{eqn:non_adaptive_drift_term5} \\
        &\qquad+\sum_{t=1}^{T-1}\gamma^t\sqrt{n_k}\left(\E\left[\frac{ C_{1:t}}{\hat{\Lambda}_{1:t}^{(-k)}\Lambda_{1:t}}\Delta\Lambda_{1:t}^{(-k)}\mu_{t+1}(\Delta \tilde Q_{t+1}^{(-k)}-\Delta Q_{t+1}^{(-k)})\mid\mathcal{D}^{(-k)}\right]\right) \label{eqn:non_adaptive_drift_term6}\\
        &\qquad+\sum_{t=1}^{T-1}\gamma^{t} \sqrt{n_k}\left(\E\left[\frac{C_{1:t}}{\hat{\Lambda}_{1:t}^{(-k)}\Lambda_{1:t}}\Delta\Lambda_{1:t}^{(-k)} \Delta \mu_{t+1}^{(-k)}(\tilde Q_{t+1} - Q_{t+1})\mid\mathcal{D}^{(-k)}\right]\right). \label{eqn:non_adaptive_drift_term7} 
    \end{align}
    We then bound each line via the Cauchy--Schwarz inequality to obtain a product of nuisance error terms, reducing the derived errors to primitive-nuisance errors via \cref{lem-v-q-error}. This
    yields either the DRL product error or the annotation product error, each of which is $o_p(1)$ by \cref{asn-rate}. \\
    
    \noindent\emph{Line \eqref{eqn:non_adaptive_drift_term1}:} 
    By the reduction of $\Delta V$ to $\Delta Q$ in \cref{lem-v-q-error}, we have
    \[
    \eqref{eqn:non_adaptive_drift_term1}\le\mathcal{O}\left(\sum_{t=1}^{T}\sqrt{n_k}\cdot\|\Delta\mu_t^{(-k)}\|_2\big(\|\Delta Q_{t+1}^{(-k)}\|_2+\|\Delta Q_t^{(-k)}\|_2\big)\right),
    \]
    and this is $o_p(1)$ by the DRL product error \cref{asn-rate}.\\
    
    \noindent\emph{Line \eqref{eqn:non_adaptive_drift_term2},\eqref{eqn:non_adaptive_drift_term3},\eqref{eqn:non_adaptive_drift_term5},\eqref{eqn:non_adaptive_drift_term6}:} Peeling the oracle density ratio $|\mu_t|\leq C_\mu$ (\cref{asn-seq_overlap}) and the its estimation error $|\Delta\mu_t^{(-k)}|\leq C_\mu+C_{\hat\mu}$, the oracle annotation probability $\Lambda_{1:t}\ge c_{\lambda}^{t}$ and its estimates $\hat\Lambda_{1:t}\ge c_{\hat\lambda}^t$ (\cref{asn-seq_overlap,asn-bdd-estimator}) leaves the product error between the annotation probability error $\Delta\Lambda_{1:t}^{(-k)}$ and the outcome error. As a result, we have
    \begin{align*}
        \eqref{eqn:non_adaptive_drift_term2}, \eqref{eqn:non_adaptive_drift_term3}&\leq \mathcal{O}\left(\sum_{t=1}^{T}\sqrt{n_k}\cdot \|\Lambda_{1:t}^{(-k)}\|_2\bigl(\|\Delta b_t^{(-k)}\|_2+\|\Delta V_{t+1}^{(-k)}\|_2+\|\Delta\tilde V_{t+1}^{(-k)}\|_2\bigr)\right), & (\because\text{Cauchy--Schwarz})\\
        &\leq \mathcal{O}\left(\sum_{t=1}^{T}\sqrt{n_k}\cdot \|\Delta\Lambda_{1:t}^{(-k)}\|_2\bigl(\|\Delta b_t^{(-k)}\|_2+\|\Delta Q_{t+1}^{(-k)}\|_2\bigr)\right) & (\because\text{\Cref{lem-v-q-error}})\\
        &\leq \mathcal{O}\left(\sum_{t=1}^{T}\sqrt{n_k}\cdot \sum_{j\leq  t}\|\Delta\lambda_{j}^{(-k)}\|_2\bigl(\|\Delta b_t^{(-k)}\|_2+\|\Delta Q_{t+1}^{(-k)}\|_2\bigr)\right), & (\because\text{\Cref{lem:prod_comparison}})\\
        \eqref{eqn:non_adaptive_drift_term5},\eqref{eqn:non_adaptive_drift_term6}&\leq \mathcal{O}\left(\sum_{t=1}^{T}\sqrt{n_k}
        \cdot\|\Delta\Lambda_{1:t}^{(-k)}\|_2\bigl(\|\Delta\tilde Q_{t+1}^{(-k)}\|_2+\|\Delta Q_{t+1}^{(-k)}\|_2\bigr)\right) & (\because\text{Cauchy--Schwarz}) \\
        &\leq \mathcal{O}\left(\sum_{t=1}^{T}\sqrt{n_k}\cdot\|\Delta\Lambda_{1:t}^{(-k)}\|_2\|\Delta Q_{t+1}^{(-k)}\|_2\right) & (\because\text{\Cref{lem-v-q-error}})  & \\
        &\leq \mathcal{O}\left(\sum_{t=1}^{T}\sqrt{n_k}\cdot\sum_{j\le t}\|\Delta\lambda_{j}^{(-k)}\|_2\|\Delta Q_{t+1}^{(-k)}\|_2\right), &(\because\text{\Cref{lem:prod_comparison}})
    \end{align*}
    which are all $o_p(1)$ under the assumed annotation product error \cref{asn-rate}.\\
    
    \noindent\emph{Line \eqref{eqn:non_adaptive_drift_term4}, \eqref{eqn:non_adaptive_drift_term7}:} Peeling the uniformly bounded factors $\hat\Lambda_{1:t}^{(-k)},\Lambda_{1:t},\{(R_t-b_t)+\gamma(V_{t+1}-\tilde V_{t+1})\}$, and $(\tilde Q_{t+1}-Q_{t+1})$ leaves the product error between the annotation probability error $\Delta\Lambda_{1:t}^{(-k)}$ and the density ratio error $\Delta\mu_t^{(-k)}$. As a consequence, we have
    \begin{align*}
        \eqref{eqn:non_adaptive_drift_term4} &\leq \mathcal{O}\left(\sum_{t=1}^{T}\sqrt{n_k}\|\Delta\Lambda_{1:t}^{(-k)}\|_2\|\Delta\mu_t^{(-k)}\|_2\right)\leq \mathcal{O}\left(\sum_{t=1}^{T}\sqrt{n_k}\cdot\sum_{j\leq t}\|\Delta\lambda_{j}^{(-k)}\|_2\|\Delta\mu_t^{(-k)}\|_2\right),\\
        \eqref{eqn:non_adaptive_drift_term7} &\leq \mathcal{O}\left(\sum_{t=1}^{T}\sqrt{n_k}\|\Delta\Lambda_{1:t}^{(-k)}\|_2\|\Delta \mu_{t+1}^{(-k)}\|_2\right) \leq \mathcal{O}\left(\sum_{t=1}^{T}\sqrt{n_k}\cdot\sum_{j\leq t}\|\Delta\lambda_{j}^{(-k)}\|_2\|\Delta\mu_{t+1}^{(-k)}\|_2\right).
    \end{align*}
    These are all $o_p(1)$ by the annotation product error \cref{asn-rate}. Therefore, combining these results ensures that the drift term is $o_p(1)$.
\end{proof}

\subsubsection{Proofs of Technical Lemmas}

\begin{proof}[Proof of \Cref{lem:diff_decomposition}]
    Using $\hat x \hat y-xy=(\hat x-x)(\hat y-y)+(\hat x-x)y+x(\hat y-y)$, one can decompose the difference between the feasible estimator and the feasible oracle estimator as follows: 
    \begin{align*}
        \tilde \Gamma(O;\hat{\bm\Lambda}, \hat{\bm\eta}_m)-\tilde\Gamma( O; \bm\Lambda, \bm\eta_m)&=(\hat m_0-m_0)+\sum_{t=1}^{T}C_{1:t}\left(\frac{1}{\hat\Lambda_{1:t}}(\hat m_t-\hat m_{t-1})-\frac{1}{\Lambda_{1:t}}(m_t-m_{t-1})\right) \\
        &=\Delta m_0-\sum_{t=1}^{T}\frac{C_{1:t}}{\hat{\Lambda}_{1:t}\Lambda_{1:t}}\Delta\Lambda_{1:t}\Delta(m_t-m_{t-1}) \\
        &\qquad\ \ \quad -\sum_{t=1}^{T}\frac{C_{1:t}}{\hat{\Lambda}_{1:t}\Lambda_{1:t}}\Delta\Lambda_{1:t}(m_t-m_{t-1}) \\
        &\qquad\ \ \quad+\sum_{t=1}^{T}\frac{C_{1:t}}{\Lambda_{1:t}} \Delta(m_t-m_{t-1}).
    \end{align*}
    This gives the intermediate result stated in the lemma above. Now it remains to prove that each full imputation estimation error, oracle projection increment, and the increment estimation error further decomposes as above. \\

    \noindent \emph{(i) Full imputation estimation error $\Delta m_0$:} Note that the closed form of the full imputation $m_0$ can be derived by replacing all reward $(R_t)$, value functions $(V_{t})$, action-value function $(Q_t)$ of the semiparametrically efficient OPE estimator to its $\mcF_0$-projections, which are $b_t, \tilde V_t, \tilde Q_t$, respectively. This gives the closed form as follows:
    \begin{align*}
        m_0=V_1+\sum_{t=1}^{T}\gamma^{t-1}\mu_t(b_t+\gamma\tilde V_{t+1}-\tilde Q_t),
    \end{align*}
    where the initial value function satisfies $\tilde V_1=V_1$, because $\tilde S_1=S_1$. Then, using the equation $\hat{x}\hat{y}-xy=(\hat{x}-x)(\hat{y}-y)+(\hat{x}-x)y+x(\hat{y}-y)$ again, we can rewrite $\Delta m_0$ as
    \begin{align*}
        \Delta m_0&=\left(\hat{V}_1+\sum_{t=1}^{T}\gamma^{t-1}\hat{\mu}_t(\hat b_t+\gamma\hat{\tilde V}_{t+1}-\hat{\tilde Q}_t)\right)-\left(V_1+\sum_{t=1}^{T}\gamma^{t-1}\mu_t(b_t+\gamma\tilde V_{t+1}-\tilde Q_t)\right) \\
        &=\Delta V_1+\sum_{t=1}^{T}\gamma^{t-1}\Delta\mu_t\Delta(b_t+\gamma\tilde V_{t+1}-\tilde Q_t)\\
        &\qquad \ \quad +\sum_{t=1}^{T}\gamma^{t-1}\Delta \mu_t(b_t+\gamma\tilde V_{t+1}-\tilde Q_t) \\
        &\qquad \ \quad + \sum_{t=1}^{T}\gamma^{t-1}\mu_t\Delta(b_t+\gamma\tilde V_{t+1}-\tilde Q_t).
    \end{align*}    
    Rewriting $\Delta(b_t+\gamma\tilde V_{t+1}-\tilde Q_t)=\Delta b_t+\gamma\Delta\tilde V_{t+1}-\Delta \tilde Q_t$ gives the same form as above. \\

    \noindent\emph{(ii) Oracle projection increment $m_t-m_{t-1}$:} Oracle projection increment term $m_t-m_{t-1}$ measures the marginal improvement on $\Gamma_T^{\pi_e}$ estimation by using $(R_t, S_{t+1})$ instead of $(\tilde R_t, \tilde S_{t+1})$, the information that annotating the data at time $t$ reveals. Therefore, $m_t-m_{t-1}$ leaves the term where revealing $R_t$ and $S_{t+1}$ impacts:
    \[
        m_t-m_{t-1}=\gamma^{t-1}\mu_t\{(R_t-b_t)+\gamma(V_{t+1}-\tilde V_{t+1})\}+\mathbf{1}\{t<T\}\gamma^t\mu_{t+1}(\tilde Q_{t+1}-Q_{t+1}).
    \]

    \noindent\emph{(iii) Increment estimation error $\Delta(m_t-m_{t-1})$:} Building on the closed form of $m_t-m_{t-1}$ below 
    \[
        m_t-m_{t-1}=\gamma^{t-1}\mu_t\{(R_t-b_t)+\gamma(V_{t+1}-\tilde V_{t+1})\}+\mathbf{1}\{t<T\}\gamma^t\mu_{t+1}(\tilde Q_{t+1}-Q_{t+1}),
    \]
    we can rewrite $\Delta(m_t-m_{t-1})$ using $\hat x\hat y-xy=(\hat x-x)(\hat{y}-y)+(\hat{x}-x)y+x(\hat{y}-y)$ again as 
    \begin{align*}
        \Delta(m_t-m_{t-1}) &= \gamma^{t-1}\Delta\mu_t\Delta\{(R_t-b_t)+\gamma(V_{t+1}-\tilde V_{t+1})\} \\
        &+\gamma^{t-1}\Delta \mu_t\{(R_t-b_t)+\gamma(V_{t+1}-\tilde V_{t+1})\} \\
        &+\gamma^{t-1}\mu_t\Delta\{(R_t-b_t)+\gamma(V_{t+1}-\tilde V_{t+1})\} \\
        &+\mathbf{1}\{t<T\}\gamma^{t}\Delta\mu_{t+1}\Delta(\tilde Q_{t+1}-Q_{t+1}) \\
        &+\mathbf{1}\{t<T\}\gamma^t\Delta\mu_{t+1}(\tilde Q_{t+1}-Q_{t+1}) \\
        &+\mathbf{1}\{t<T\}\gamma^t\mu_{t+1}\Delta(\tilde Q_{t+1}-Q_{t+1}),
    \end{align*}
    which further reduces to
    \begin{align*}
        \Delta(m_t-m_{t-1}) &= \gamma^{t-1}\Delta\mu_t(-\Delta b_t+\gamma\Delta V_{t+1}-\gamma\Delta\tilde V_{t+1}) \\
        &+\gamma^{t-1}\Delta \mu_t\{(R_t-b_t)+\gamma(V_{t+1}-\tilde V_{t+1})\} \\
        &+\gamma^{t-1}\mu_t(-\Delta b_t+\gamma\Delta V_{t+1}-\gamma\Delta\tilde V_{t+1}) \\
        &+\mathbf{1}\{t<T\}\gamma^{t}\Delta\mu_{t+1}(\Delta\tilde Q_{t+1}-\Delta Q_{t+1}) \\
        &+\mathbf{1}\{t<T\}\gamma^t\Delta\mu_{t+1}(\tilde Q_{t+1}-Q_{t+1}) \\
        &+\mathbf{1}\{t<T\}\gamma^t\mu_{t+1}(\Delta\tilde Q_{t+1}-\Delta Q_{t+1}),
    \end{align*}
    as $\Delta R_t=0$. This concludes the proof.
\end{proof}
\begin{proof}[Proof of \Cref{lem:mean_diff_decomposition}]
    In \cref{lem:diff_decomposition}, we proved 
    \begin{align*}
        \tilde\Gamma(O;\hat{\bm\Lambda},\hat{\bm\eta}_m)-\tilde\Gamma(O;\bm\Lambda,\bm\eta_m)
        =\Delta m_0
        &+\sum_{t=1}^{T}\frac{C_{1:t}}{\Lambda_{1:t}}\Delta\big(m_t-m_{t-1}\big).\\
        &-\sum_{t=1}^{T}\frac{C_{1:t}}{\hat{\Lambda}_{1:t}\Lambda_{1:t}}\Delta\Lambda_{1:t} \big(m_t-m_{t-1}\big) \\        
        &-\sum_{t=1}^{T}\frac{C_{1:t}}{\hat{\Lambda}_{1:t}\Lambda_{1:t}}\Delta\Lambda_{1:t} \Delta\big(m_t-m_{t-1}\big) 
    \end{align*}
    Taking an expectation on the both sides reduces the first line in the RHS to $\E[\Delta m_T]$ due to telescoping, i.e., 
    \begin{align*}
        \E\left[\Delta m_0\right]+\sum_{t=1}^{T}\E\left[\frac{C_{1:t}}{\Lambda_{1:t}}(\Delta m_t- \Delta m_{t-1})\right] &= \E[\Delta m_0] + \sum_{t=1}^{T}\E[\Delta m_{t}-\Delta m_{t-1}]=\E[\Delta m_T],
    \end{align*}
    and eliminates the second line in the RHS due to the martingale property, i.e.,
    \begin{align*}
        \sum_{t=1}^{T}\E\left[\frac{C_{1:t}}{\hat{\Lambda}_{1:t}\Lambda_{1:t}}\Delta\Lambda_{1:t}(m_t-m_{t-1})\right]&=\sum_{t=1}^{T}\E\left[\frac{\Delta\Lambda_{1:t}}{\hat\Lambda_{1:t}}\underbrace{\E[(m_t-m_{t-1})\mid\mcF_{t-1}]}_{=0}\right] = 0.
    \end{align*}
    As a result, the expectation simplifies as follows:
    \begin{align*}
        \tilde \Gamma(O; \hat{\bm\Lambda},\hat{\bm\eta}_m)-\tilde\Gamma(O; \bm\Lambda, \bm\eta_m) &= \E[\Delta m_T] - \E\left[\sum_{t=1}^{T}\frac{C_{1:t}}{\hat{\Lambda}_{1:t}\Lambda_{1:t}}\Delta\Lambda_{1:t}\Delta(m_t-m_{t-1})\right],
    \end{align*}
    Since we derived the closed form of $\Delta(m_t-m_{t-1})$ in the previous lemma, it remains to derive the closed form of double reinforcement learning (DRL) error $\Delta m_T$. Note that $m_T$ is the projection of $\Gamma_T^{\pi_e}$ on the complete history $\mcF_T$, which is $\Gamma_T^{\pi_e}$ itself:
    \[
        m_T=V_1+\sum_{t=1}^{T}\gamma^{t-1}\mu_t(R_t+\gamma V_{t+1}-Q_t).
    \]
    Note that DRL error decomposes as
    \begin{align*}
        \Delta m_T&=\left\{\hat V_1+\sum_{t=1}^{T}\gamma^{t-1}\hat\mu_t(R_t+\gamma\hat{V}_{t+1}-\hat{Q}_t)\right\} - \left\{V_1 + \sum_{t=1}^{T}\gamma^{t-1}\mu_t(R_t+\gamma V_{t+1}-Q_t)\right\} \\
        &=\Delta V_1+ \sum_{t=1}^{T}\gamma^{t-1}\mu_t\Delta(R_t+\gamma V_{t+1}-Q_t)\\
        &\qquad\ \quad+\sum_{t=1}^{T}\gamma^{t-1}\Delta\mu_t(R_t+\gamma V_{t+1}-Q_t)\\
        &\qquad\ \quad +\sum_{t=1}^{T}\gamma^{t-1}\Delta \mu_t\Delta(R_t+\gamma V_{t+1}- Q_t).
    \end{align*}
    Taking an expectation on both sides eliminates the first line in the RHS due to telescoping, i.e., 
    \begin{align*}
        \E[\Delta V_1]+ \sum_{t=1}^{T}\gamma^{t-1}\E[\mu_t\Delta(R_t+\gamma V_{t+1}-Q_t)]&=\E[\Delta V_1] + \sum_{t=1}^{T}\gamma^{t-1}\E[\mu_t (\gamma \Delta V_{t+1}- \Delta Q_t)]  \\
        &=\E[\Delta V_1]+\sum_{t=1}^{T}(\gamma^{t}\E[\Delta V_{t+1}]-\gamma^{t-1}\E[\Delta Q_t]) \\
        &=\E[\Delta V_1]+\sum_{t=1}^{T}(\gamma^{t}\E[\Delta V_{t+1}]-\gamma^{t-1}\E[\Delta V_t] )\\
        &=\gamma^T\E[\Delta V_{T+1}]=0
    \end{align*}
    where the second equality is due to the change of measure from data generating distribution induced by behavior policy to the one by target policy and the third equality is due to the tower rule. \\
    
    \noindent Next, the second line also vanishes due to the Bellman equation:
    \[
        \E[\gamma^{t-1}\Delta\mu_t(R_t+\gamma V_{t+1}-Q_t)]=\E[\gamma^{t-1}\Delta\mu_t\underbrace{\E[(R_t+\gamma V_{t+1}-Q_t) \mid S_t, A_t]}_{=0}] = 0.
    \]
    As a result, 
    \begin{align*}
        \E[\Delta m_T]&=\sum_{t=1}^{T}\gamma^{t-1}\E[\Delta\mu_t\Delta(R_t+\gamma V_{t+1}-Q_t)] =\sum_{t=1}^{T}\gamma^{t-1}\E[\Delta \mu_t(\gamma\Delta V_{t+1}-\Delta Q_t)].
    \end{align*}
    Combining this result with the closed form of $\Delta (m_t-m_{t-1})$, we have:
    \begin{align*}
        \E[\tilde{\Gamma}(O;\hat{\bm\Lambda},\hat{\bm\eta}_m) - \tilde\Gamma(O;\bm\Lambda, \bm\eta_m)] &= \E[\Delta m_T]-\E\left[\sum_{t=1}^{T}\frac{C_{1:t}}{\hat\Lambda_{1:t}\Lambda_{1:t}}\Delta\Lambda_{1:t}\Delta(m_t-m_{t-1})\right]\\
        &=\sum_{t=1}^{T}\E\left[\gamma^{t-1}\Delta\mu_t(\gamma\Delta V_{t+1}-\Delta Q_t)\right]\\
        &+\sum_{t=1}^{T} \gamma^{t-1} \E\left[\frac{C_{1:t}}{\hat{\Lambda}_{1:t}\Lambda_{1:t}}\Delta\Lambda_{1:t} \Delta\mu_t(\Delta b_t -\gamma \Delta V_{t+1} + \gamma\Delta \tilde{V}_{t+1})\right]\\
        &+\sum_{t=1}^{T}\gamma^{t-1}\E\left[\frac{ C_{1:t}}{\hat{\Lambda}_{1:t}\Lambda_{1:t}}\Delta\Lambda_{1:t}\mu_t(\Delta b_t -\gamma \Delta V_{t+1}+\gamma\Delta \tilde{V}_{t+1})\right]\\
        &+\sum_{t=1}^{T}\gamma^{t-1}\E\left[\frac{C_{1:t}}{\hat{\Lambda}_{1:t}\Lambda_{1:t}}\Delta\Lambda_{1:t} \Delta\mu_t\{(b_t - R_t)+\gamma(\tilde{V}_{t+1}- V_{t+1})\}\right]  \\
        &+\sum_{t=1}^{T-1}\gamma^t\E\left[\frac{ C_{1:t}}{\hat{\Lambda}_{1:t}\Lambda_{1:t}}\Delta\Lambda_{1:t}\Delta \mu_{t+1}(\Delta Q_{t+1}- \Delta\tilde{Q}_{t+1})\right] \\
        &+\sum_{t=1}^{T-1}\gamma^t\E\left[\frac{ C_{1:t}}{\hat{\Lambda}_{1:t}\Lambda_{1:t}}\Delta\Lambda_{1:t}\mu_{t+1}(\Delta Q_{t+1}- \Delta \tilde Q_{t+1})\right]  \\
        &+\sum_{t=1}^{T-1}\gamma^{t} \E\left[\frac{C_{1:t}}{\hat{\Lambda}_{1:t}\Lambda_{1:t}}\Delta\Lambda_{1:t} \Delta \mu_{t+1}(Q_{t+1}- \tilde Q_{t+1})\right] . 
    \end{align*}
\end{proof}

\begin{proof}[Proof of \Cref{lem-v-q-error}]

    \begin{align*}
        \|\Delta V_1^{(-k)}\|_{2}^2&=\E_{S_1\sim \mathbb{P}_0}[(\E_{A_1\sim \pi_e(\cdot|S_1)}[\Delta Q_1^{(-k)}(S_1, A_1) \mid S_1])^2\mid \mathcal{D}^{(-k)}] \\
        &\leq \E_{S_1\sim \mathbb{P}_0}[\E_{A_1\sim \pi_e(\cdot|S_1)}[(\Delta Q_1^{(-k)}(S_1, A_1))^2  \mid S_1]\mid \mathcal{D}^{(-k)}] \\
        & \qquad\qquad\qquad\qquad\qquad\qquad\qquad\qquad\qquad\qquad\qquad\qquad\qquad(\because\text{Jensen's inequality})\\
        &\leq \E_{S_1\sim \mathbb{P}_0}[\E_{A_1\sim \pi_b(\cdot|S_1)}[(\pi_e/\pi_b)(A_1\mid S_1)(\Delta Q_1^{(-k)}(S_1, A_1))^2  \mid S_1]\mid \mathcal{D}^{(-k)}] \\
        &\qquad\qquad\qquad\qquad\qquad\qquad\qquad\qquad\qquad\qquad\qquad\qquad\quad(\because\text{Importance sampling})\\
        &\leq \epsilon_b^{-1}\E_{S_1\sim\mathbb{P}_0}[\E_{A_1\sim\pi_b(\cdot\mid S_1)}[\Delta Q_1^{(-k)}(S_1, A_1))^2\mid S_1]\mid \mathcal{D}^{(-k)}] \\
        &\qquad\qquad\qquad\qquad\qquad\qquad\qquad\qquad\qquad\qquad\qquad\qquad\qquad\qquad(\because\text{\Cref{asn-seq_overlap}})\\
        &=\epsilon_b^{-1}\E_{(S_1, A_1)}[(\Delta Q_1^{(-k)}(S_1, A_1))^2\mid \mathcal{D}^{(-k)}] \\
        &\qquad\qquad\qquad\qquad\qquad\qquad\qquad\qquad\qquad\qquad\qquad\qquad\qquad\quad\  (\because\text{Tower property})\\
        &=\epsilon_b^{-1}\|\Delta Q_1^{(-k)}\|_{2}^2=o_p(1),
    \end{align*}
    \begin{align*}
        \|\Delta \tilde V_{t+1}^{(-k)}\|_{2}^2 &= \E_{\tilde S_{t+1}}\Bigl[\underbrace{\bigl(\E_{S_{t+1}} \bigl[\underbrace{\E_{A_{t+1}\sim\pi_e(\cdot \mid S_{t+1})}[\Delta Q_{t+1}^{(-k)}(S_{t+1}, A_{t+1}) \mid S_{t+1}]}_{=\Delta V_{t+1}^{(-k)}(S_{t+1})}\mid \tilde S_{t+1} \bigr]\bigr)^2}_{=(\Delta \tilde V_{t+1}^{(-k)}(\tilde S_{t+1}))^2} \mid \mathcal{D}^{(-k)}\Bigr] \\
        &\leq \E_{\tilde S_{t+1}}\Bigl[\E_{S_{t+1}}\bigl[\E_{A_{t+1}\sim\pi_e(\cdot \mid S_{t+1})}[(\Delta Q_{t+1}^{(-k)}(S_{t+1}, A_{t+1}))^2\mid S_{t+1}]\mid\tilde S_{t+1}\bigr]\mid \mathcal{D}^{(-k)}\Bigr] \\
        & \qquad\qquad\qquad\qquad\qquad\qquad\qquad\qquad\qquad\qquad\qquad\qquad\qquad\qquad(\because\text{Jensen's inequality}) \\
        &\leq \E_{\tilde S_{t+1}}\Bigl[\E_{S_{t+1}}\bigl[\E_{A_{t+1}\sim\pi_b(\cdot \mid S_{t+1})}[(\pi_e/\pi_b)(A_{t+1}\mid S_{t+1})(\Delta Q_{t+1}^{(-k)}(S_{t+1}, A_{t+1}))^2\mid S_{t+1}]\mid\tilde S_{t+1}\bigr]\mid \mathcal{D}^{(-k)}\Bigr]\\
        & \qquad\qquad\qquad\qquad\qquad\qquad\qquad\qquad\qquad\qquad\qquad\qquad\qquad\ \ \ (\because\text{Importance sampling}) \\
        &\leq \epsilon_{b}^{-1}\E_{\tilde S_{t+1}}\Bigl[\E_{S_{t+1}}[\E_{A_{t+1}\sim\pi_b(\cdot \mid S_{t+1})}[\Delta Q_{t+1}^{(-k)}(S_{t+1}, A_{t+1})^2 \mid S_{t+1}]\mid \tilde S_{t+1}]\mid \mathcal{D}^{(-k)}\Bigr] \\
        & \qquad\qquad\qquad\qquad\qquad\qquad\qquad\qquad\qquad\qquad\qquad\qquad\qquad\qquad\qquad\ (\because\text{\Cref{asn-seq_overlap}})\\
        &=\epsilon_b^{-1}\E_{(S_{t+1}, A_{t+1})}[(\Delta Q_{t+1}^{(-k)}(S_{t+1}, A_{t+1}))^2\mid \mathcal{D}^{(-k)}]\\
        &\qquad\qquad\qquad\qquad\qquad\qquad\qquad\qquad\qquad\qquad\qquad\qquad\qquad\qquad\qquad (\because\text{Tower property})\\
        &= \epsilon_b^{-1} \|\Delta Q_{t+1}^{(-k)}\|_{2}^{2}=o_p(1),
    \end{align*}
    and
    \begin{align*}
        \|\Delta\tilde Q_t^{(-k)}\|_{2}^{2}&=\E_{(\tilde S_t, A_t)}\big[\big(\E_{S_t}[\Delta Q_t^{(-k)}(S_t, A_t) \mid \tilde S_t, A_t]\big)^2\big] \\
        &\leq \E_{(\tilde S_t, A_t)}\big[\E_{S_t}\big[(\Delta Q_t^{(-k)}(S_t, A_t))^2 \mid \tilde S_t, A_t\big]\big] \qquad\qquad\qquad\qquad\ \ (\because\text{Jensen's inequality})\\
        &= \E_{(S_t,A_t)}\big[(\Delta Q_t^{(-k)}(S_t, A_t))^2\big] \qquad\qquad\qquad\qquad\qquad\qquad\qquad\qquad(\because\text{Tower property})\\
        &= \|\Delta Q_t^{(-k)}\|_{2}^{2}=o_p(1).
    \end{align*}
    \end{proof}
    
\subsection{Proof of \Cref{thm:clt_adaptive}} \label{pf-clt_adaptive}
\begin{proof}
    Recall that the batch-adaptive estimator $\hat\psi_{\mathrm{ad}}$ takes the form of
    \[
        \hat\psi_{\mathrm{ad}}=\frac{1}{n}\sum_{k=1}^{K}\sum_{(b,i)\in\mathcal{I}_k}\tilde\Gamma\left(O_{b,i};\hat{\bm\Lambda}^{(-k)}, \hat{\bm\eta}_m^{(-k)}\right), \quad \text{where} \quad \tilde\Gamma(O;\bm\Lambda,\bm\eta_m)=m_0+\sum_{t=1}^{T}\frac{C_{1:t}}{\Lambda_{1:t}}(m_t-m_{t-1}).
    \]
    In contrast to the non-adaptive estimator, the outcomes in each batch--fold are now drawn under a potentially different, (batch--fold)--dependent random annotation probability $\hat\lambda_{b,t}^{(k)}$, rather than a single non-random oracle $\lambda_t$ across all batches. We therefore need to compare the realized experiment to a counterfactual, non-adaptive experiment. For this reason, we define $\tilde O_{b,i}$ to denote the counterfactual outcome of unit $i$ in batch $b$ under the (nonrandom) oracle annotation policy $\lambda_{b,t}^{\ast}$, and set
    \[
        \varphi_{b,i}^{(k)} := \tilde\Gamma\big(O_{b,i};\hat{\bm\Lambda}^{(-k)}, \hat{\bm\eta}_m^{(-k)}\big)-\tilde\Gamma\big(\tilde O_{b,i}; \bm\Lambda^{\ast}, \bm\eta_m\big).
    \]
    For $b=1,\dots,M$ and $k=1,\dots,K$, let
    $\mathcal D_{1:b-1}^{(k)}:=\sigma\big(\{O_{u,i}:(u,i)\in\mathcal I_k,\ u<b\}\big)$ be the $\sigma$-algebra that contains every available data for experiment design in batch-$b$, fold-$k$. This includes the fold-$k$ data of all batches collected before batch $b$. By construction, $\hat{\bm\Lambda}_b^{(k)}$ is $\mathcal D_{1:b-1}^{(k)}$-measurable, and the batch-$b$, fold-$k$ trajectories are independent of
    $\mathcal D^{(-k)}$. Splitting the first term by centering each summand at its
    $\sigma(\mathcal D_{1:b-1}^{(k)}, \mathcal{D}^{(-k)})$-conditional mean, we obtain the three-term decomposition
    \begin{align*}
        \sqrt{n}(\hat \psi_{\mathrm{ad}} - \Phi^{\pi_e}) &=\underbrace{\sum_{k=1}^{K}\sum_{b=1}^{M}\sqrt{\frac{n_{b,k}}{n}}\left\{\frac{1}{\sqrt{n_{b,k}}}\sum_{i:(b, i)\in\mathcal{I}_k}  \left(\varphi_{b,i}^{(k)}-\E\big[\varphi_{b,i}^{(k)} \mid \mathcal D_{1:b-1}^{(k)}, \mathcal{D}^{(-k)} \big]\right)\right\}}_{\text{empirical process term}} \\
        &+\underbrace{\sum_{k=1}^{K}\sum_{b=1}^{M}\sqrt{\frac{n_{b,k}}{n}}\left\{ \frac{1}{\sqrt{n_{b,k}}}\sum_{i:(b, i)\in\mathcal{I}_k}\E\big[\varphi_{b,i}^{(k)} \mid \mathcal D_{1:b-1}^{(k)}, \mathcal{D}^{(-k)} \big]\right\}}_{\text{drift term}} \\
        &+\underbrace{\sum_{b=1}^{M}\sqrt{\frac{n_b}{n}}\left\{ \frac{1}{\sqrt{n_b}}\sum_{i\in\mathcal{D}_b}\left(\tilde\Gamma(\tilde O_{b,i}; \bm\Lambda^\dagger, \bm\eta_m)-\Phi^{\pi_e}\right)\right\}}_{\text{oracle term}},
    \end{align*}
    where $\mathcal{D}_b$ denotes a set of data in batch-$b$ and $n_{b,k}=|\mathcal{D}_b^{(k)}|$ refers to the size of data in batch-$b$, fold-$k$, thus $\sum_{k=1}^{K}\sum_{b=1}^{M} n_{b,k}=n$. Structurally, the proof closely follows that of \cref{thm:clt_non-adaptive}. We establish the asymptotic behavior of the empirical process, drift, and oracle term in \cref{prop-emp-process-adaptive,prop-drift-adaptive,prop-oracle_adaptive}.

\end{proof}

\begin{proposition}[Empirical process term] 
    \label{prop-emp-process-adaptive}
    With $\varphi_{b,i}^{(k)}:=\tilde \Gamma(O_{b,i};\hat{\bm\Lambda}^{(-k)}, \hat{\bm\eta}_m^{(-k)})-\tilde \Gamma(\tilde O_{b,i};\bm\Lambda^\dagger, \bm\eta_m)$ for $(b,i)\in \mathcal{I}_k$, the empirical process term is $o_p(1)$, i.e.,
    \[
        \sum_{k=1}^{K}\sum_{b=1}^{M}\sqrt{\frac{n_{b,k}}{n}}\left\{\frac{1}{\sqrt{n_{b,k}}}\sum_{i:(b, i)\in\mathcal{I}_k}  \left(\varphi_{b,i}^{(k)}-\E\big[\varphi_{b,i}^{(k)} \mid \mathcal D_{1:b-1}^{(k)}, \mathcal{D}^{(-k)} \big]\right)\right\}=o_p(1).
    \]
    
\end{proposition}
\begin{proof}[Proof of \cref{prop-emp-process-adaptive}]
    For any batch-$b$ and fold-$k$, it suffices to show
    \begin{equation} \label{eqn:pf_cvg_prob_adaptive}
        \lim_{n_{b,k}\to \infty} P\left[\bigg\vert \frac{1}{\sqrt{n_{b,k}}}\sum_{i:(b, i)\in\mathcal{I}_k}\left(\varphi_{b,i}^{(k)}-\E[\varphi_{b,i}^{(k)}\mid \mathcal{D}_{1:b-1}^{(k)}, \mathcal{D}^{(-k)}]\right) \bigg\vert > \epsilon \Bigg \vert \mathcal{D}_{1:b-1}^{(k)}, \mathcal{D}^{(-k)} \right]=0,
    \end{equation}
    because once the above statement holds, the bounded convergence theorem implies that
    \[
        \lim_{n_{b,k}\to \infty}\mathbb{P}\left[\bigg\vert \frac{1}{\sqrt{n_{b,k}}}\sum_{i:(b, i)\in\mathcal{I}_k}\left(\varphi_{b,i}^{(k)}-\E[\varphi_{b,i}^{(k)}\mid \mathcal{D}_{1:b-1}^{(k)},\mathcal{D}^{(-k)}]\right) \bigg\vert > \epsilon \right]=0,
    \]
    or equivalently,
    \[
        \frac{1}{\sqrt{n_{b,k}}}\sum_{i:(b, i)\in\mathcal{I}_k}\left(\varphi_{b,i}^{(k)}-\E[\varphi_{b,i}^{(k)}\mid \mathcal{D}_{1:b-1}^{(k)}, \mathcal{D}^{(-k)}]\right) =o_p(1).
    \]
    Then it follows that
    \[
        \sum_{k=1}^{K}\sum_{b=1}^{M}\sqrt{\frac{n_{b,k}}{n}}\left\{ \frac{1}{\sqrt{n_{b,k}}}\sum_{i:(b, i)\in\mathcal{I}_k}  \left(\varphi_{b,i}^{(k)}-\E[\varphi_{b,i}^{(k)} \mid \mathcal{D}_{1:b-1}^{(k)},\mathcal{D}^{(-k)}]\right)\right\} = o_p(1),
    \]
    since the number of batches and folds are finite (i.e., $K<\infty, M<\infty$ and $\sqrt{n_{b,k}/n}\to\sqrt{\kappa_b/K}$ as $n_{b,k}\to\infty$).
    To show the requirement \eqref{eqn:pf_cvg_prob_adaptive}, we proceed by deriving the upper bound of probability using Chebyshev's inequality, and show that the upper bound is $o_p(1)$. With that, we start with the conditional mean. Note that the conditional mean is zero:
    \[
        \E\left[\frac{1}{\sqrt{n_{b,k}}}\sum_{i:(b, i)\in\mathcal{I}_k}\left(\varphi_{b,i}^{(k)}-\E[\varphi_{b,i}^{(k)}\mid \mathcal{D}_{1:b-1}^{(k)}, \mathcal{D}^{(-k)}]\right) \bigg\vert \mathcal{D}_{1:b-1}^{(k)}, \mathcal{D}^{(-k)}\right]=0.
    \]
    Next, we consider the conditional variance. Note that the conditional variance can be simplified as
    \begin{align*}
        &\Var\left(\frac{1}{\sqrt{n_{b,k}}}\sum_{i:(b, i)\in\mathcal{I}_k}\left(\varphi_{b,i}^{(k)}-\E[\varphi_{b,i}^{(k)} \mid \mathcal{D}_{1:b-1}^{(k)}, \mathcal{D}^{(-k)}]\right)\bigg\vert \mathcal{D}_{1:b-1}^{(k)}, \mathcal{D}^{(-k)}\right)\\
        &\qquad =\frac{1}{n_{b,k}}\sum_{i:(b, i)\in\mathcal{I}_k}\Var\left(\varphi_{b,i}^{(k)} \big \vert \mathcal{D}_{1:b-1}^{(k)}, \mathcal{D}^{(-k)}\right) =\frac{1}{n_{b,k}}\sum_{i:(b, i)\in\mathcal{I}_k}\E\left[(\varphi_{b,i}^{(k)})^2 \mid \mathcal{D}_{1:b-1}^{(k)}, \mathcal{D}^{(-k)}\right],
    \end{align*}
    where the first equality is due to the fact that every $\varphi_{b,i}^{(k)}$ for $(b,i)\in\mathcal{I}_k$ is independent given $\mathcal{D}_{1:b-1}^{(k)}$ and $\mathcal{D}^{(-k)}$, and the second equality is due to the fact that the conditional mean of $\varphi_{b,i}^{(k)}$ is zero. Provided that $\varphi_{b,i}^{(k)}$ decomposes as
    \[
        \varphi_{b,i}^{(k)}=\left(\tilde\Gamma(O_{b,i};\hat{\bm\Lambda}^{(-k)},\hat{\bm\eta}_m^{(-k)})-\tilde\Gamma(\tilde O_{b,i};\hat{\bm\Lambda}^{(-k)},\hat{\bm\eta}_m^{(-k)})\right)+\left(\tilde\Gamma(\tilde O_{b,i};\hat{\bm\Lambda}^{(-k)},\hat{\bm\eta}_m^{(-k)})-\tilde\Gamma(\tilde O_{b,i};\bm\Lambda^\dagger,\bm\eta_m)\right),
    \]
    one can decompose the conditional second moment of $\varphi_{b,i}^{(k)}$ as follows using $(a+b)^2 \leq 2(a^2+b^2)$:
    \begin{align*}
        \E\left[(\varphi_{b,i}^{(k)})^2 \mid \mathcal{D}_{1:b-1}^{(k)}, \mathcal{D}^{(-k)}\right]&\leq 2\E\left[\left(\tilde\Gamma(O_{b,i};\hat{\bm\Lambda}^{(-k)},\hat{\bm\eta}_m^{(-k)})-\tilde\Gamma(\tilde O_{b,i};\hat{\bm\Lambda}^{(-k)},\hat{\bm\eta}_m^{(-k)}\right)^2 \mid \mathcal{D}_{1:b-1}^{(k)},\mathcal{D}^{(-k)}\right] \\
        &+2\E\left[\left(\tilde\Gamma(\tilde O_{b,i};\hat{\bm\Lambda}^{(-k)},\hat{\bm\eta}_m^{(-k)})-\tilde\Gamma(\tilde O_{b,i}; \bm\Lambda^\dagger, \bm\eta_m)\right)^2 \mid \mathcal{D}_{1:b-1}^{(k)}, \mathcal{D}^{(-k)} \right].
    \end{align*}
    We show that both terms are $o_p(1)$.\\
    
    \noindent When it comes to the first term, note that
    \begin{align*}
        &\E\left[\left(\tilde\Gamma(O_{b,i};\hat{\bm\Lambda}^{(-k)},\hat{\bm\eta}_m^{(-k)})-\tilde\Gamma(\tilde O_{b,i};\hat{\bm\Lambda}^{(-k)},\hat{\bm\eta}_m^{(-k)}\right)^2 \mid \mathcal{D}_{1:b-1}^{(k)},\mathcal{D}^{(-k)}\right] \\
        &\qquad =\E\left[\left(\sum_{t=1}^{T}\left(\frac{C_{i,1:t}-\tilde C_{i,1:t}}{\hat\Lambda_{1:t}^{(-k)}}\right)(\hat m_t^{(-k)}-\hat m_{t-1}^{(-k)})\right)^2\bigg\vert \mathcal{D}_{1:b-1}^{(k)},\mathcal{D}^{(-k)}\right] \\
        &\qquad \leq T\sum_{t=1}^{T}\E\left[\left(\frac{C_{i,1:t}-\tilde C_{i,1:t}}{\hat\Lambda_{1:t}^{(-k)}}\right)^2(\hat m_t^{(-k)}-\hat m_{t-1}^{(-k)})^2 \bigg\vert \mathcal{D}_{1:b-1}^{(k)},\mathcal{D}^{(-k)}\right] \\
        &\qquad \leq\sum_{t=1}^{T}\mathcal{O}\left(\E\left[|\hat \Lambda_{b,1:t}^{(k)}-\Lambda_{b,1:t}^\ast| \big\vert \mathcal{D}_{1:b-1}^{(k)},\mathcal{D}^{(-k)}\right]\right) \\
        &\qquad \leq \sum_{t=1}^{T}\sum_{j\leq t} \mathcal{O}\left(\E\left[|\hat\lambda_{b,j}^{(k)}-\lambda_{b,t}^\ast|\mid \mathcal{D}_{1:b-1}^{(k)}, \mathcal{D}^{(-k)}\right]\right)=o_p(1),
    \end{align*}
    where the first inequality holds by applying the Cauchy--Schwarz inequality over the sum. The second inequality follows from bounding $(\hat m_t^{(-k)}-\hat m_{t-1}^{(-k)})$ and $\hat\Lambda_{1:t}^{(-k)}$ using the bounded estimators \cref{asn-bdd_estimators_adaptive}, followed by the equation
    \begin{align*}
        \E[(C_{i,1:t}-\tilde C_{i,1:t})^2 \mid \mathcal{D}_{1:b-1}^{(k)},\mathcal{D}^{(-k)}]
        &=\E\big[\E[(C_{i,1:t}-\tilde C_{i,1:t})^2 \mid \mcF_{i,t}, \mathcal{D}_{1:b-1}^{(k)},\mathcal{D}^{(-k)}]\mid  \mathcal{D}_{1:b-1}^{(k)},\mathcal{D}^{(-k)}\big]\\
        &=\E\big[|\hat \Lambda_{b,1:t}^{(k)}-\Lambda_{b,1:t}^{\ast}| \mid \mathcal{D}_{1:b-1}^{(k)}, \mathcal{D}^{(-k)}\big].
    \end{align*}
    The last inequality is due to the product-comparison \cref{lem:prod_comparison}) and the final equality holds by the rate \cref{asn-rate}.\\

    \noindent Next, we bound the second term. Note that $\tilde\Gamma(\tilde O_{b,i};\hat{\bm\Lambda}^{(-k)},\hat{\bm\eta}_m^{(-k)}) - \tilde\Gamma(\tilde O_{b,i};\bm\Lambda^\dagger,\bm\eta_m)$ decomposes as (see \cref{lem:diff_decomposition})
    \begin{align*}
        \tilde\Gamma(\tilde O_{b,i};\hat{\bm\Lambda}^{(-k)},\hat{\bm\eta}_m^{(-k)})-\tilde\Gamma(\tilde O_{b,i};\bm\Lambda^\dagger,\bm\eta_m)=\Delta m_0^{(-k)}&-\sum_{t=1}^{T}\frac{\tilde C_{i,1:t}}{\hat\Lambda_{1:t}^{(-k)}\Lambda_{1:t}^\dagger}\Delta\Lambda_{1:t}^{(-k)}(\hat m_t^{(-k)}-\hat m_{t-1}^{(-k)})\\
        &-\sum_{t=1}^{T}\frac{\tilde C_{i,1:t}}{\Lambda_{1:t}^\dagger}\Delta(m_t^{(-k)}-m_{t-1}^{(-k)}).
    \end{align*}
    Then, by applying $(a+b+c)^2\leq 3(a^2+b^2+c^2)$ and the Cauchy-Schwarz inequality, we have
    \begin{align*}
        &\E\left[\left(\tilde\Gamma(\tilde O_{b,i};\hat{\bm\Lambda}^{(-k)},\hat{\bm\eta}_m^{(-k)})-\tilde\Gamma(\tilde O_{b,i}; \bm\Lambda^\dagger, \bm\eta_m)\right)^2 \mid \mathcal{D}_{1:b-1}^{(k)}, \mathcal{D}^{(-k)} \right] \\
        &\qquad \le 3\E\left[(\Delta m_0^{(-k)})^2 \mid \mathcal{D}_{1:b-1}^{(k)}, \mathcal{D}^{(-k)}\right] \\
        &\qquad +3T\sum_{t=1}^{T}\E\left[\left(\frac{\tilde C_{i,1:t}}{\hat\Lambda_{1:t}^{(-k)}\Lambda_{1:t}^\dagger}\Delta\Lambda_{1:t}^{(-k)}(\hat m_{t}^{(-k)}-\hat m_{t-1}^{(-k)})\right)^2 \mid \mathcal{D}_{1:b-1}^{(k)}, \mathcal{D}^{(-k)}\right] \\
        &\qquad + 3T\sum_{t=1}^{T}\E\left[\left(\frac{\tilde C_{i,1:t}}{\Lambda_{1:t}^\dagger}\Delta(m_t^{(-k)}-m_{t-1}^{(-k)})\right)^2 \mid \mathcal{D}_{1:b-1}^{(-k)}, \mathcal{D}^{(-k)}\right].
    \end{align*}
    Since $\hat \Lambda_{1:t}^{(-k)}, \Lambda_{1:t}^\dagger$ and $\hat m_t^{(-k)}-\hat m_{t-1}^{(-k)}$ are bounded by \cref{asn-seq_overlap,asn-bdd_estimators_adaptive}, one can simplify above as
    \begin{align*}
        &\E\left[\left(\tilde\Gamma(\tilde O_{b,i};\hat{\bm\Lambda}^{(-k)},\hat{\bm\eta}_m^{(-k)})-\tilde\Gamma(\tilde O_{b,i}; \bm\Lambda^\dagger, \bm\eta_m)\right)^2 \mid \mathcal{D}_{1:b-1}^{(k)}, \mathcal{D}^{(-k)} \right] \\
        &\qquad = \mathcal{O}\left(\E\left[(\Delta m_0^{(-k)})^2+\sum_{t=1}^{T}(\Delta \Lambda_{1:t}^{(-k)})^2 +\sum_{t=1}^{T}(\Delta(m_t^{(-k)}-m_{t-1}^{(-k)}))^2 \mid \mathcal{D}_{1:b-1}^{(k)}, \mathcal{D}^{(-k)}\right]\right),
    \end{align*}
    which is the same reduction that we obtained in the empirical process term proof of the non-adaptive estimator (\cref{prop-emp-process-term-non-adaptive}). The same argument applies to the proof, just by additionally conditioning on $\mathcal{D}_{1:b-1}^{(k)}$. Therefore, we omit the proof and conclude that it is $o_p(1)$. As a result, the conditional second moment of $\varphi_{b,i}^{(k)}$ is $o_p(1)$. Then it follows that 
    \begin{align*}
        &\Var\left(\frac{1}{\sqrt{n_{b,k}}}\sum_{i:(b, i)\in\mathcal{I}_k}\left(\varphi_{b,i}^{(k)}-\E[\varphi_{b,i}^{(k)} \mid \mathcal{D}_{1:b-1}^{(k)}, \mathcal{D}^{(-k)}]\right)\bigg\vert \mathcal{D}_{1:b-1}^{(k)}, \mathcal{D}^{(-k)}\right)\\ 
        &\qquad = \frac{1}{n_{b,k}}\sum_{i:(b, i)\in\mathcal{I}_k}\E\left[(\varphi_{b,i}^{(k)})^2 \mid \mathcal{D}_{1:b-1}^{(k)}, \mathcal{D}^{(-k)}\right] = o_p(1),
    \end{align*}
    and by conditional Chebyshev's inequality, 
    \begin{align*}
        &P\left[\bigg\vert \frac{1}{\sqrt{n_{b,k}}}\sum_{i:(b, i)\in\mathcal{I}_k}\left(\varphi_{b,i}^{(k)}-\E[\varphi_{b,i}^{(k)}\mid \mathcal{D}_{1:b-1}^{(k)}, \mathcal{D}^{(-k)}]\right) \bigg\vert > \epsilon \Bigg \vert \mathcal{D}_{1:b-1}^{(k)}, \mathcal{D}^{(-k)} \right] \\
        &\qquad \leq \frac{\Var\left(\frac{1}{\sqrt{n_{b,k}}}\sum_{i:(b, i)\in \mathcal{I}_k}\left(\varphi_{b,i}^{(k)}-\E[\varphi_{b,i}^{(k)}\mid \mathcal{D}_{1:b-1}^{(k)},\mathcal{D}^{(-k)}]\right)\mid \mathcal{D}_{1:b-1}^{(k)}, \mathcal{D}^{(-k)}\right)}{\epsilon^2}=o_p(1),
    \end{align*}
    which is the conditional statement \eqref{eqn:pf_cvg_prob_adaptive} required above. This concludes the proof.

\end{proof}

\begin{proposition}[Drift term] \label{prop-drift-adaptive}
With $\varphi_{b,i}^{(k)}:=\tilde \Gamma(O_{b,i};\hat{\bm\Lambda}^{(-k)}, \hat{\bm\eta}_m^{(-k)})-\tilde \Gamma(\tilde O_{b,i};\bm\Lambda^\dagger, \bm\eta_m)$ for $(b,i)\in \mathcal{I}_k$, the drift term is $o_p(1)$, i.e,
\[
    \sum_{k=1}^{K}\sum_{b=1}^{M}\sqrt{\frac{n_{b,k}}{n}}\left\{\frac{1}{\sqrt{n_{b,k}}}\sum_{i:(b,i)\in\mathcal{I}_k}\E\left[\varphi_{b,i}^{(k)} \mid \mathcal{D}_{1:b-1}^{(k)},\mathcal{D}^{(-k)}\right]\right\}=o_p(1).
\]
\end{proposition}
\begin{proof}[Proof of \cref{prop-drift-adaptive}]
    Fix fold $k$. Batch $b$'s first-order terms in \cref{lem:mean_diff_decomposition} include the factor $\Lambda^\ast_{b,1:t}/\Lambda^\dagger_{1:t}$; since $\sum_b\kappa_b\Lambda^\ast_{b,1:t}=\Lambda^\dagger_{1:t}$, they cancel after summing over batches. It remains to show, for each batch $b$, that the second-order remainder of $\frac{1}{\sqrt{n_{b,k}}}\sum_{i:(b, i)\in\mathcal{I}_k}\E[\varphi_{b,i}^{(-k)} \mid \mathcal{D}_{1:b-1}^{(k)},\mathcal{D}^{(-k)}]$ is $o_p(1)$.

    \noindent To show this, we decompose $\E[\tilde\Gamma(O_{b,i};\hat{\bm\Lambda}^{(-k)},\hat{\bm\eta}_m^{(-k)})-\tilde\Gamma(\tilde O_{b,i};\bm\Lambda^\dagger,\bm\eta_m)]$ following the same steps as in \cref{lem:mean_diff_decomposition}. The remaining difference arises in the annotation-probability error term. Since we are now comparing the realized outcome $O$ against the counterfactual outcome $\tilde O$, the weight difference $C_{1:t}/\hat\Lambda_{1:t}^{(-k)}-\tilde C_{1:t}/\Lambda_{1:t}^\dagger$ takes the place of the difference $C_{1:t}/\hat\Lambda_{1:t}^{(-k)}-C_{1:t}/\Lambda_{1:t}^\dagger$ used there. i.e,
    \[
        \text{(non-adaptive)}\quad -\frac{C_{1:t}}{\hat\Lambda_{1:t}^{(-k)}\Lambda_{1:t}^\dagger}\Delta\Lambda_{1:t}^{(-k)}=\frac{C_{1:t}}{\hat\Lambda_{1:t}^{(-k)}}-\frac{C_{1:t}}{\Lambda_{1:t}^\dagger}  \quad \Longrightarrow \quad \frac{C_{1:t}}{\hat\Lambda_{1:t}^{(-k)}}-\frac{\tilde C_{1:t}}{\Lambda_{1:t}^\dagger} \quad \text{(batch-adaptive)}
    \]
    Even with this modification, one can reach the same conclusion under one additional annotation product error rate assumption. This follows from the decomposition
    \[
        \frac{C_{1:t}}{\hat\Lambda_{1:t}^{(-k)}}-\frac{\tilde C_{1:t}}{\Lambda_{1:t}^\dagger}
        =\frac{C_{1:t}-\tilde C_{1:t}}{\hat\Lambda_{1:t}^{(-k)}}
        +\tilde C_{1:t}\left(\frac{\Lambda_{1:t}^\dagger-\hat\Lambda_{1:t}^{(-k)}}{\hat{\Lambda}_{1:t}^{(-k)}\Lambda_{1:t}^\dagger}\right),
    \]
    which yields the annotation error term as follows: $\forall i\in(b,i)\in\mathcal{I}_k$,
    \begin{align*}
        &\left\vert\E\left[\frac{C_{i,1:t}}{\hat\Lambda_{1:t}^{(-k)}}-\frac{\tilde C_{i,1:t}}{\Lambda_{1:t}^\dagger}\mid \mathcal{D}_{1:b-1}^{(k)}, \mathcal{D}^{(-k)}\right]\right\vert  \\
        &\qquad \leq \left\vert\E\left[\frac{C_{i,1:t}-\tilde C_{i,1:t}}{\hat\Lambda_{1:t}^{(-k)}} \mid \mathcal{D}_{1:b-1}^{(k)},\mathcal{D}^{(-k)}\right]\right\vert+\left\vert\E\left[\tilde C_{i,1:t}\left(\frac{\Lambda_{1:t}^\dagger-\hat\Lambda_{1:t}^{(-k)}}{\hat\Lambda_{1:t}^{(-k)}\Lambda_{1:t}^\dagger}\right)\mid \mathcal{D}_{1:b-1}^{(k)},\mathcal{D}^{(-k)}\right]\right\vert \\
        &\qquad=\mathcal{O}\left(\left\vert\E\left[\hat\Lambda_{b,1:t}^{(k)}-\Lambda_{b,1:t}^\ast \mid \mathcal{D}_{1:b-1}^{(k)},\mathcal{D}^{(-k)}\right]\right\vert+\left\vert\E\left[\hat\Lambda_{1:t}^{(-k)}-\Lambda_{1:t}^\dagger\mid \mathcal{D}_{1:b-1}^{(k)},\mathcal{D}^{(-k)}\right]\right\vert\right) \\
        &\qquad =\mathcal{O}\left(\left\vert\E\left[\Delta\Lambda_{b,1:t}^{(k)}\mid \mathcal{D}_{1:b-1}^{(k)},\mathcal{D}^{(-k)}\right]\right\vert+\left\vert\E\left[\Delta\Lambda_{1:t}^{(-k)}\mid \mathcal{D}_{1:b-1}^{(k)},\mathcal{D}^{(-k)}\right]\right\vert\right),
    \end{align*}
    where the first equality follows from three facts: (i) $\hat\Lambda_{1:t}^{(-k)}\ge c_{\hat\lambda}^t$ and $\Lambda_{1:t}^\dagger\ge c_\lambda^t$ (\cref{asn-seq_overlap,asn-bdd_estimators_adaptive}), used to bound $1/\hat\Lambda_{1:t}^{(-k)}$ and $1/(\hat\Lambda_{1:t}^{(-k)}\Lambda_{1:t}^\dagger)$ above by a constant; (ii) $|\tilde C_{i,1:t}|\le1$; and (iii) the identity $\E[C_{i,1:t}-\tilde C_{i,1:t}\mid\mathcal D_{1:b-1}^{(k)},\mathcal D^{(-k)}]=\hat\Lambda_{b,1:t}^{(k)}-\Lambda_{b,1:t}^\ast$ from the coupling construction. This introduces an additional annotation-probability estimation error, $\Delta\hat\Lambda_{b,1:t}^{(k)}$, that does not appear in the drift term analysis of the non-adaptive estimator, necessitating the additional product error rate assumption
    \[
        \sum_{j\leq t}\|\Delta\hat\lambda_{b,j}^{(k)}\|_2\left(\|\Delta b_t^{(-k)}\|_2+\|\Delta Q_{t+1}^{(-k)}\|_2+\|\Delta \mu_t^{(-k)}\|_2+\|\Delta \mu_{t+1}^{(-k)}\|_2\right)=o_p(1)
    \]
    for all $1\leq b\leq M$, $1\le k\leq K$, and $1\leq t\leq T$. We omit the remaining proof to avoid redundancy and conclude that the drift term is $o_p(1)$.
\end{proof}

\begin{proposition}[Oracle term] \label{prop-oracle_adaptive}
The oracle term converges to $\mathcal{N}(0, \sigma_{\ast}^2)$ in distribution, where $\sigma_{\ast}^{2}=\Var(m_0)+\sum_{t=1}^{T}\E[\Var(m_t\mid \mcF_{t-1})/\Lambda_{1:t}^\dagger]$, i.e,
\[
    \sum_{b=1}^{M}\sqrt{\frac{n_b}{n}}\left\{\frac{1}{\sqrt{n_b}}\sum_{i\in\mathcal{D}_b}\left(\tilde\Gamma(\tilde O_{b,i};\bm\Lambda^\dagger, \bm\eta_m)-\Phi^{\pi_e}\right)\right\} \Longrightarrow \mathcal{N}(0, \sigma_{\ast}^2).    
\]
\end{proposition}
\begin{proof}
    Note that the inner term converges in distribution to $\mathcal N(0,\sigma_{b,\ast}^2)$, i.e.,
    \[
        \frac{1}{\sqrt{n_b}}\sum_{i\in\mathcal{D}_b}\left(\tilde\Gamma(\tilde O_{b,i};\bm\Lambda^\dagger, \bm\eta_m)-\Phi^{\pi_e}\right)\Longrightarrow\mathcal{N}(0,\sigma_{b,\ast}^2), \quad \sigma_{b,\ast}^2=\Var (m_0)+\sum_{t=1}^{T}\E\left[\frac{\Lambda_{b,1:t}^\ast\Var(m_t\mid \mcF_{t-1})}{(\Lambda_{1:t}^\dagger)^2}\right].
    \]
    Aggregating across batches then yields the following due to $n_b/n\to \kappa_b$, followed by Slutsky's lemma:
    \begin{align*}
         \sum_{b=1}^{M}\sqrt{\frac{n_b}{n}}\left\{\frac{1}{\sqrt{n_b}}\sum_{i\in\mathcal{D}_b}\left(\tilde\Gamma(\tilde O_{b,i};\bm\Lambda^\dagger, \bm\eta_m)-\Phi^{\pi_e}\right)\right\}
        \Longrightarrow \sum_{b=1}^{M}\sqrt{\kappa_b}\mathcal{N}(0, \sigma_{b,\ast}^2).
    \end{align*}
    Since batches are mutually independent as they do not overlap, limiting distribution is a mixture of independent normal random variables with mean zero and variance equal to the weighted sum of the batch variances $\sum_{b=1}^{M}\kappa_b \sigma_{b,\ast}^2$:
    \begin{align*}
        \sigma_\ast^2&=\sum_{b=1}^{M}\kappa_b\left(\Var(m_0)+\sum_{t=1}^{T}\E\left[\frac{\Lambda_{b,1:t}^\ast\Var(m_t\mid\mcF_{t-1})}{(\Lambda_{1:t}^\dagger)^2}\right]\right) =\Var(m_0)+\sum_{t=1}^{T}\E\left[\frac{\Var(m_t\mid\mcF_{t-1})}{\Lambda_{1:t}^\dagger}\right],
    \end{align*}
    where the second equality is the definition $\Lambda_{1:t}^\dagger=\sum_b\kappa_b\Lambda_{b,1:t}^\ast$ of the implemented pooled design. This concludes the proof.
\end{proof}

\subsection{Batch-adaptive asymptotic normality with realized annotation probabilities}

\begin{proposition}[Batch-adaptive asymptotic normality with realized annotation probabilities]
\label{prop:clt_adaptive_endpoint}
Under \cref{setting-mdp,asn-seq-igno-stilde,asn-seq_overlap,asn-anno-ignorability} and bounded rewards (\Cref{asn-bdd-reward}), use the independent-fold batch-adaptive experiment and cross-fitting scheme of \Cref{thm:clt_adaptive}. In \cref{eqn:pooled_adaptive}, use the \textit{realized} annotation probabilities $\hat{\bm\Lambda}_b^{(k)}$ and set $\hat m_T^{(-k)}=\hat\Gamma_T^{\pi_e,(-k)}$ as in \Cref{prop:approximate-projection-bias}.

Assume condition~\textup{(R1)} of \Cref{asn-rate_adaptive}, the bounds in \Cref{asn-bdd_estimators_adaptive} for the fitted state-action functions and recorded probabilities, and uniformly bounded predictions. For every batch and fold, require
consistency of $\hat\mu_t$, $\hat Q_t$, $\hat\lambda_{b,t}$, and $\hat m_t$ for their respective limits, which may differ from oracle nuisances.
Then, provided $\sigma_\dagger^2>0$,
\begin{align}
\sqrt n(\hat\psi_{\mathrm{ad}}-\Phi^{\pi_e})&\Longrightarrow\mathcal N(0,\sigma_\dagger^2),\nonumber
\\
&\sigma_\dagger^2=\sum_{b=1}^M\kappa_b\Var\!\left[
\tilde\Gamma_T^{\pi_e}(\tilde O_b;\bm\Lambda_b^\ast,\bm\eta_m^\dagger)\right].
\label{eq:batch-variance}
\end{align}
\end{proposition}

\begin{proof}[Proof of \cref{prop:clt_adaptive_endpoint}]
\textbf{Drift term.}
Condition on $\mathcal D^{(-k)}$ and $\mathcal D_{1:b-1}^{(k)}$. By \cref{lem:collapse} and the endpoint equality, the conditional mean of a fitted score is $\E_{\pi_b}[\hat\Gamma_T^{\pi_e,(-k)}\mid\mathcal D^{(-k)}]$.

By \cref{prop:approximate-projection-bias}, its difference from $\Phi^{\pi_e}$ equals
\[
\sum_{t=1}^T\gamma^{t-1}\E_{\pi_b}\!\left[
(\hat\mu_t^{(-k)}-\mu_t)
\{\gamma(\hat V_{t+1}^{(-k)}-V_{t+1})-(\hat Q_t^{(-k)}-Q_t)\}
\mid\mathcal D^{(-k)}\right].
\]
Cauchy--Schwarz, \cref{lem-v-q-error}, and condition~\textup{(R1)} of \cref{asn-rate_adaptive} bound this by $o_p(n^{-1/2})$.

\textbf{Empirical process term.}
Compare the fitted score with the score using $m^\dagger$, but with the same recorded annotation weights. Overlap bounds the conditional second moment of their difference by a constant times
\[
\sum_{t=0}^{T-1}\|\hat m_t^{(-k)}-m_t^\dagger\|_2^2
+\|\hat\Gamma_T^{\pi_e,(-k)}-\Gamma_T^{\pi_e}\|_2^2=o_p(1).
\]
Consistency of the full-state nuisance estimates,
  together with
  appropriate boundedness conditions, implies
  $\|\hat\Gamma_T^{\pi_e}-\Gamma_T^{\pi_e}\|_{L_2(P_b)}
  =o_p(1)$. Conditional on training and preceding same-fold batches, current trajectories and their annotation draws are independent. The conditional Chebyshev argument in \cref{prop-emp-process-adaptive} therefore makes the centered root-$n$ empirical difference $o_p(1)$. Summing over a fixed number of batches and folds gives the result.

\textbf{Oracle term.}
 The oracle term follows by the batch-wise CLT argument
  in
  \cref{prop-oracle_adaptive}. With recorded annotation
  probabilities
  and the full-data endpoint, telescoping preserves
  conditional mean
  $\Phi^{\pi_e}$ even when the intermediate predictions
  are approximate.
  Under the stated convergence, boundedness and overlap
  conditions,
  the limiting variance is \eqref{eq:batch-variance}.
  Combining this with the drift and empirical-process
  bounds
  proves asymptotic normality.
  
\textbf{Variance estimation.}
At a limiting batch design, summation by parts yields
\[
\tilde\Gamma-\Gamma_T^{\pi_e}
=\sum_{t=1}^T\left(\frac{C_{1:t}}{\Lambda_{b,1:t}^\ast}
-\frac{C_{1:t-1}}{\Lambda_{b,1:t-1}^\ast}\right)
(\Gamma_T^{\pi_e}-m_{t-1}^\dagger),
\]
with $C_{1:0}=\Lambda_{b,1:0}^\ast=1$. Conditional annotation-weight orthogonality gives
\[
\begin{aligned}
\Var[\tilde\Gamma(\tilde O_b;\bm\Lambda_b^\ast,\bm\eta_m^\dagger)]
={}&\Var(\Gamma_T^{\pi_e})+\sum_{t=1}^T\E\!\left[
\left(\frac1{\Lambda_{b,1:t}^\ast}-\frac1{\Lambda_{b,1:t-1}^\ast}\right)
(\Gamma_T^{\pi_e}-m_{t-1}^\dagger)^2\right].
\end{aligned}
\]
A bounded martingale law of large numbers gives convergence of oracle first and second empirical moments. The $L_2$ comparison above transfers these limits to fitted scores. All batch means converge to $\Phi^{\pi_e}$, so their pooled empirical variance consistently estimates $\sigma_\dagger^2$. 
\end{proof}

\subsection{Auxiliary lemmas}

\begin{lemma}[Weight collapse]\label{lem:collapse}
Under \cref{asn-anno-ignorability}, for every integrable $\mcF_T$-measurable $X$ and $t=1,\dots,T$,
\[
\E[w_tX]=\E[X],\qquad \E[\hat w_tX]=\E\Big[\tfrac{\Lambda_{1:t}}{\hat\Lambda_{1:t}}X\Big].
\]
\end{lemma}
\begin{proof}
For $\mcF_T$-measurable $Y$, conditioning on $\sigma(\mcF_T,C_1,\dots,C_{s-1})$ and using \cref{asn-anno-ignorability} on $\{C_{1:s-1}=1\}$ (both sides vanish on the complement) gives $\E[C_{1:s}Y]=\E[C_{1:s-1}\lambda_sY]$. Since $\lambda_sY$ is again $\mcF_T$-measurable, iterating from $s=t$ down to $s=1$ with $Y=X/\Lambda_{1:t}$ (resp.\ $Y=X/\hat\Lambda_{1:t}$; both denominators are $\mcF_{t-1}$-measurable) yields the claims.
\end{proof}

\begin{lemma}[Product comparison]\label{lem:prod_comparison}
    For $i=1,\dots, n$, suppose that $a_i, b_i \in \mathbb{R}$, with $|a_i|\leq M$ and $|b_i| \leq M$. Then,
    \[
        |\prod_{i=1}^{n} a_i - \prod_{i=1}^{n} b_i| \leq M\sum_{i=1}^{n} |a_i-b_i|.
    \]
\end{lemma}
\begin{proof}
    For $n=2$, we merely have to write
    \[
        a_1a_2-b_1b_2 = a_1(a_2-b_2)+(a_1-b_1)b_2.
    \]
    Finish by taking absolute values and using the fact that all $|a_i|$ and $|b_i|$ are bounded $M$. For general $n$, use induction.
\end{proof}

\section{Details on experiments}\label{apx-experiments}

First, we outline some variants in implementation that we consider when relevant. Such implementation variants can improve empirical performance or adapt our general method to the challenges or opportunities of specific applications. 
\begin{revision}
\paragraph{Practical conditional variance estimates: plug-in proxies for $\sigma_t^2$.}
The optimal probabilities of \Cref{sec:annotation_prob_derivation} depend on $\sigma_t^2(\mcF_{t-1})$, which is unknown and must be estimated from data. In every experiment the implemented design replaces it by a \emph{design signal} $\hat s_t^2(\mcF_{t-1})$, any nonnegative $\mcF_{t-1}$-measurable quantity, and sets $\lambda_t\propto\hat s_t$ subject to the floor and prefix constraints of \eqref{eqn-opt_prob}. By \Cref{prop:seq-annotation-unbiased}, different choices of signals only affect efficiency through how closely $\hat s_t$ tracks $\sigma_t$. Different options include: 
\begin{itemize}
\item \emph{Residual regression.} The standard way to estimate the conditional variance $\hat s_t^2$ is by regressing squared residuals on $\mcF_{t-1}$, which must be learned from the first batch.
    \item 
\emph{Binary rewards and a link function.} When outcomes are known to be binary, we can leverage the conditional variance expression for a Bernoulli outcome, $p(1-p)$. For a binary reward, $\hat s_t^2=\hat\mu_t^2\,\hat b_t(1-\hat b_t)$, where $\hat b_t$ estimates $b_t=\E[R_t\mid\mcF_{t-1}]$ on the first batch of \Cref{subsec:batch-adaptive-annotation} (when the revealed reward also determines the next state, the same two-point variance is applied to the full stage increment; see \Cref{apx:arena}). 
\item \emph{Always-observed silver label variance.} Sometimes an always-observed silver or surrogate label $\tilde R_t$ is available for every unit, such as in our casenote annotation application with zero-shot LLM annotation. Further, we may have an ensemble of the silver labels available, such as LLM annotations from different models. 
Then we can use silver-label disagreement from zero-shot models as a variance proxy.
Such zero-shot variance proxies are therefore $\mcF_0$-measurable. 
\end{itemize}
\paragraph{Implemented allocation protocol.}
The experiments run a simplified variant of \Cref{alg:batch-adaptive-seq-annotation} that differs from it in three respects. We use the known annotation probabilities in $\tilde\Gamma$ rather than re-estimation, which helps in small pilots. Second, the fold-$k$ conditional variance $\hat s_t$ is fit on first-batch units \emph{outside} fold $k$ so that we can use out-of-fold data, $(K-1)/K$ of the first batch rather than $1/K$. The cross-fold design does introduce additional dependence beyond the independence structure assumed in \Cref{sec:analysis}. We see this may have led to slightly worsened coverage, although variance improvements persist. Third, the second-batch probabilities are the optimized allocation of the \emph{residual} budget $B-\kappa_1$, which clips the optimal annotation probabilities to feasibility after an initial pilot. 

\end{revision}

\subsection{Simulated data}
\paragraph{Data-generating process.}
Let $p=5$. For each unit, $S_1 \sim N(0,I_p),
  S_2 = 0.5S_1 + \sqrt{0.75}\,\eta$.
Define the truncated logistic map $\ell_{0.05}(z) = 0.05 + 0.90\{1+\exp(-z)\}^{-1}$.
The behavior policy is
\begin{align*}
  p_1
    &= \Pr(A_1=1\mid S_1)
     = \ell_{0.05}\{1.5(0.5 - S_{1,2} + S_{1,3})\}, \\
  p_2
    &= \Pr(A_2=1\mid S_1,A_1,X_2)
     = \ell_{0.05}\{1.5(0.5 - X_{2,2} + X_{2,3} + 0.35A_1)\}.
\end{align*}
Rewards are generated as
\begin{align*}
  R_1
    &= 0.5(S_{1,1}-2S_{1,2}) + 1.1A_1 + \epsilon_1, \\
  R_2
    &= 0.5(S_{2,1}-2S_{2,2}) + 2.2A_2
       + 0.5A_1A_2 + \epsilon_2.
\end{align*}
The high-variance arm is the less common control arm at both stages:
\begin{align*}
  \Var(\epsilon_t\mid S_t,A_t=0)
    &= 4.0\max\{3.5+0.3\cos(S_{t,3}),0.05\}, \\
  \Var(\epsilon_t\mid S_t,A_t=1)
    &= \max\{1.3+0.4\sin(S_{t,1}),0.05\},
  \qquad t\in\{1,2\}.
\end{align*}
Nuisance functions are fit by kernel ridge regression with $5$-fold cross-fitting. Annotation is stagewise: after a pilot that annotates full trajectories with probability $p$, the adaptive batch draws $C_1\sim\mathrm{Bern}(\lambda_1)$ and $C_2\mid C_1=1\sim\mathrm{Bern}(\Lambda_{1:2}/\lambda_1)$, with $(\lambda_1,\Lambda_{1:2})$ the joint two-stage allocation of \Cref{thm-twostage-optimal-allocation-T=2} for the target regime, including its binding case.

\paragraph{Coverage and all four policies.} \Cref{fig:synth-all-regimes} extends \Cref{fig:simulated} to coverage and to each of the four unpersonalized policies $(A_1,A_2)=(d_1,d_2)$, each with its own stagewise design. Coverage of the true value is $0.91$--$0.96$ (adaptive) and $0.91$--$0.98$ (random) across policies and budgets. For every policy, from a $40\%$ budget the adaptive RMSE is within $5\%$ of the full-annotation RMSE, while random annotation remains $12$--$16\%$ above it at $80\%$. For the policy $(0,0)$, plug-in estimation (g-computation with kernel ridge regression) has lower RMSE than DRL, but its intervals undercover ($0.10$--$0.19$ under both designs), so we report DRL.

\begin{figure}[p]
  \centering
  \includegraphics[width=0.3\textwidth]{figs/synth_d00_rmse.pdf}\includegraphics[width=0.3\textwidth]{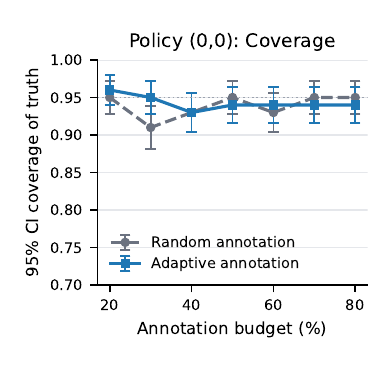}\includegraphics[width=0.3\textwidth]{figs/synth_d00_width.pdf}\\
  \includegraphics[width=0.3\textwidth]{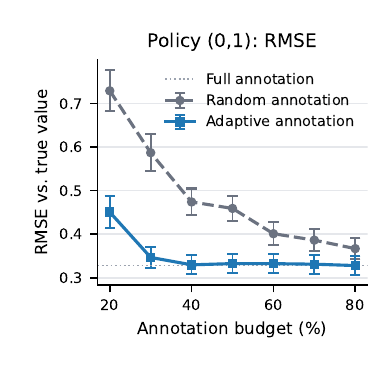}\includegraphics[width=0.3\textwidth]{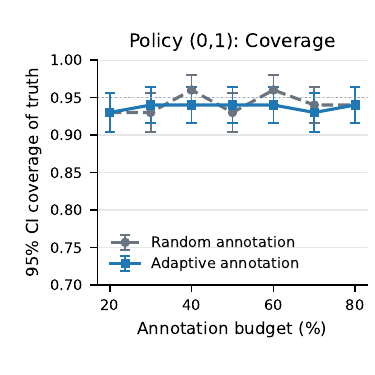}\includegraphics[width=0.3\textwidth]{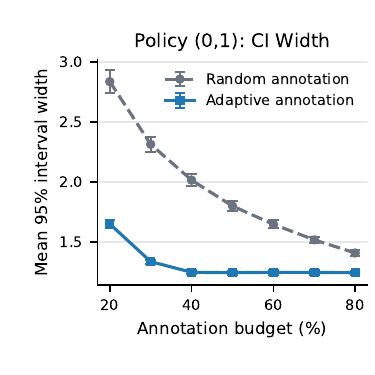}\\
  \includegraphics[width=0.3\textwidth]{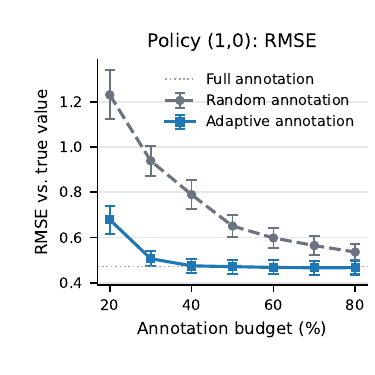}\includegraphics[width=0.3\textwidth]{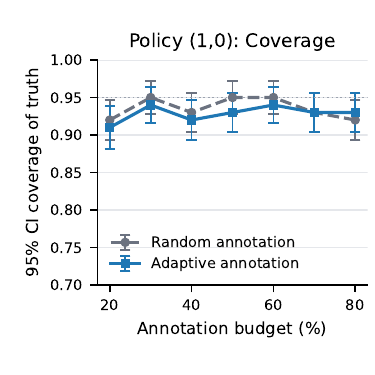}\includegraphics[width=0.3\textwidth]{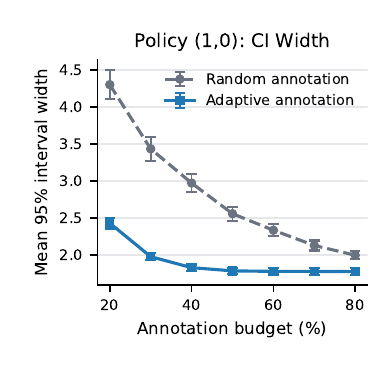}\\
  \includegraphics[width=0.3\textwidth]{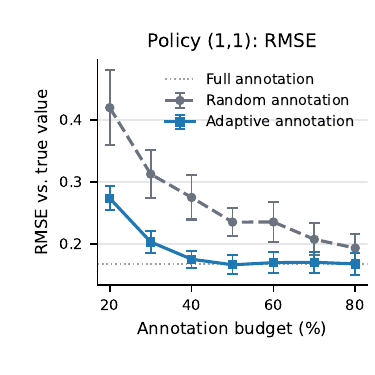}\includegraphics[width=0.3\textwidth]{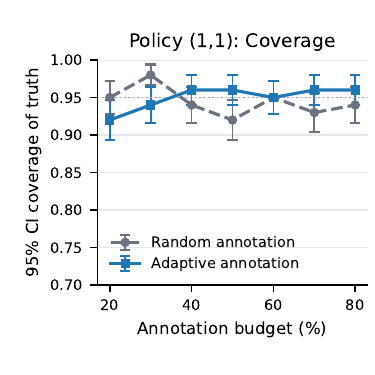}\includegraphics[width=0.3\textwidth]{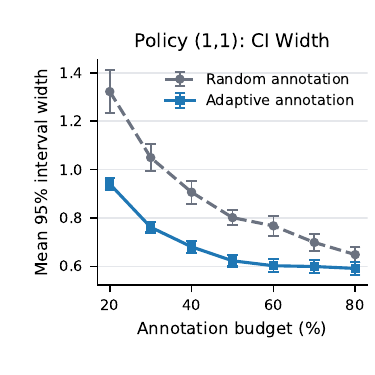}
  \caption{Simulated data, all four unpersonalized policies (rows), DRL estimator, $n=1000$, $100$ replications; the top row is the policy of \Cref{fig:simulated}. Left: RMSE against the true policy value (dotted line: full annotation); middle: coverage of that value by the model-based $95\%$ interval (dotted line at $0.95$); right: mean interval width. Error bars are Monte Carlo standard errors.}\label{fig:synth-all-regimes}
\end{figure}

\subsection{Casenote cohort}\label{apx:bg-descriptives}

\Cref{tab:bg-quintiles} gives the outreach quintiles that define the actions. A recorded contact is a logged outreach encounter, including attempts that did not reach the client (about half of recorded contacts); about $72\%$ of recorded contacts in stages 1--2 have a casenote. Quintiles are formed in each stage among all cohort clients with at least one recorded contact in that stage, so the target policy's one-quintile shift moves a client, for example, from $11$--$21$ to $22$--$52$ recorded contacts in stage 1. \Cref{tab:bg-descriptives} summarizes casenote volume, contact types, and the always-observed state $S_1'$ for the $775$ analysis clients. The stage-2 history adds the stage-1 outreach quintile and contact count and the stage-1 progress summaries.

\begin{table}[h]
\centering\small
\caption{Outreach quintiles: recorded contacts per six-month stage (per month in parentheses), and number of analysis clients in each quintile.}\label{tab:bg-quintiles}
\begin{tabular}{@{}lccccc@{}}
\toprule
 & Q1 & Q2 & Q3 & Q4 & Q5 \\
\midrule
Stage 1 contacts & 1--4 & 5--10 & 11--21 & 22--52 & $\geq 53$ \\
\quad per month & (0.2--0.7) & (0.8--1.7) & (1.8--3.5) & (3.7--8.7) & ($\geq 8.8$) \\
\quad clients & 128 & 140 & 149 & 176 & 182 \\
Stage 2 contacts & 1--4 & 5--10 & 11--20 & 21--48 & $\geq 49$ \\
\quad clients & 120 & 137 & 151 & 166 & 201 \\
\bottomrule
\end{tabular}
\end{table}

\begin{table}[h]
\centering\small
\caption{Analysis cohort ($n=775$): recorded contacts, casenotes, contact types, and the always-observed state $S_1'$ at cohort entry. Median (IQR) or \%. Cells with fewer than $10$ clients are suppressed.}\label{tab:bg-descriptives}
\begin{tabular}{@{}lcc@{}}
\toprule
 & Stage 1 & Stage 2 \\
\midrule
\multicolumn{3}{@{}l}{\emph{Recorded contacts per client}} \\
\quad per stage & 19 (7--50) & 19 (8--50) \\
\quad stages 1--2 & \multicolumn{2}{c}{43 (19--100)} \\
\quad total recorded contacts & 29{,}076 & 29{,}346 \\
\multicolumn{3}{@{}l}{\emph{Casenotes per client}} \\
\quad per stage & 12 (5--35) & 14 (5--37) \\
\quad stages 1--2 & \multicolumn{2}{c}{30 (13--70)} \\
\quad total casenotes & 20{,}233 & 21{,}966 \\
\multicolumn{3}{@{}l}{\emph{Contact type (\% of recorded contacts)}} \\
\quad reached, on street & 41.0 & 44.6 \\
\quad reached, other (housing site, office, phone) & 6.2 & 7.0 \\
\quad attempted, not reached & 50.6 & 46.3 \\
\quad collateral (third party) & 2.3 & 2.1 \\
\midrule
\multicolumn{3}{@{}l}{\emph{State $S_1'$ at cohort entry}} \\
\quad years known to outreach program & \multicolumn{2}{c}{2.3 (0.6--4.5)} \\
\quad prior recorded contacts & \multicolumn{2}{c}{99 (9--257)} \\
\quad prior placement level 0 / 1 / 2 / 3 & \multicolumn{2}{c}{55.0 / 1.3 / 29.2 / 14.6} \\
\quad believed chronically homeless & \multicolumn{2}{c}{80.6} \\
\quad important documents on file & \multicolumn{2}{c}{31.0} \\
\quad service refusal recorded & \multicolumn{2}{c}{16.3} \\
\quad benefits recorded & \multicolumn{2}{c}{6.8} \\
\quad housing application recorded & \multicolumn{2}{c}{1.5} \\
\quad 311 call recorded & \multicolumn{2}{c}{3.7} \\
\quad male / female & \multicolumn{2}{c}{80.5 / 18.3} \\
\quad age $<30$ / $30$--$50$ / $>50$ / unknown & \multicolumn{2}{c}{4.5 / 37.0 / 46.3 / 12.1} \\
\quad White / Black / Asian / other race & \multicolumn{2}{c}{37.8 / 29.9 / 1.8 / 29.2} \\
\quad Hispanic or Latino / not / unknown & \multicolumn{2}{c}{30.3 / 47.4 / 21.2} \\
\quad two largest programs & \multicolumn{2}{c}{47.0 / 41.3} \\
\bottomrule
\end{tabular}
\end{table}

\subsection{Reconciled progress reward}\label{apx:reconciled}

The progress reward of a client in a six-month stage is the maximum progress label over the client's casenotes in that stage. In the experiments of \Cref{sec:experiments}, we ground-truthed the maximum-label casenote in each client-stage and a sample of near-maximum casenotes, and recomputed each stage's maximum over this set.

\paragraph{Blind review.} For each of the $1{,}488$ client-stages of the analysis cohort with a nonzero fine-tuned maximum, one casenote attaining the maximum (drawn uniformly among ties) was coded by an author with the study codebook, blind to all model labels, together with $100$ casenotes drawn at random from those labeled one level below their stage's maximum.

\subsection{Trajectory-wise versus stagewise annotation on the casenote data}\label{apx-bundle}
\begin{revision}
\paragraph{Trajectory-wise annotation as an additional constraint on allocation optimization.}
Forward monotone annotation is most relevant when reward variance differs from stage to stage, and can be learned. If instead reward variance is similar at every timestep, an alternative is to annotate whole trajectories. We give an example in the two-stage setting. Since $\lambda_1$ is $\mcF_0$-measurable, $\E[\sigma_2^2(\mcF_1)/\lambda_1(\mcF_0)]=\E[\E[\sigma_2^2(\mcF_1)\mid\mcF_0]/\lambda_1(\mcF_0)]$, and the restricted problem reduces to a single-stage Neyman problem on the summed increment variances:
\begin{equation}\label{eqn-bundle}
\min_{\lambda_1\in[c_\lambda,1]}\;\E\!\left[\frac{\sigma_1^2(\mcF_0)+\E[\sigma_2^2(\mcF_1)\mid\mcF_0]}{\lambda_1(\mcF_0)}\right]
\quad\text{s.t.}\quad \E[\lambda_1(\mcF_0)]\le B,
\end{equation}
The budget-normalized solution is:
\[
\lambda_1^{\rm traj}(\mcF_0)\;\propto\;\sqrt{\sigma_1^2(\mcF_0)+\E[\sigma_2^2(\mcF_1)\mid\mcF_0]}
\]
The solution is therefore suboptimal for the original problem, but still results in variance reduction compared to uniform allocation. 
\end{revision}
\begin{revision}
\Cref{fig:bg-bundle} compares the two designs on the progress task at equal expected label cost ($100$ replications of $600$ clients, otherwise as in \Cref{fig:bg}). The stagewise design uses its freedom --- it forgoes the second label for about $20\%$ of the clients it annotates at stage~1 at low budgets, $\Lambda_{1:2}/\lambda_1\approx0.64$, rising to $0.85$ at an $80\%$ budget --- but the two designs perform similarly overall: the RMSE ratio stagewise/trajectory-wise ranges from $0.89$ to $1.21$ across budgets $30$--$80\%$, and mean interval widths agree within $2\%$. The trajectory-wise priority only needs the design signal to rank clients, and a common miscalibration of the silver variance cancels in the normalization; the stagewise allocation additionally trades labels between stages within a client, so it needs the stage-1 and stage-2 signals to be calibrated relative to each other, which the within-window variance of LLM codes is not designed to be. Realizing the theoretical gap of \eqref{eqn-bundle} therefore requires a stage-calibrated variance signal (the binary link, a per-stage residual regression, or the silver variance rescaled by a per-stage constant fit on the first batch). The simulator of \Cref{apx-simulator}, whose silver labels are generated by one mechanism at both stages, illustrates this: there the stagewise design attains $0.80$--$0.99$ of the trajectory-wise RMSE against the true value (largest gain, $20\%$, at a $50\%$ budget) with $1$--$17\%$ narrower intervals.
\end{revision}
\begin{figure}[t]
  \centering
  \includegraphics[width=0.32\textwidth]{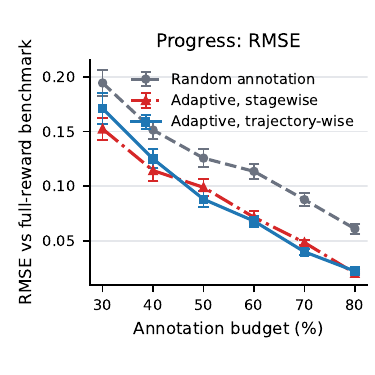}\includegraphics[width=0.32\textwidth]{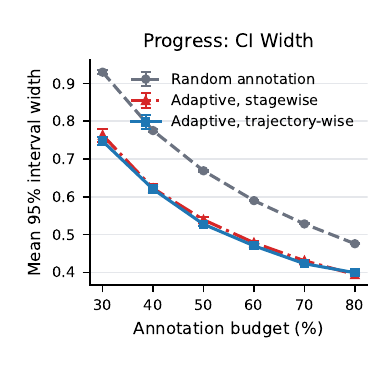}
  \caption{Casenote progress task: random annotation, the trajectory-wise design \eqref{eqn-bundle}, and the stagewise design of \Cref{thm-twostage-optimal-allocation-T=2}, all with the silver-variance signal and $\kappa_1=0$; $100$ replications of $600$ clients, error bars are Monte Carlo standard errors.}\label{fig:bg-bundle}
\end{figure}

\subsection{Simulator-grounded validation of the real-data protocols}\label{apx-simulator}

We repeat the two protocols of \Cref{fig:bg} on a simulator fit to the casenote cohort, where the value of the target policy is known, running the estimation pipeline unchanged on simulated cohorts written in the format of the real data.

\paragraph{Simulator.} The simulator is a monthly outreach MDP. The state is the compact state of the cohort's monthly panel: last recorded placement in $\{0,1,2,3\}$, last recorded progress bin, previous-month service bin, age--gender group, and prior-shelter and cumulative-contact bins. The action is the month's outreach count, binned as $\{1,2\text{--}3,4\text{--}7,8+\}$. Transitions of placement, progress bin and service bin are multinomial logistic regressions on random Fourier features of state and action ($300$ features; bandwidth chosen by held-out likelihood over a client split), fit on the $22{,}401$ client-months with positive outreach; months without outreach leave the state unchanged. Predicted and held-out change rates agree (placement $6.8\%$ vs.\ $6.0\%$; progress $54.8\%$ vs.\ $53.9\%$). Outreach intensity is drawn once per six-month stage from a multinomial logistic behavior policy $\pi_b$ fit on the real analysis population, stage~1 on the baseline covariates and stage~2 additionally on the stage-1 quintile, contact count and silver progress summaries, so that $\pi_b$ depends only on the tabular states $S_1',S_2'$ and \Cref{asn-seq-igno-stilde} holds by construction. Given a client's outreach quintile for the stage, the simulator first draws the number of contacts in that six-month stage from the empirical distribution of stage totals among cohort clients in the same quintile. It then spreads those contacts over the six months with a Dirichlet-multinomial draw, whose concentration parameter is chosen so that, as in the cohort, $45\%$ of client-months have no contact. Year~2 outreach follows $\pi_b$. Within a month, every contact is stamped with the client's placement status at the start of the month, except the last contact of the month, which is stamped with the updated status produced by that month's transition. Each contact also receives a gold progress label, drawn from the empirical distribution of real 2019 labels given the client's end-of-month progress bin; the last contact of every month is labeled, and each earlier contact is labeled independently with probability $0.70$, matching the cohort's label coverage. Silver labels are a misspecified linear imputation of the label from pre-annotation covariates, fit on the real 2019 labels and mixed with the label so that the silver--gold correlation matches the data ($0.83$ vs.\ $0.81$). Initial states are resampled from the $769$ cohort clients with a January 2019 panel state. 

Although prediction of placement changes and the type of progress is quite noisy, the simulated population matches the real data on aggregate population moments. These include stage contact counts (median $21$ vs.\ $19$), placement-improvement rate ($0.215$ vs.\ $0.225$), and the baseline placement and action distributions; stage rewards are somewhat higher and less dispersed ($R_1$ mean $2.34$, sd $1.02$, vs.\ $2.20$, $1.16$).

\paragraph{Protocol.} We draw a pool of $8{,}000$ clients and subsample $600$ for each of $100$ replications. The estimation protocol follows that used for the real data. 
The true target-policy value $\Phi^{\pi_e}$ is evaluated under the model from $100{,}000$ rollouts ($5.53$ total progress and $0.252$ placement-improvement probability; the behavior values are $4.68$ and $0.211$). Since $\pi_b$ is known, the pipeline uses $\pi_b(\cdot\mid S_t')$ in place of the estimated propensity, both in $\pi_e$ and in the DRL weights; the outcome models and the annotation design are estimated as on the real data. The experiment therefore validates the design, the outcome models and the inference, not propensity estimation.

\begin{figure}[t]
  \centering
  \includegraphics[width=0.32\textwidth]{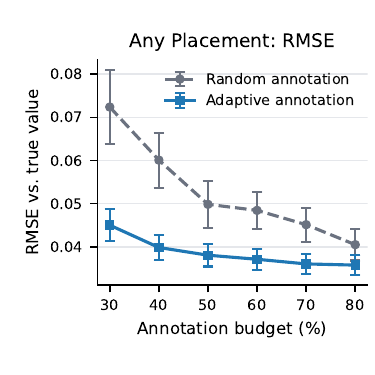}\includegraphics[width=0.32\textwidth]{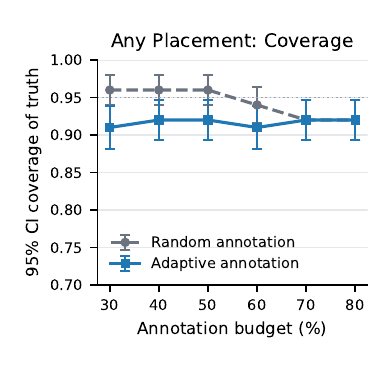}\includegraphics[width=0.32\textwidth]{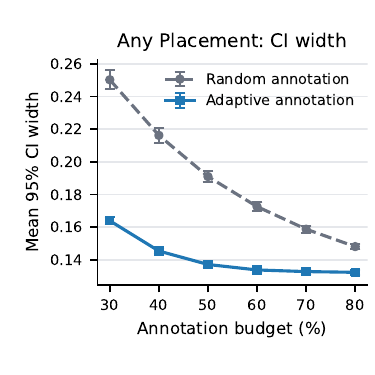}\\
  \includegraphics[width=0.32\textwidth]{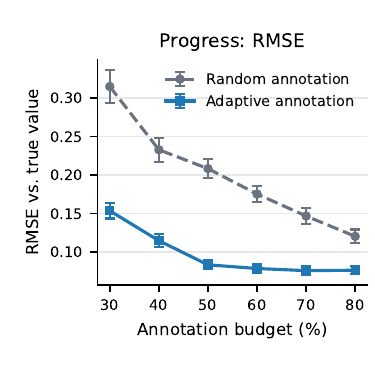}\includegraphics[width=0.32\textwidth]{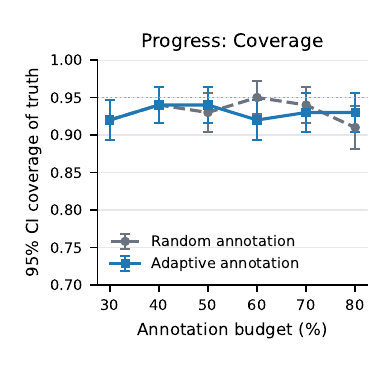}\includegraphics[width=0.32\textwidth]{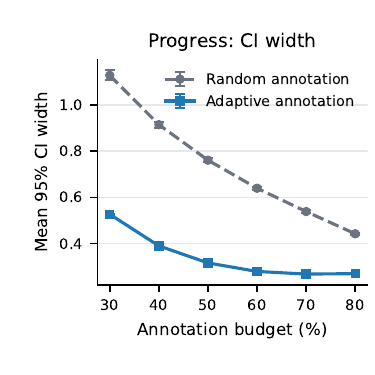}
  \caption{Simulator-grounded validation of the two protocols of \Cref{fig:bg} (DRL, $100$ replications of $600$ clients). Top: any placement improvement (binary-link signal, first batch $\kappa_1=0.3B$, $p=1$). Bottom: progress (silver-variance signal, $\kappa_1=0$). Left: RMSE against the true target-policy value; middle: coverage of that value by the model-based $95\%$ interval (dotted line at $0.95$); right: mean interval width. Error bars are Monte Carlo standard errors.}\label{fig:sim-validation}
\end{figure}

\paragraph{Results.} \Cref{fig:sim-validation} shows RMSE against $\Phi^{\pi_e}$, coverage of $\Phi^{\pi_e}$, and interval width. Coverage is close to nominal at every budget for both samplers: $0.92$--$0.96$ (random) and $0.91$--$0.92$ (adaptive) for placement, $0.91$--$0.95$ and $0.92$--$0.94$ for progress. Adaptive annotation reduces RMSE against the truth by $12$--$38\%$ (placement) and $37$--$60\%$ (progress), and interval width by $11$--$34\%$ and $39$--$58\%$. These match \Cref{fig:bg} for placement and exceed it for progress (interval width: $39$--$58\%$ vs.\ $17$--$20\%$). Plug-in intervals undercover here as well ($0.2$--$0.5$). 

\subsection{LMArena experiment details}\label{apx:arena}

\paragraph{Sessions and observations.}
The downloaded release contains $135{,}634$ battles dated April 17--July 24, 2025. We retain the earliest record for each session/order pair, removing $15$ duplicate records, and exclude $7{,}068$ sessions without a first vote. Of the $108{,}304$ eligible sessions, $10{,}196$ contain both first and second votes. Session identifiers determine the units; the release does not provide user identifiers, so distinct sessions cannot be guaranteed to come from distinct users.

Let $Z_1$ indicate that the release contains a second vote. This is observed metadata in $\mcF_0$, even though it describes continuation after the first action. The state model is
\[
S_1=\text{week},\qquad
S_2=\begin{cases}
(\text{week},R_1),& Z_1=1,\\
\dagger,&Z_1=0,
\end{cases}
\qquad S_3=\dagger.
\]
Set $R_2=V_2^{\pi_e}(\dagger)=Q_2^{\pi_e}(\dagger,a)=0$ on $\{Z_1=0\}$. The observed component $S_2'$ contains week and continuation status, while $R_1$ is withheld until the first annotation. We therefore distinguish $S_2$ from $S_2'$ even though the fitted state representation is tabular. This representation is a modeling restriction: it omits other prompt and session history.

\paragraph{Actions and policies.}
We split the $53$ models into $18$ open-weight and $35$ closed models. Actions are the three resulting pair categories. The category intervention changes their probabilities while retaining the behavior distribution of model identities within each category and week. Behavior probabilities are empirical category frequencies by week, estimated on the original $64{,}886$-session training partition and held fixed in the annotation experiment.

\paragraph{Full-data nuisance functions.}
We estimate $Q_2^{\pi_e}$ by the mean reward within each week, first-reward and second-action cell. We obtain $V_2^{\pi_e}$ by averaging under $\pi_e$, and fit $Q_1^{\pi_e}$ to $R_1+V_2^{\pi_e}(S_2)$. Cells with fewer than $20$ training sessions use the corresponding cell pooled across weeks. The stage-2 density ratio includes the change in the distribution of first rewards and recorded continuation:
\[
\mu_2^{\pi_e}(S_2,A_2)
=\frac{p_{\pi_e}(R_1,Z_1=1\mid\text{week})}
{p_{\pi_b}(R_1,Z_1=1\mid\text{week})}
\frac{\pi_e(A_2)}{\pi_b(A_2\mid\text{week})},\qquad Z_1=1.
\]
The fitted transition tables determine the first factor. Because it depends on $R_1$, this ratio is generally not $\mcF_0$-measurable. The implementation evaluates it separately at $R_1=0$ and $R_1=1$ when constructing the initial prediction.

\paragraph{Reward predictions and score.}
Reward predictors use response metadata, Skywork-Reward-V2-Llama-3.1-8B scores, and a fitted TF--IDF text score. We fit $b_1$ and $b_2$ using standardized logistic regression with regularization parameter $C=0.1$; the stage-2 predictor also uses the revealed $R_1$. Predictions are clipped to $[0.001,0.999]$. 

\paragraph{Allocation criterion.}
The design uses the two-stage constrained allocation solver with probability floor $0.01$ and second-stage cost $Z_1$, with the binary-link design signals of \Cref{sec:experiments}. At stage 2, $\hat s_2^2=\hat\mu_2^2\hat b_2(1-\hat b_2)$. At stage 1 the revealed $R_1$ also determines $S_2$, so the stage-1 increment $\hat\mu_1\{R_1+\hat V_2(S_2)\}-\hat\mu_2\hat Q_2$ is a function $h(R_1)$ of the binary $R_1$ alone given $\mcF_0$, with $h(r)=\hat\mu_1\{r+\hat V_2(r)\}-\hat\mu_2(r)\hat Q_2(r)$ evaluated at the state $(\text{week},r)$; its conditional variance is the two-point variance $\hat s_1^2=\hat b_1(1-\hat b_1)\{h(1)-h(0)\}^2$. The stage-1 signal omits the change in $\hat\mu_2\hat b_2$ with $R_1$, although the score retains it through $\hat m_1$ and the endpoint correction; $\hat s_1^2$ is therefore a surrogate for $\sigma_1^2=\Var(m_1\mid\mcF_0)$, affecting efficiency but not validity.

\paragraph{Sampling and fitting.}
We run $1{,}000$ replications at budget fractions $B\in\{0.3,0.4,0.5,0.6,0.7,0.8\}$ of all available votes. Each replication draws $21{,}685$ sessions without replacement and assigns them to five session folds. An initial uniform sample reveals all available votes with probability $p=0.3B$, where $B$ is the budget fraction. 


\paragraph{Evaluation.}
We recompute the full-annotation benchmark on each sampled set of sessions. RMSE is calculated from the difference between the partially and fully annotated estimates. Reported intervals have width $2(1.96)\,\hat{\operatorname{sd}}(\text{scores})/\sqrt n$. These intervals include sampling variation that is shared with the same-session benchmark, so coverage of that benchmark does not establish coverage of the population policy value. We report interval widths, not a coverage guarantee.


\paragraph{Results.}
\Cref{tab:arena} reports the quantities plotted in \Cref{fig:arena}. 
\begin{table}[h]
\centering\small
\begin{tabular}{@{}l cccc cccc@{}}
\toprule
& \multicolumn{4}{c}{equal exposure} & \multicolumn{4}{c}{tilt $\times1.5$} \\
\cmidrule(lr){2-5}\cmidrule(lr){6-9}
& \multicolumn{2}{c}{RMSE}$\times10^{-3}$ & \multicolumn{2}{c}{interval width$\times10^{-3}$} & \multicolumn{2}{c}{RMSE$\times10^{-3}$} & \multicolumn{2}{c}{interval width$\times10^{-3}$} \\
$B$ & random & adaptive & random & adaptive & random & adaptive & random & adaptive \\
\midrule
0.3 & 8.62 & 3.26 & 33.0 & 21.7 & 4.13 & 3.60 & 20.1 & 18.2 \\
0.4 & 5.33 & 2.41 & 26.4 & 19.9 & 3.06 & 2.74 & 17.1 & 16.1 \\
0.5 & 4.12 & 1.83 & 23.3 & 18.9 & 2.43 & 2.26 & 15.4 & 14.8 \\
0.6 & 3.19 & 1.44 & 21.2 & 18.4 & 2.00 & 1.79 & 14.2 & 13.8 \\
0.7 & 2.48 & 1.06 & 19.9 & 18.0 & 1.60 & 1.36 & 13.4 & 13.1 \\
0.8 & 1.87 & 0.76 & 18.9 & 17.8 & 1.25 & 0.99 & 12.8 & 12.6 \\
\bottomrule
\end{tabular}
\caption{LMArena, $1{,}000$ replications of $21{,}685$ sessions; all entries $\times10^{-3}$. RMSE is against each replication's full-annotation estimate; interval width is the mean model-based $95\%$ width.}\label{tab:arena}
\end{table}

\end{document}